\documentclass[a4paper,onecolumn,10pt,unpublished]{quantumarticle}
\pdfoutput=1
\usepackage[utf8]{inputenc}
\usepackage[english]{babel}
\usepackage[T1]{fontenc}
\usepackage{amsmath}
\usepackage{amssymb}
\usepackage{hyperref}
\usepackage{enumitem}
\usepackage[numbers,sort&compress]{natbib}

\usepackage{lineno}

\usepackage{tabularx}
\usepackage[dvipsnames]{xcolor}
\usepackage{tikz}
\usepackage{lipsum}
\usepackage{subcaption}

\usetikzlibrary{positioning}
\usetikzlibrary{shapes.geometric, arrows}
\usetikzlibrary{matrix}

\usepackage{amsthm}
\usepackage{arydshln}

\let\mod\relax
\DeclareMathOperator{\mod}{mod}
\DeclareMathOperator{\RREF}{RREF}
\DeclareMathOperator{\res}{Res}
\DeclareMathOperator{\LCM}{LCM}
\DeclareMathOperator{\GCD}{GCD}
\DeclareMathOperator{\Span}{Span}
\DeclareMathOperator{\wt}{wt}

\DeclareMathOperator{\How}{How}
\DeclareMathOperator{\Ker}{Ker}
\DeclareMathOperator{\dist}{dist}
\DeclareMathOperator{\ZSupp}{ZSupp}

\newtheoremstyle{break}%
    {}{}%
    {}{}%
    {\bfseries}{}
    {\newline}{}

\theoremstyle{break}
\newtheorem{proposition}{Proposition}[section]
\newtheorem{lemma}[proposition]{Lemma}
\newtheorem{corollary}[proposition]{Corollary}
\newtheorem{example}{Example}[section]

\theoremstyle{remark}

\begin{document}
\author{Mark A. Webster}
\email{mark.webster@sydney.edu.au}
\affiliation{Centre for Engineered Quantum Systems, School of Physics,
University of Sydney, Sydney, NSW 2006, Australia}
\affiliation{Sydney Quantum Academy, Sydney, NSW, Australia}

\author{Benjamin J. Brown}
\affiliation{Centre for Engineered Quantum Systems, School of Physics,
University of Sydney, Sydney, NSW 2006, Australia}
\affiliation{Niels Bohr International Academy, Niels Bohr Institute, Blegdamsvej 17, University of Copenhagen, 2100 Copenhagen, Denmark}

\author{Stephen D. Bartlett}
\affiliation{Centre for Engineered Quantum Systems, School of Physics,
University of Sydney, Sydney, NSW 2006, Australia}
\title{The XP Stabiliser Formalism: a Generalisation of the Pauli Stabiliser Formalism with Arbitrary Phases}
\maketitle
\begin{abstract}
We propose an extension to the Pauli stabiliser formalism that includes fractional $2\pi/N$ rotations around the $Z$ axis for some integer $N$. The resulting generalised stabiliser formalism -- denoted the \emph{XP stabiliser formalism} -- allows for a wider range of states and codespaces to be represented.  We describe the states which arise in the formalism, and demonstrate an equivalence between XP stabiliser states and ‘weighted hypergraph states’ - a generalisation of both hypergraph and weighted graph states. Given an arbitrary set of XP operators, we present algorithms for determining the codespace and logical operators for an XP code. Finally, we consider whether measurements of XP operators on XP codes can be classically simulated.
\end{abstract}

\section{Introduction}

Representing general quantum states of $n$ qubits requires an amount of information that is exponential in $n$.  For tractable theoretical study of quantum systems, we require more compact representations of quantum states of interest. Some examples of such representations include tensor network states e.g. Matrix Product States and Projected Entangled Pair States~\cite{tensor} and states created by low-depth quantum circuits~\cite{deutsch}.

The Pauli stabiliser formalism allows for the efficient description and manipulation of an important subset of quantum states, known as Pauli stabiliser states~\cite{gottesman}. Pauli stabiliser states and codes can be efficiently described using a set of stabiliser generators, which are elements of the Pauli group $\langle iI, X, Z\rangle^{\otimes n}$, with only $n$ such generators needed to describe a Pauli stabiliser state.  Many key quantum features such as entanglement and superposition can be captured using stabiliser states. Given the stabiliser generators, efficient algorithms exist for determining the codespace, logical operators and simulating measurements of Pauli operators. On the other hand, the states which can be represented within the Pauli stabiliser formalism are quite limited.  The fact that we can simulate the behaviour of Pauli stabiliser states  suggests that we are missing some important aspects of quantum advantage.

Given the ubiquity and the utility of the Pauli stabiliser formalism in quantum information theory, it is no surprise that there have been several studies into generalising this formalism in order to broaden the classes of stabiliser states. One such approach is the \emph{qudit} stabiliser formalism, where we fix a dimension $D$ and stabiliser generators are from the generalised Pauli group on $D$-level systems. The algorithms from the Pauli stabiliser formalism can be extended to qudit codes~\cite{qudit_codes}, and measurement of any generalised Pauli group operator is known to be classically simulable~\cite{qudit-gottesman-knill}. 

Another generalisation is the \emph{XS stabiliser formalism}~\cite{xs}, wherein stabiliser codes are defined using elements of $\langle \alpha I, X, S\rangle^{\otimes n}$ which act on qubits and where $\alpha := e^{i\pi/4}$ and $S := \text{diag}(1, i)$ is the phase gate. In Ref.~\cite{xs}, the authors demonstrate that states outside the Pauli stabiliser formalism can be represented as XS stabiliser states. In particular, they demonstrate that Quantum Twisted Double models which harbour non-Abelian anyons can be represented in the XS formalism. The authors present algorithms for determining the codespace, expectation values and logical operators for such codes, but these are limited to certain `regular' codes.

In this paper, we introduce the XP stabiliser formalism which generalises the concept of the XS formalism. In the XP formalism, we fix an integer $N$ and define XP codes using elements of $\langle \omega I, X, P\rangle^{\otimes n}$ where $\omega := e^{i\pi/N}, P := \text{diag}(1, \omega^2)$. We prove that hypergraph and weighted graph states can be represented as XP stabiliser states. These are useful classes of states which cannot be represented as Pauli stabiliser states. We present XP versions of many of the algorithms available in the Pauli stabiliser formalism, which apply to any XP code. This includes algorithms to determine the codespace, logical operators and simulate the measurement of diagonal Pauli operators on an XP codespace. Computational complexity of these tasks increases with the precision $N$ of the code. We  demonstrate that measurement of XP operators on a codespace is not always efficiently simulable. Hence, the XP formalism lies on the boundary between what is classically simulable and what is not, suggesting that we may be able to capture some aspects of quantum advantage.

The XP stabiliser formalism has a number of potential applications. For instance, the logical operator structure of XP codes is much richer than is the case for Pauli stabiliser codes. We could potentially use the XP formalism to describe codes with transversal logical non-Clifford operations (for instance T or CCZ gates) that could be used for fault-tolerant preparation of magic states. As the stabiliser generators of XP codes do not commute, certain no-go theorems may not apply, for instance, in the area of quantum memories~\cite{finitetemp}. Compared to the Pauli stabiliser formalism, a wider range of topological phases can be represented, which makes studying them more straightforward using the techniques presented in this paper.

This paper is structured to present the framework and tools necessary to start exploration of the XP formalism, with examples used throughout the paper. Where possible, examples are linked to \href{https://mybinder.org/v2/gh/m-webster/XPFpackage/HEAD?urlpath=tree/Examples}{interactive Jupyter notebooks} that allow for further exploration and modification of these examples by the reader.

In Chapter~\ref{chap:xpformalism}, we review the Pauli stabiliser formalism and the XS Formalism introduced in Ref.~\cite{xs}. We then introduce the XP formalism and summarise the main results of this paper.  In Chapter~\ref{chap:xpalgebra}, we set out definitions for XP operators and give a full description of their algebra. In Chapter~\ref{chap:codewords}, we show how to identify the codespace of an arbitrary set of XP operators and calculate a set of codewords which form a basis of the codespace.  In Chapter~\ref{chap:hypergraph}, we classify the states which arise in the XP formalism and show that two important classes of states can be represented as XP stabiliser states:  hypergraph~\cite{hypergraph} and weighted graph states~\cite{weightedgraph1,weightedgraphuniversal}. 

In Chapter~\ref{chap:LO}, we address the logical operator structure of XP codes. We show how to find XP operators which generate the logical operator group for a given codespace. We can allocate quantum numbers to the codewords of Chapter~\ref{chap:codewords}, and this leads to a classification of XP codes into XP-regular and non-XP-regular codes. We show that any XP-regular code can be mapped to a CSS code with similar logical operator structure. We show how to determine all possible logical actions which can be applied by diagonal operators and demonstrate that non-XP-regular codes give rise to logical operators with complex actions.

In Chapter~\ref{chap:Measurements}, we look at measurements in the XP formalism. We demonstrate an efficient algorithm for simulation of measurement of diagonal Pauli operators on any XP code. Beyond this special case, XP operators cannot in general be measured within the XP formalism. Estimating measurement outcome probabilities for XP operators is in general an NP-complete problem.

We conclude in Chapter~\ref{chap:openquestions} with a discussion and list of outstanding questions for the XP formalism.

\section{The XP Stabiliser Formalism}\label{chap:xpformalism}
The Pauli stabiliser formalism (PSF) has been hugely valuable for describing the quantum error-correcting codes we need to overcome the decohering effects of the environment on quantum systems. Moreover, its structure is such that we can prove powerful results on the simulability of Clifford circuits. 

Nevertheless, this formalism is limited in the types of quantum states it can describe. Extending the stabiliser formalism gives us new tools to describe more general quantum systems, and to explore their potential applications.

In this section, we first review the PSF and discuss the desirable features that we would like to carry over into an extension to the formalism. We next discuss the XS Formalism introduced in Ref.~\cite{xs}, which is an extension of the PSF. We then introduce our new formalism, the XP formalism, which generalises the concept of the XS Formalism. Finally, we outline the main properties of the XP formalism and summarise the results presented in this paper.

\subsection{Review of the Pauli Stabiliser Formalism}
Pauli stabiliser codes are amongst the most commonly studied quantum error correction codes. A stabiliser code is specified by a list of \textbf{stabiliser generators}. The stabiliser generators are tensor products of operators in the single-qubit Pauli group  $\mathcal{P} = \langle iI, X, Z\rangle$.  The \textbf{stabiliser group} is the group generated by the stabiliser generators.  The \textbf{codespace} is the simultaneous $+1$ eigenspace of the stabiliser group. The \textbf{codewords} are a basis of the codespace (i.e. spanning and independent).

The power of the Pauli stabiliser formalism is that stabiliser codes can be analysed and simulated  by operations on the \textbf{stabiliser generators} rather the \textbf{codewords}. In Chapter~10 of Ref.~\cite{nielsen_chuang_2010}, efficient algorithms involving operations on the stabiliser generators are set out for the following tasks:
\begin{enumerate}
\item \textbf{Check if a set of Pauli operators identify a non-trivial codespace}:  equivalent to checking whether $-I$ is in the stabiliser group.
\item \textbf{Find codewords} and \textbf{logical operators}: represent the stabiliser generators as binary vectors which form the rows of a \textbf{check matrix}, then apply linear algebra techniques.
\item \textbf{Simulate action of Clifford Unitary Operator $U$}: the updated code generators are obtained by conjugating each of the  generators by $U$.
\item \textbf{Simulate measurement of Pauli Operator $A$}: update the stabiliser generators by determining if they commute or anticommute with $A$.
\end{enumerate}
The above properties allow us to detect and correct errors in Pauli stabiliser codes by measuring the stabiliser generators.
\subsection{Extending the Pauli Stabiliser Formalism - Existing Work}

The Pauli stabiliser formalism (PSF) is capable of describing only a limited subset of all possible quantum states.  One way to extend the PSF is to consider qudit stabiliser codes using the Generalised Pauli Group~\cite{qudit_codes}, and there is a fairly substantial literature on this generalisation.

Another way is to work with qubits, but to admit operators with finer rotations around the Z axis in stabiliser groups. In the XS stabiliser formalism, introduced in Ref.~\cite{xs}, stabiliser generators are tensor products of operators in $\langle \alpha I, X, S\rangle$ where  $S := \text{diag}(1,i)$ so that $S^4 = I$ and  $\alpha := \exp(i\pi/4)$ so that $\alpha^8 = 1$.  The XS stabiliser formalism can describe a wider range of states compared to the PSF, meaning there are XS code states that cannot be expressed within the PSF.  Examples of such XS code states are the twisted quantum double models~\cite{tqd}.

A set of XS stabiliser generators do not need to commute to form a valid code.  As a result, we cannot in general perform simultaneous measurements of the generators, as is commonly done in PSF codes for error correction. Even so, commuting parent Hamiltonians for XS codes exist, and so error correction is possible with XS codes.

In the XS stabiliser formalism, \textbf{regular codes} are a special case, defined as those where the diagonal generators are tensor products of $\langle -I, Z\rangle$. For these regular codes, Ref.~\cite{xs} presents algorithms to: calculate the codewords from a list of stabiliser generators; calculate the stabiliser generators for a given codeword; calculate logical $Z$ and $X$ operators for a codespace; construct a circuit to find expectation values for measurements of Pauli operators on the code.

In summary, Ref.~\cite{xs} provides a useful generalisation of the PSF, as well as generalising a number of  algorithms for analysing and simulating codes. One of the main limitations of the XS stabiliser formalism is that many algorithms only work for \textbf{regular codes}, which are a subset of possible XS codes.

\subsection{The XP Stabiliser Formalism}

In our work, we introduce a formalism that generalises the the XS stabiliser formalism concept, allowing us to represent an even wider set of states. We demonstrate algorithms for most of the operations we have under the PSF, with a corresponding increase in computational complexity. The algorithms work on any XP code and are not restricted to regular codes as in Ref.~\cite{xs}. The XP formalism is at a similar level of generality as that presented for qudits in Ref.~\cite{qudit_codes}, and we prove many analogous results. The properties of XP and qudit stabiliser codes are compared in Table~\ref{tab:qudit}.

\begin{table}[t]
\centering
\begin{tabularx}{0.9\textwidth} { 
  | >{\raggedright\arraybackslash}X 
  || >{\raggedright\arraybackslash}X 
  | >{\raggedright\arraybackslash}X | }
\hline
 & \textbf{Qudit Codes Ref.~\cite{qudit_codes}} & \textbf{XP Codes}\\
\hline
\hline
\textbf{Parameters} & Qudit Dimension $D$, Number of qudits $n$ & Precision $N$, Number of qubits $n$\\
\hline
\textbf{Phases} & $\omega = \exp(2\pi i/D)$ & $\omega = \exp(\pi i/N)$\\
\hline
\textbf{Generalised X} & $X = \sum_{0 \le j \le D-1}|j\rangle\langle j+1|$ & Pauli $X$\\
\hline
\textbf{Generalised Z} & $Z = \sum_{0 \le j \le D-1}\omega^{j}|j\rangle\langle j|$ & $P = \text{diag}(1,\omega^2)$\\
\hline
\textbf{Vector form of operators} & $\omega^p X^{\mathbf{x}}Z^\mathbf{z}$ where $p \in \mathbb{Z}_D, \mathbf{x}, \mathbf{z} \in \mathbb{Z}_D^n$ & $\omega^p X^{\mathbf{x}}P^\mathbf{z}$ where $p \in \mathbb{Z}_{2N}, \mathbf{x}\in \mathbb{Z}_2^n, \mathbf{z} \in \mathbb{Z}_N^n$\\
\hline
\textbf{Commutators} & Operators commute, up to phase & Operators commute, up to a diagonal operator\\
\hline
\textbf{Clifford (Normaliser) Group} & Known - Table III of Ref.~\cite{qudit_codes}& Unknown - most likely restricted to tensors of single-qubit XP operators and Controlled Z operators\\
\hline
\textbf{Generalised Hadamard} & Quantum Fourier Transform & None\\
\hline
\textbf{Code Stabilisers $\langle G \rangle$ Commute} & Yes & No\\
\hline
\textbf{Code Stabilisers Uniquely defined by Codespace} & Yes & No - but Logical Identity Group is Unique (Section~\ref{sec:code_space_test})\\
\hline
\textbf{Standard form of Stabiliser Groups} & Yes - Using Smith Normal Form & Yes - Using Howell Matrix Form\\
\hline
\textbf{Codespace Dimension} & $\dim(\mathcal{C}) = D^n/|\langle G\rangle|$ & Arbitrary\\
\hline
\textbf{Classical Simulation of Measurements} & Yes - any generalised Pauli Ref.~\cite{qudit-gottesman-knill} & Not possible for arbitrary XP operators\\
\hline
\end{tabularx}
\caption{Qudit Stabiliser Codes Compared to XP Codes}\label{tab:qudit}
\end{table}

An XP stabiliser code is defined by fixing an integer $N \ge 2$ which we refer to as the \textbf{precision} of the code. We then specify a set of stabiliser generators which are from $\langle \omega I, X, P\rangle^{\otimes n}$ where
\begin{align}
    \omega &:= \exp\Big(\frac{\pi i}{N}\Big) = \exp\Big(\frac{1}{2N} 2 \pi i\Big) \text{ so } \omega^{2N} = 1\label{eq:omega_def}\\
    P &:= \text{diag}(1,\omega^2) \text{ so } P^N = I \label{eq:P_def}\,.
\end{align}

Each choice of precision $N$ leads to a different stabiliser formalism.  For example, $N=2$ corresponds to the standard Pauli stabiliser formalism, with $\omega = i$ and $Z^2 = I$.  The XS stabiliser formalism of Ref.~\cite{xs} corresponds to $N=4$, with $\omega = \sqrt{i}$ and $S^4 = I$. Note that $N$ does not need to be a power of 2 - e.g., $N=6$ or $N=7$ are valid choices.

Unlike the Pauli stabiliser formalism, the XPF is not closed under conjugation by Hadamard operators. In particular, if $N = 2M$ then:
\begin{align}
HPH^{-1} = \sqrt[M]{X}
\end{align}
One could consider expanding the formalism to allow fractional $X$ operators of this type, but this does not lead to a finite set of operators that is closed under group operations.

There are a number of open questions from Ref.~\cite{xs} which we also address for both XS codes and the general XP case:
\begin{itemize}
\item How do we find all transversal \textbf{logical operators} for a given code?
\item Which sets of operators stabilise the \textbf{same codespace}?
\item Can we \textbf{simulate measurements} efficiently?
\item Which classes of states can be described within a generalisation of the stabiliser formalism? 

\end{itemize}

\begin{table}[t]
\centering
\begin{tabularx}{0.9\textwidth} { 
  | >{\raggedright\arraybackslash}X 
  || >{\raggedright\arraybackslash}X 
  | >{\raggedright\arraybackslash}X | }
\hline
 & \textbf{XS Formalism}~\cite{xs} & \textbf{XP Formalism} \\
\hline
\hline
\textbf{Form of Stabiliser Generators} &
\(\langle \alpha I,X, S\rangle^{\otimes n}: S^4 = I,\alpha^8 = 1\) &
\(\langle \omega I, X, P\rangle^{\otimes n}: P^N = I,\omega^{2N} = 1\) \\
\hline
\textbf{Determine Codewords from Stabiliser Generators} & Regular codes only & All XP codes \\
\hline
\textbf{Determine Stabiliser Generators from Codewords} & States only (1D codespaces)  & Any codespace \\
\hline
\textbf{Which stabiliser groups have the same codespace?} & Open question & Same
logical identity group \(\Leftrightarrow\) same codespace \\
\hline
\textbf{Determine Logical Operators for a Code} & Regular codes only: $\overline{Z}$ and $\overline{X}$ & All XP codes: Generators for XP logical operator group, which may include non-Clifford logical operators \\
\hline
\textbf{Simulate Measurements on a Code} & Regular codes: circuit method to calculate expectation value when measuring Paulis & All XP codes: Stabiliser algorithm to measure diagonal Paulis;  Codeword algorithm to calculate outcome probabilities when measuring any XP operator.\\
\hline
\end{tabularx}
\caption{Summary of Results and Comparison with XS Formalism}\label{tab:summary}
\end{table}

\subsection{Summary of Results}

We present algorithms within the XPF for many of the operations that are possible within the Pauli stabiliser formalism (PSF). Our formalism is broader than the XS formalism and answers a number of open questions from Ref.~\cite{xs}. Table~\ref{tab:summary} compares the XS stabiliser formalism with the XPF and summaries the results we demonstrate in this paper.

We start by setting out the algebra of XP operators in Chapter~\ref{chap:xpalgebra}. We show how to represent XP operators as vectors of integers. By generalising the symplectic product of the PSF, we write elegant closed form expressions for the main algebraic operations on XP operators - including multiplication, inverses, powers, conjugation and commutation. We also show efficient ways to calculate the eigenvalues and the action of projectors of XP operators.

We specify XP codes by giving a list of XP operators which have a non-trivial simultaneous $+1$ eigenspace. In Chapter~\ref{chap:codewords}, we show how to identify the codespace for an arbitrary set of XP operators. In Ref.~\cite{xs}, a method is presented for doing this, but only for `regular' codes where the diagonal stabiliser generators are from $\langle -I, Z\rangle^{\otimes n}$. We demonstrate an algorithm which works for any XP code, whether regular or not. 

Having introduced some of the basic techniques for working with XP codes, in Chapter~\ref{chap:hypergraph} we determine which states arise under the XP formalism. We identify the form of the phase function of XP stabiliser states and show an equivalence between `weighted hypergraph states' and XP stabiliser states. In particular, two important classes of states - hypergraph and weighted graph states - can be represented as XP stabiliser states. We give examples which have potential uses in measurement-based quantum computation and self-correcting quantum memories.

Understanding which logical operators arise for a given XP code has important implications for which error-protected operations are possible. In Chapter~\ref{chap:LO}  we show how to calculate generators for the entire group of logical operators of XP form. The algorithm works for all XP codes, and yields any non-Clifford logical operators (e.g., logical T, CCZ etc) of XP form which act on the codespace. 

The stabiliser group of a codespace is not unique in the XPF, but the group of XP operators which act as logical identities on the codespace is unique. We show how to efficiently calculate generators for the logical identity group for a given XP group, resulting in a test for whether two XP groups stabilise the same codespace.

In Section~\ref{sec:classification} we introduce a classification of XP codes into XP-regular and non-XP-regular codes. XP-regular codes include all `regular' codes as defined in Ref.~\cite{xs}, but is a broader class. We show that, as in the PSF, the codespace dimension of XP-regular codes is a power of 2. For non-XP-regular codes, the codespace dimension is arbitrary. Non-XP-regular codes are non-additive and have a structure similar to that of codeword stabilised (CWS) quantum codes~\cite{cws}. Each XP-regular code can be mapped to a CSS code which has a similar logical operator structure. We demonstrate that for non-XP-regular codes, more complex diagonal operators can arise compared to the PSF. A summary of the differences between XP-regular and non XP-regular codes is in Table~\ref{tab:regular}.

\begin{table}[hbt!]
\centering
\begin{tabularx}{0.9\textwidth} { 
  | >{\raggedright\arraybackslash}X 
  || >{\raggedright\arraybackslash}X 
  | >{\raggedright\arraybackslash}X | }
\hline
\textbf{Property} & \textbf{XP-Regular} & \textbf{Non-XP-Regular} \\
\hline
\hline
\textbf{Codespace Dimension} & \(2^{k}\) for some \(k\) & Arbitrary \\
\hline
\textbf{Additive} & Yes & No \\
\hline
\textbf{Related CSS Code} & Maps to a CSS code with similar logical operator structure &
No related CSS code \\
\hline
\textbf{Diagonal Logicals} & Same as related CSS code & Exotic \\
\hline
\textbf{Non-diagonal Logicals} & Similar to related CSS code &
Transversal logical X may not exist \\
\hline
\end{tabularx}
\caption{XP-Regular vs Non XP-Regular Codes}\label{tab:regular}
\end{table}

Determining the extent to which computations on a quantum computer can be classically simulated is one of the central questions in the field of quantum information. In the Pauli  stabiliser formalism, we can classically simulate the measurement of any Pauli operator on any stabiliser code. A similar result holds for the qudit stabiliser formalism for generalised Pauli group operators. Chapter~\ref{chap:Measurements} covers measurement in the XP formalism. Measurement of diagonal Pauli operators can be efficiently simulated on any XP code, and we present an efficient stabiliser method for calculating the outcome probabilities and updated XP code. We then consider extending the algorithm to precision 4 operators. We show that finding the outcome probabilities when measuring collections of diagonal precision 4 diagonal operators is NP-complete. We also give examples where the measurement of precision 4 operators cannot be done within the XP formalism.

Finally in Chapter~\ref{chap:openquestions}, we summarise what is known about the XP formalism, discuss implications, and lay out possible future research directions. 

\subsection{XPF Software Package}
We have produced a Python software package implementing all algorithms discussed in this paper. The \href{https://github.com/m-webster/XPFpackage}{Github repository} is made available subject to \href{https://www.gnu.org/licenses/gpl-3.0.en.html}{GPL licensing}. Interactive  \href{https://mybinder.org/v2/gh/m-webster/XPFpackage/HEAD?urlpath=tree/Examples}{Jupyter notebooks}  for all examples in this paper are included and can be modified to run different scenarios.

\section{Algebra of XP Operators}\label{chap:xpalgebra}

In this section, we lay out fundamental results for the algebra of XP operators, including closed form expressions for algebraic operations which support efficient simulation on a classical computer. We show how to write a unique vector representation for XP operators. We generalise the symplectic product of the Pauli stabiliser formalism and show how to express algebraic operations in terms of the vector representation and the generalised symplectic product. We also introduce the concepts of the degree and fundamental phase of XP operators. These concepts allow us to determine the eigenvalues and actions of the projectors of XP operators.

\subsection{Vector Representation of XP Operators}\label{sec:vector_representation}

In this section, we show that there is a natural identification of XP operators on $n$ qubits with vectors of integers. Let $\mathbf{u} = (p|\mathbf{x}|\mathbf{z}) \in \mathbb{Z} \times \mathbb{Z}^{n} \times \mathbb{Z}^{n}$ be an integer vector of length $2n + 1$. Define the XP operator of precision $N$ corresponding to $\mathbf{u}$ as:
\begin{align}
    XP_N(\mathbf{u}) &:= \omega^p \bigotimes_{0\le i < n} X^{\mathbf{x}[i]} P^{\mathbf{z}[i]}
\end{align}
where $\omega$ and $P$ are as defined in Eqs.~\eqref{eq:omega_def} and~\eqref{eq:P_def}. Each component is periodic in that we can write:
\begin{align}
XP_N(p|\mathbf{x}|\mathbf{z}) &= XP_N(p \mod 2N|\mathbf{x}\mod 2|\mathbf{z}\mod N)
\end{align}
Accordingly, we can write a \textbf{unique vector representation} $(p|\mathbf{x}|\mathbf{z}) \in \mathbb{Z}_{2N} \times \mathbb{Z}_{2}^n \times \mathbb{Z}_{N}^n $ for each XP operator. We call $p$ the \textbf{phase component}, $\mathbf{x}$ the \textbf{X component} and $\mathbf{z}$ the \textbf{Z component}.

Here we list some properties of this notation:
\begin{enumerate}
\item The identity XP operator is $XP_N(0|\mathbf{0}|\mathbf{0})$ where $\mathbf{0}$ is the length $n$ vector with all entries $0$.
\item Because $\omega^N = -1$, we have $XP_N(N|\mathbf{0}|\mathbf{0}) = -I$.
\item The single qubit X operator is $XP_N(0|1|0)$.
\item Diagonal operators are of form $XP_N(p|\mathbf{0}|\mathbf{z})$, i.e. the $X$ component is the zero vector.
\item If $N$ is even, the single qubit $Z$ operator is $XP_N(0|0|\frac{N}{2})$.  For $N$ odd, $Z$ operators cannot be represented as XP operators.  Note that one may rescale the code to be of precision $2N$ by doubling the phase and Z components of stabiliser generators, in which case the rescaled code has the same codespace and accommodates $Z$ operators.
\end{enumerate}


\begin{example}[Using XP operator notation]
Consider the following example of an XP operator:
\begin{align}
    A = XP_8(12|1110000|0040000)\,.
\end{align}
The \textbf{precision} is specified by the subscript $8$, in this case $N=8$. This means that $\omega = \exp(\frac{1}{16}2\pi i)$ and $P^8 = I$ so $P = T$ where $T$ is the operator $\text{diag}(1, \sqrt{i})$.  In other words, this is an XT operator. Most of the examples we consider in this paper are precision $N=8$ codes. 

The components of $A$ are as follows. The phase component is a value $p \in \mathbb{Z}_{16}$. In this case $p = 12$ means the overall phase of the operator is $\omega^{12} = \exp(\frac{12}{16} 2\pi i ) = \exp(\frac{3}{4} 2\pi i ) = -i$. The X component is a binary vector of length $n$ so that $\mathbf{x} \in \mathbb{Z}_2^n$. In this case, $\mathbf{x} = 1110000$, representing $X_1 X_2 X_3$, where $X_i$ represents the operator which applies $X$ to the $i$th qubit and $I$ elsewhere. The Z component is a value $\mathbf{z} \in \mathbb{Z}_8^n$. In this case $\mathbf{z} = 0040000$ representing $T_3^4$.

In terms of  $X$ and $T$ operators, we can write: 
  \begin{align}
  XP_8(12|1110000|0040000) &= \exp(\frac{12}{16} 2\pi i ) X_1  X_2 X_3 T_3^4 \\
  &= -i X_1  X_2 X_3 Z_3
\end{align}
As the phase and Z components are divisible by 4, we can \textbf{rescale} $A$ and write it as a precision 2 operator by dividing the phase and Z components by 4:
\begin{align}
XP_8(12|1110000|0040000) &= XP_2(3|1110000|0010000)
\end{align}
\end{example}

\subsection{Multiplication Rule and Generalised Symplectic Product}\label{sec:MUL}
In the Pauli stabiliser formalism, we represent operators as binary vectors and understand commutation relations in terms of the symplectic product. In this section, we generalise the symplectic product to the XP formalism, and this allows us to write a simple rule for multiplying XP operators.

Let $\mathbf{z} \in \mathbb{Z}^n$ be an integer vector of length $n$ with $i$th  component denoted $\mathbf{z}[i]$. The \textbf{antisymmetric operator} of precision $N$ corresponding to $\mathbf{z}$ is: 
\begin{align}
D_N(\mathbf{z}) &:= XP_N(\sum_i \mathbf{z}[i] |\mathbf{0}|{-}\mathbf{z})\,.
\end{align}
In this paper, arithmetic operations on vectors are \textbf{component-wise} in $\mathbb{Z}$ i.e.:
\begin{align}
    (\mathbf{a} + \mathbf{b})[i] &:= \mathbf{a}[i] + \mathbf{b}[i]\\
    (\mathbf{a}\mathbf{b})[i] &:= \mathbf{a}[i]\mathbf{b}[i]\,.
\end{align}
We can then express the multiplication of two XP operators as the sum of their vector representations, adjusted by an antisymmetric operator:
\begin{proposition}[Multiplication of XP Operators]
The product of two XP operators given in vector format is:
\begin{align}XP_N(\mathbf{u_1}) XP_N(\mathbf{u_2}) = XP_N(\mathbf{u_1} + \mathbf{u_2})D_N(2\mathbf{x_2} \mathbf{z_1})\end{align}
\end{proposition}

\begin{proof}
For $n=1$, looking at matrix representations of $X$ and $P$ we see that:
\begin{align}PX = \omega^2 X P^{-1} = (X P) (\omega^2 P^{-2}) = (XP) D_N(2)\end{align}
From this, we can show that the multiplication rule applies for all single qubit XP operators in $\langle\omega I, X, P \rangle$ and in turn to the tensor products of such operators.
\end{proof}

 \begin{example}[Multiplication of XP Operators]
Let $N = 4, n=3$ so that unique vector representations of XP operators are $XP_N(p|\mathbf{x}|\mathbf{z})$ where  $(p|\mathbf{x}|\mathbf{z}) \in \mathbb{Z}_{8} \times \mathbb{Z}_2^3\times \mathbb{Z}_4^3$.  Consider two example XP operators $A_1$ and $A_2$ defined as
\begin{align*}
A_1 &= XP_4( 2 | 1 1 1 |3 3 0) \,, \quad \text{with}\ \mathbf{u}_1 = ( 2 | 1 1 1 |3 3 0)\,,\\
A_2 &= XP_4(6 | 0 1 0 |0 2 0) \,, \quad \text{with}\ \mathbf{u}_2 = ( 6 | 0 1 0 |0 2 0)\,.
\end{align*}
Then
\begin{equation}
    A_1 A_2 = XP_4(\mathbf{u}_1 + \mathbf{u}_2)D_N(2\mathbf{x}_2\mathbf{z}_1) =XP_4(6 | 1 0 1 |3 3 0)\,.
\end{equation}
This example is worked out in detail in the \href{https://github.com/m-webster/XPFpackage/blob/main/Examples/3.2_multiplication.ipynb}{linked Jupyter notebook}. You can also explore how multiplication works for random XP operators of arbitrary precision and length.
\end{example}

\subsection{Other Algebraic Identities}\label{sec:algebraic_identities}
We can write simple closed form identities for various algebraic operations in terms of antisymmetric operators, and these are summarised in Table~\ref{tab:algebraic_identities}. These identities allow us to efficiently implement algebraic operations in the \href{https://github.com/m-webster/XPFpackage}{XPF software package}.

\begin{table}[hbt!]
\centering
\begin{tabular}{|l||l|}
\hline
\textbf{Name} &\textbf{Rule} \\
\hline
\hline
\textbf{MUL}  & 
Multiplication of two XP operators \\
&
$XP_N(\mathbf{u_1}) XP_N(\mathbf{u_2}) = XP_N(\mathbf{u_1} + \mathbf{u_2})D_N(2\mathbf{x_2} \mathbf{z_1})$ \\
\hline
\textbf{SQ}  & 
Square of an XP operator\\
&  
$A^2 =  XP_N(2p|0|2\mathbf{z})D_N(2\mathbf{x}\mathbf{z})$\\
\hline
\textbf{POW}  & 
XP operator raised to a power \\
&  
$A^m = XP_N(mp|a\mathbf{x}|m\mathbf{z})D_N((m-a)\mathbf{x}\mathbf{z}), \text{ where }a = m \mod 2$\\
\hline
\textbf{INV}  & 
Inverse of an XP operator  \\
&
$A^{-1}= XP_N(-p|\mathbf{x}|-\mathbf{z})D_N(-2\mathbf{x}\mathbf{z})$\\
\hline
\textbf{CONJ}  & 
Conjugation of XP operators  \\
&
$A_1 A_2 A_1^{-1} = A_2 D_N(2\mathbf{x_1}\mathbf{z_2} + 2\mathbf{x_2}\mathbf{z_1}-4\mathbf{x_1}\mathbf{x_2}\mathbf{z_1})$\\
\hline
\textbf{COMM}  & 
Commutator of XP operators  \\
&
$A_1 A_2 A_1^{-1}A_2^{-1} = D_N(2\mathbf{x_1}\mathbf{z_2} - 2\mathbf{x_2}\mathbf{z_1}+4\mathbf{x_1}\mathbf{x_2}\mathbf{z_1}-4\mathbf{x_1}\mathbf{x_2}\mathbf{z_2})$\\
\hline
\textbf{OP}  & 
Action of an XP operator on a computational basis vector \\
&
$XP_N(p|\mathbf{x}|\mathbf{z}) |\mathbf{e}\rangle = \omega^{p+2\mathbf{e} \cdot \mathbf{z}}|\mathbf{e} \oplus \mathbf{x}\rangle$ \\
\hline
\end{tabular}

\caption{Algebraic Identities for XP Operators}\label{tab:algebraic_identities}
\end{table}

\begin{example}[Algebraic Identities]\label{eg:identities}
The following are some consequences of algebraic identities in Table~\ref{tab:algebraic_identities}:
\begin{enumerate}
\item The MUL and INV rules imply that products and inverses of diagonal operators are diagonal.
\item The SQ and COMM rules imply that squares and commutators of XP operators $A,B$ are always diagonal.
\item The CONJ rule implies that conjugating an operator $A$ by an operator $B$ results in $A$ times a diagonal operator.
\end{enumerate}
\end{example}

\subsection{Group Structure of XP Operators}\label{sec:xp_group_structure}

Because we have rules for products and inverses of XP operators, the XP operators of precision $N$ on $n$ qubits form a group, denoted $\mathcal{XP}_{N,n}$.

For any set of XP operators $\mathbf{G}$, we can determine the group generated by the operators which we denote $\mathcal{G} = \langle \mathbf{G}\rangle$. The subset of diagonal XP operators $\mathcal{G}_Z$ forms an \textbf{Abelian subgroup}. This is because  diagonal operators commute and $\mathcal{G}_Z$ is closed under multiplication (and so includes all inverses and the identity operator). 

There is a natural \textbf{group homomorphism} $\text{Zp}$ between $\mathcal{G}_Z$ over multiplication and $\mathbb{Z}_{2N}^{n+1}$ over addition modulo $2N$. For diagonal operators, the action of $\text{Zp}$  is:
\begin{align}
    \text{Zp}(XP_N(p|\mathbf{0}|\mathbf{z})) = (2\mathbf{z}|p) \mod 2N
    \end{align}
The $\text{Zp}$ map is well defined and is a group homomorphism because:
\begin{align}
    A_1 A_2 &= XP_N((p_1 + p_2)\mod 2N |\mathbf{0}|(\mathbf{z}_1+\mathbf{z}_2)\mod N)
    \end{align}
and so:
\begin{align}
    \text{Zp}(A_1 A_2) &=
    (2(\mathbf{z}_1+\mathbf{z}_2) \mod 2N |(p_1 + p_2) \mod 2N)\\
    &= (\text{Zp}(A_1) + \text{Zp}(A_2)) \mod 2N
\end{align}
The inverse map $\text{Zp}^{-1}(\mathbf{z}|p) = XP_N(p|\mathbf{0}|\mathbf{z}/2)$ is well defined providing each component of $\mathbf{z}$ is divisible by $2$ in $\mathbb{Z}_{2N}$ (even Z components). Addition over $\mathbb{Z}_{2N}^{n+1}$ takes vectors with even Z components to vectors with even Z components. Hence, we can find the generators of $\mathcal{G}_Z$ by finding a set of vectors $B \subset \mathbb{Z}_{2N}^{n+1}$ which span $\text{Zp}(\mathcal{G}_Z)$.  Using the Howell matrix form of Appendix~\ref{app:linalg}, we set $B = \text{How}_{\mathbb{Z}_N}(\text{Zp}(\mathcal{G}))$. The set $\text{Zp}^{-1}(B)$ generates $\mathcal{G}_Z$. This method is used to determine a unique set of canonical generators for an XP group (see Section~\ref{sec:canonical_generators}).

\subsection{Eigenvalues and Projectors of XP Operators}\label{sec:degree+fundamental_phase}
Identifying the eigenvalues and eigenvectors of XP operators will be important when considering measurements.  We first show how to determine the action of an XP operator on computational basis elements. The degree and fundamental phase of an XP operator, defined in this section, allow us to determine the eigenvalues of XP operators efficiently. These results are used  in the chapters on identifying the codespace and measurements in the XP formalism (Chapters~\ref{chap:codewords},~\ref{chap:Measurements}).

The action of an XP operator on a computational basis element $|\mathbf{e}\rangle$ of $\mathcal{H}_2^n$ where $\mathbf{e} \in \mathbb{Z}_2^n$ is:
\begin{align}
XP_N(p|\mathbf{x}|\mathbf{z}) |\mathbf{e}\rangle = XP_N(p|\mathbf{x}|\mathbf{z}) XP_N(0|\mathbf{e}|\mathbf{0}) |\mathbf{0}\rangle = \omega^{p + 2\mathbf{e} \cdot \mathbf{z}}|\mathbf{e} \oplus \mathbf{x}\rangle\label{eq:action_on_basis_elts}\end{align}
When calculating the action on computational basis elements, we apply the diagonal part of the operator first, then the X component. The notation $\mathbf{e} \cdot \mathbf{z} = \sum_i \mathbf{e}[i] \mathbf{z}[i]$ is the usual dot product for vectors in $\mathbb{Z}$. The notation $\mathbf{e} \oplus \mathbf{x}=(\mathbf{e} + \mathbf{x})\mod 2$ denotes component-wise addition modulo 2 which is equivalent to XOR for binary vectors $\mathbf{e}$ and $\mathbf{x}$.

For a given XP operator $A, A^{2N} = I$ so there must be a minimal $\deg(A) \in \mathbb{Z}_{2N}$ such that for some $q \in \mathbb{Z}_{2N}$:
\begin{align}
A^{\deg(A)} = \omega^qI
\end{align}
We call $\deg(A)$ the \textbf{degree} of $A$ and $q$ the \textbf{fundamental phase} of $A$. The degree can be calculated efficiently via the method below:

\begin{proposition}[Calculating Degree of Operator]
We calculate the degree of XP operator $A = XP_N(p|\mathbf{x}|\mathbf{z})$ as follows:
\begin{enumerate}
    \item If $A$ is diagonal: $\deg(A) = \LCM\{N/\GCD(N,\mathbf{z}[i]): 0 \le i < n\}$
    \item If $A$ is non-diagonal: $\deg(A) = 2 \deg(A^2)$, noting that $A^2$ is diagonal.
\end{enumerate}
\end{proposition}

\begin{proof}

To show 1, note that where $A$ is diagonal, $A^m = XP_N(mp|\mathbf{0}|m\mathbf{z})$. We need to solve for $m\mathbf{z} = \mathbf{0} \mod N$. To show 2, note that odd powers of $A$ are non-diagonal, so the degree must be even. Apply 1 to $A^2$ which is diagonal. 
\end{proof}

Once we have the degree of an operator, the fundamental phase is the phase component of $A^{\deg(A)}$. Determining the fundamental phase and degree of an XP operator allows us to identify its eigenvalues, as follows:

\begin{proposition}[Eigenvalues of Operator]
If $A$ has degree $d$ and fundamental phase $q$, the only possible eigenvalues of $A$ are $\omega^m : m = (q + 2Nj)/d$ for $j \in [0\ldots d-1]$.
\end{proposition}
\begin{proof}
Let $|\psi\rangle$ be an eigenvector with $A|\psi\rangle = \omega^p|\psi\rangle$. By the definition of degree and fundamental phase, $A^d = \omega^qI$ so $A^d|\psi\rangle = \omega^{dp}|\psi\rangle = \omega^q|\psi\rangle$. Hence $dp = q \mod 2N$ and the result follows.
\end{proof}

The following proposition allows us to calculate the action of projectors of XP operators on a computational basis element:

\begin{proposition}[Action of XP Projectors on computational basis elements]\label{prop:xp_projectors}
Consider the projectors $A_\lambda$ of $A$ onto the $\lambda$-eigenspace of $A$. If $A$ is diagonal, the action of $A_\lambda$ on the basis element $|\mathbf{e}\rangle$ is:
\begin{align}
    A_\lambda |\mathbf{e}\rangle &= \begin{cases} |\mathbf{e}\rangle:\text{ if }A |\mathbf{e}\rangle = \lambda |\mathbf{e}\rangle\\ 0: \text{ if }A |\mathbf{e}\rangle \neq \lambda |\mathbf{e}\rangle\end{cases}\label{eq:proj_diag}
    \end{align}
If $A$ is non-diagonal, the action of $A_\lambda$ is: 
\begin{align}
    A_\lambda |\mathbf{e}\rangle &= \begin{cases}  \frac{1}{2}(I + \lambda^{-1}A) |\mathbf{e}\rangle:\text{ if }A^2 |\mathbf{e}\rangle = \lambda^2|\mathbf{e}\rangle\\0: \text{ if }A^2 |\mathbf{e}\rangle \neq \lambda^2 |\mathbf{e}\rangle\end{cases}\label{eq:proj_nondiag}
\end{align}
\end{proposition}

\begin{proof}
To verify Eq.~\eqref{eq:proj_nondiag}, note that where $A^2|\mathbf{e}\rangle = \lambda^2 |\mathbf{e}\rangle$:
\begin{equation}
    A\Big(\frac{1}{2}(I + \lambda^{-1}A)|\mathbf{e}\rangle\Big)
    =\frac{1}{2}(A + \lambda^{-1}A^2)|\mathbf{e}\rangle =\frac{1}{2}(A + \lambda^{-1}\lambda^2 I)|\mathbf{e}\rangle =\lambda\Big(\frac{1}{2}(I + \lambda^{-1}A)|\mathbf{e}\rangle\Big)\,.
\end{equation}
\end{proof}

\section{Calculating Codewords from Stabiliser Generators}\label{chap:codewords}

In this section, we show how to identify the codespace stabilised by a given set of XP  operators. In the Pauli stabiliser formalism, there is a very simple relationship between the number of stabiliser generators and the dimension of the codespace the stabiliser group defines.  Given a stabiliser group on $n$ qubits with $k$ independent commuting generators, the codespace has dimension $2^{(n-k)}$.

In the XP formalism, this relationship is much more complex.  For example, the eigenspace dimensions of XP operators vary widely and are not in general powers of 2. As an illustration of this complexity, the $+1$ eigenspace dimensions which arise for various 7-qubit diagonal XP operators of precision 8 are listed in Table~\ref{tab:codespace_dim}.  Readers can explore eigenspaces of diagonal XP operators in the \href{https://github.com/m-webster/XPFpackage/blob/main/Examples/4.1_eigenspaces.ipynb}{linked Jupyter notebook}. 

\begin{table}[hbt!]
    \centering
$$
\begin{array}{|l r||l r||l r|}
\hline
\textbf{Operator}&\textbf{Dim}&\textbf{Operator}&\textbf{Dim}&\textbf{Operator}&\textbf{Dim}\\
\hline
\hline
XP_8( 0|\mathbf{0}|3333333)&  1			&XP_8( 0|\mathbf{0}|1333355)& 15			 &XP_8( 0|\mathbf{0}|2222266)& 28 \\
XP_8( 0|\mathbf{0}|2555555)&  2			 &XP_8( 0|\mathbf{0}|6133555)& 16			 &XP_8( 0|\mathbf{0}|6111177)& 30 \\
XP_8( 0|\mathbf{0}|0133333)&  4			 &XP_8( 0|\mathbf{0}|1173335)& 17			 &XP_8( 0|\mathbf{0}|4222666)& 32 \\
XP_8( 0|\mathbf{0}|2355555)&  6			 &XP_8( 0|\mathbf{0}|6111735)& 18			 &XP_8( 0|\mathbf{0}|3333555)& 35 \\
XP_8( 0|\mathbf{0}|3333335)&  7			 &XP_8( 0|\mathbf{0}|1173355)& 19			 &XP_8( 0|\mathbf{0}|2222666)& 36 \\
XP_8( 0|\mathbf{0}|2223555)&  8			 &XP_8( 0|\mathbf{0}|6135555)& 20			 &XP_8( 0|\mathbf{0}|0333555)& 40 \\
XP_8( 0|\mathbf{0}|6133335)& 10			 &XP_8( 0|\mathbf{0}|3333355)& 21			 &XP_8( 0|\mathbf{0}|0003355)& 48 \\
XP_8( 0|\mathbf{0}|6133355)& 12			 &XP_8( 0|\mathbf{0}|6155555)& 22			 &XP_8( 0|\mathbf{0}|4444444)& 64 \\
XP_8( 0|\mathbf{0}|1733333)& 13			 &XP_8( 0|\mathbf{0}|2661117)& 24			 &XP_8( 0|\mathbf{0}|0000000)&128 \\
XP_8( 0|\mathbf{0}|6113555)& 14			 &XP_8( 0|\mathbf{0}|6111117)& 26		& &	  \\
\hline
\end{array}
$$

    \caption{Example:  Eigenspace dimensions for selected diagonal XP operators with $n=7$, $N = 8$.}
    \label{tab:codespace_dim}
\end{table}

Here we present an algorithm to identify the codespace of a set of $XP$ operators.  The input for our algorithm is an arbitrary list of XP operators $\mathbf{G} \subset \mathcal{XP}_{N,n}$. The output is a list of codewords $\{|\kappa_i\rangle : 0 \le i < \dim(\mathcal{C})\}$ that form a basis for the codespace $\mathcal{C}$ stabilised by the group $\mathcal{G} = \langle \mathbf{G}\rangle$, or an empty set if there is no codespace.  The algorithm operates in two steps:
\begin{enumerate}
\item Convert the set of XP operators into a \textbf{canonical form}. This is a set of generators $\mathbf{S}$ in a particular form  which  generate the stabiliser group $\langle\mathbf{S}\rangle = \langle\mathbf{G}\rangle$. We split the canonical generators into diagonal $\mathbf{S}_Z$ and non-diagonal generators $\mathbf{S}_X$, where the diagonal canonical generators $\mathbf{S}_Z$ generate the diagonal subgroup of the stabiliser group.
\item Calculate an independent set of codewords $|\kappa_i\rangle$ that span the codespace.  We do this by applying the orbit operator (defined below in terms of the $\mathbf{S}_X$) to particular computational basis elements $|\mathbf{m}_i\rangle$ in the simultaneous $+1$ eigenspace of the $\mathbf{S}_Z$.
\end{enumerate}

We will describe these steps in detail in the following sections.

\subsection{Canonical Generators of XP Groups}\label{sec:canonical_generators}

For any set of XP operators $\mathbf{G}$, we can calculate a set of operators in canonical form that generate the same XP group as $\mathbf{G}$. Specifically, we calculate an independent set of diagonal ($\mathbf{S}_Z$) and non-diagonal ($\mathbf{S}_X$) operators that generate $\langle\mathbf{G}\rangle$. The diagonal subgroup of $\langle \mathbf{G}\rangle$ is generated by $\mathbf{S}_Z$.

Proposition~\ref{prop:canonical_generators} sets out the form and properties of the canonical generators. The proposition uses the concept of a generator product which is defined as follows. The \textbf{generator product} of an ordered set of XP operators $\mathbf{S} = \{S_0,\dots,S_{m-1}\}$ specified by a vector of  integers $\mathbf{a} \in \mathbb{Z}^m$ is:
\begin{align}
    \mathbf{S}^\mathbf{a} = \prod_{0 \le i < m }S_i^{\mathbf{a}[i]}\label{eq:generator_product}
\end{align}

\begin{proposition}[Canonical Generators of an XP Group]\label{prop:canonical_generators}

For any set of XP Operators $\mathbf{G} = \{G_i: 0 \le i < m \}$, there exists a unique set of diagonal operators $\mathbf{S}_Z := \{B_j : 0 \le j < s\}$ and non-diagonal operators $\mathbf{S}_X := \{A_i : 0 \le i < r\}$ with the following form:
\begin{enumerate}
\item Let $S_X$ be the $r \times n$ binary matrix formed from the X-components of the $\mathbf{S}_X$. The matrix $S_X$ is in Reduced Row Echelon Form (RREF).
\item Let $S_{Zp}$ be the $s \times (n+1)$ matrix with rows taken from the image of  $\mathbf{S}_Z$ under the Zp map of Section~\ref{sec:xp_group_structure} (i.e. $\text{Zp}(XP_N(p|\mathbf{0}|\mathbf{z})) = (2\mathbf{z}|p)$). The matrix $S_{Zp}$ is in Howell Form (see Appendix~\ref{app:linalg}).
\item For $XP_N(p|\mathbf{x}|\mathbf{z}) \in \mathbf{S}_X$, the matrix $\begin{pmatrix}1&(2\mathbf{z}|p)\\\mathbf{0}&S_{Zp}\end{pmatrix}$ is in Howell  Form.
\end{enumerate}

The following properties hold for the canonical generators:

\paragraph{Property 1:} All group elements $G \in \langle\mathbf{G}\rangle$ can be expressed as $G = \mathbf{S}_X^\mathbf{a} \mathbf{S}_Z^\mathbf{b}$ where $\mathbf{a} \in \mathbb{Z}_2^{|\mathbf{S}_X|}$, $\mathbf{b}  \in \mathbb{Z}_N^{|\mathbf{S}_Z|}$
\paragraph{Property 2:} Two sets of XP operators of precision $N$ generate the same group if and only if they have the same canonical generators.

\end{proposition}

Conditions 1-3 ensure the canonical generators are unique for a given set of operators, and so can form the basis of the test in Property 2. The entries of $S_{Zp}$ are from $\mathbb{Z}_{2N}$ which is in general a ring. The Howell matrix form is a generalisation of the RREF which gives us a canonical basis for the row span of a matrix over a ring. See Appendix~\ref{app:linalg} for a full description of the Howell matrix form and linear algebra over rings. 

We briefly consider the implications of Properties 1 and 2.  Let $\langle \mathbf{G}\rangle$ be an XP stabiliser group. If there is an operator of form $\omega^qI, q \ne 0$ in $\langle \mathbf{G}\rangle$, the codespace is empty. Due to the Howell Property (see Section~\ref{sec:howell_matrix_form}), we can determine if this is the case by checking if $\omega^qI \in \mathbf{S}_Z$ for some $q \ne 0$. This is a generalisation of the requirement that $-I \notin \langle\mathbf{G}\rangle$ in the Pauli stabiliser formalism and the concept of admissible generating sets in the XS stabiliser formalism (see Ref.~\cite{xs} on page 7). Going forward, we assume that XP codes are specified in terms of their canonical generators $\mathbf{S}_Z, \mathbf{S}_X$ and that there is no element $\omega^qI, q \ne 0$ in $\mathbf{S}_Z$.

Because the matrices $S_X, S_{Zp}$ are in echelon form, this imposes a natural ordering on $\mathbf{S}_X, \mathbf{S}_Z$. Property 1 states that we can write any $G \in \langle \mathbf{G} \rangle$ as a product of the canonical generators where operators are applied in this order. It implies that $\mathbf{S}_Z$ generates the diagonal subgroup of $\langle \mathbf{G}\rangle$  because the diagonal subgroup is the set of operators where $\mathbf{a} = \mathbf{0}$. Recalling Example~\ref{eg:identities}, all commutators and squares of elements in $\langle \mathbf{G}\rangle$ are diagonal and so are in $\langle\mathbf{S}_Z\rangle$. Applying the results in Section~3 of Ref.~\cite{howell}, we can also determine the size of the group $\langle \mathbf{G}\rangle$ once we have the canonical generator form.

In Appendix~\ref{sec:canonical_generator_proof}, we demonstrate an algorithm for calculating the canonical generators and prove Proposition~\ref{prop:canonical_generators}.

\subsection{Finding a Basis of the Codespace}\label{sec:find_code_space}
In this section, we will show how to find a basis of the codespace stabilised by the  canonical generators $\mathbf{S}_Z , \mathbf{S}_X$ of Section~\ref{sec:canonical_generators}. The result is a set of independent codewords that span the codespace. 

The codespace is the intersection of the simultaneous $+1$ eigenspace of the diagonal generators and the $+1$ eigenspace of the non-diagonal generators: \begin{align}
\mathcal{C} = \mathcal{C}_Z \cap \mathcal{C}_X 
\end{align}

The diagonal generators determine the \textbf{Z-support} of the codewords. We define the \textbf{Z-support} of a state $|\psi\rangle$  in $\mathcal{H}_2^n$ as the set of length $n$ binary vectors $\mathbf{e}$ such that the coefficient of the corresponding computational basis vector $|\mathbf{e}\rangle$ in $|\psi\rangle$ is non-zero. That is:
\begin{align}
    \ZSupp(|\psi\rangle) &= \{\mathbf{e} \in \mathbb{Z}_2^n: \langle \mathbf{e}|\psi\rangle \ne 0\}
\end{align}
Because all elements in $\mathbf{S}_Z$ are diagonal, we can write a basis of $\mathcal{C}_Z$ as a set of computational basis vectors:
\begin{align}
\mathcal{C}_Z &= \Span_{\mathbb{C}}\{|\mathbf{e}\rangle : \mathbf{e} \in \mathbb{Z}_2^n, B|\mathbf{e}\rangle = |\mathbf{e}\rangle, \forall B \in \mathbf{S}_Z\}\label{eq:C_Z}
\end{align}
Let $E := \ZSupp(\mathcal{C}_Z)$ be the binary vectors corresponding to the computational basis vectors in $\mathcal{C}_Z$. Any codeword expressed in terms of the computational basis must be a linear combination over $\mathbb{C}$ of $|\mathbf{e}\rangle, \mathbf{e}\in E$.

The non-diagonal generators determine the relative phases of the computational basis vectors in the codewords. The relative phase information is captured by the \textbf{orbit operator}. Let $\mathbf{S}_X$ be the non-diagonal canonical generators $\{A_i : 0 \le i < r\}$ ordered as in Section~\ref{sec:canonical_generators}. Using the generator product notation of Eq.~\eqref{eq:generator_product}, the orbit operator is defined as:
\begin{align}
O_{\mathbf{S}_X} := \sum_{\mathbf{v} \in \mathbb{Z}_2^r} \mathbf{S}_X^\mathbf{v} \label{eq:orbit_operator}
\end{align}
Where $\mathbf{e} \in E$, the image of $|\mathbf{e}\rangle$ under the orbit operator,  $O_{\mathbf{S}_X}|\mathbf{e}\rangle$,  is fixed by all elements of the stabiliser group $\langle \mathbf{S}_Z, \mathbf{S}_X\rangle$ (see Proposition~\ref{prop:cworbit}). In the next section, we demonstrate how to find an independent set of codewords of this form which span the codespace.

\subsubsection{Coset Structure of $E$ and Orbit Representatives}\label{sec:coset_structure_of_E}
In this section, we assume we are given $E$, the Z-support of $\mathcal{C}_Z$ as in Eq.~\eqref{eq:C_Z}, and show how to identify a subset $E_m$ of $E$ such that the image of $E_m$ under the orbit operator is a basis of the codespace (i.e. an independent spanning set).  The resulting basis is a set of un-normalised codewords $|\kappa_i\rangle$ such that:
\begin{align}
    |\kappa_i \rangle &:= O_{\mathbf{S}_X}|\mathbf{m}_i \rangle: \mathbf{m}_i \in E_m\label{eq:kappa}
\end{align}
The normalisation constant is $\frac{1}{\sqrt{2^r}}$ where $r$ is the number of non-diagonal canonical generators, and is omitted for clarity. Once we have $E_m$, we know the dimension of the codespace:
\begin{align}
    \dim(\mathcal{C}) &= |E_m|
\end{align}
We identify the subset $E_m$ by looking at the \textbf{coset structure} of $E$. First, we show that the Z-support of a codeword in the form of Eq.~\eqref{eq:kappa} can be viewed as a coset in the group $\mathbb{Z}_2^n$ under component-wise addition modulo 2. Let $S_X$ be the binary matrix formed from the X-components of the $\mathbf{S}_X$ which is in RREF by construction (See Proposition~\ref{prop:canonical_generators}). Then $\text{Span}_{\mathbf{Z}_2}(S_X) = \langle S_X\rangle$ is a subgroup of $\mathbb{Z}_2^n$ of size $2^r$ where $r=|S_X|$. The Z-support of $O_{\mathbf{S}_X}|\mathbf{e}\rangle$ can be identified with a \textbf{coset} in the group of binary vectors $\mathbb{Z}_2^n$:
\begin{align}
     \ZSupp(O_{\mathbf{S}_X}|\mathbf{e}\rangle) = \mathbf{e} + \langle S_X\rangle := \{(\mathbf{e} + \mathbf{u}S_X)\mod 2: \mathbf{u} \in \mathbb{Z}_2^r\}  
\end{align} 

Next, we introduce the \textbf{residue function} which tells us whether two vectors are in the same coset, and hence occur in the Z-support of the same codeword. Let $\mathbf{m} = \text{Res}_{\mathbb{Z}_2}(S_X,\mathbf{e})$ be defined as:
\begin{align}
\begin{pmatrix}1&\mathbf{m}\\\mathbf{0}&S_X\end{pmatrix} :=\text{RREF}_{\mathbb{Z}_2}\begin{pmatrix}1&\mathbf{e}\\\mathbf{0}&S_X\end{pmatrix} 
\end{align}
Two vectors $\mathbf{e}_1, \mathbf{e}_2 \in E$ are in the same coset if and only if $\text{Res}_{\mathbb{Z}_2}(S_X,\mathbf{e}_1) =  \text{Res}_{\mathbb{Z}_2}(S_X,\mathbf{e}_2)$. The residue of $\mathbf{e}$ is zero if and only if $\mathbf{e} \in \langle S_X\rangle$.

We use the residue function to identify a subset of $E$ of minimal size whose image under the orbit operator yields a basis of the codespace. The set of \textbf{orbit representatives} $E_m$ is defined as the image of $E$ under the residue function:
\begin{align}
E_m := \{ \text{Res}_{\mathbb{Z}_2}(S_X,\mathbf{e}): \mathbf{e}\in E\}
\end{align}
The cosets of $E_m$ partition $E$ (see Proposition~\ref{prop:partition}). Accordingly, the image of $E_m$ under the orbit operator is a basis of the codespace (see Proposition~\ref{prop:kappa_are_basis}). 

\subsubsection{Codewords Notation}\label{sec:code_words_notation}
The following notation for codewords is used throughout this paper. Let $S_X$ be the matrix formed from the X-components of the non-diagonal canonical generators $\mathbf{S}_X$. As $S_X$ is in echelon form, rows have a natural ordering and we interchangeably consider $S_X$ to be a set of binary vectors. Let $S_X$ have $r$ non-zero rows or $|S_X| = r$. The set of codewords generated from the orbit representatives is uniquely determined. When written in terms of the computational basis, we refer to the following as the \textbf{orbit form} of the codewords:
\begin{align}
|\kappa_i \rangle &= O_{\mathbf{S}_X}|\mathbf{m}_i \rangle := \sum_{0 \le j < 2^r} \omega^{p_{ij}} |\mathbf{e}_{ij}\rangle: \mathbf{m}_i \in E_m, \exists p_{ij} \in \mathbb{Z}_{2N}, \mathbf{e}_{ij} \in \mathbb{Z}_2^n \label{eq:orbitform}
\end{align}
The \textbf{Z-support of the codewords} is the same as the Z-support of $\mathcal{C}_Z$, the simultaneous $+1$ eigenspace of the $\mathbf{S}_Z$, and is denoted $E$:
\begin{align}
E &= \ZSupp(\mathcal{C}_Z) = \bigcup_i \ZSupp(|\kappa_i\rangle)
\end{align}
We can write a \textbf{coset decomposition} of $E$ in terms of $E_m$ as follows:
\begin{align}
E &= E_m + \langle S_X\rangle = \{(\mathbf{m}_i + \mathbf{u}S_X)\mod 2: \mathbf{m}_i \in E_m, \mathbf{u} \in \mathbb{Z}_2^r\}
\end{align}
There is a direct relationship with the Z-support of each codeword as follows:
\begin{align}
 E_i &:= \ZSupp(|\kappa_i\rangle) = \mathbf{m}_i + \langle S_X\rangle
\end{align}
In addition, there is a unique coset decomposition of $E_m$ so that $E_m = E_q + \langle L_X\rangle$ for sets of binary vectors $E_q$ and $L_X$ so that $E = E_q +  \langle L_X\rangle +  \langle S_X\rangle$. We demonstrate how to find this decomposition in Section~\ref{sec:L_X}. The full coset decomposition is useful for the following reasons:
\begin{itemize}
    \item The X-components of logical operators must be in $\langle L_X\rangle + \langle S_X\rangle$, and we can calculate a generating set of non-diagonal logical operators with X components in $L_X$ (Section~\ref{sec:L_X})
    \item The size of $E_q$ gives rise to a natural classification of XP codes (Section~\ref{sec:classification}).
    \item We assign quantum numbers to each codeword (Section~\ref{sec:quantum_numbers}) based on the coset decomposition of $E_m$ which then allows us to analyse the logical action of operators (Section~\ref{sec:LO_Action}).

\end{itemize}

\subsection{Calculating Orbit Representatives from the Canonical Generators}\label{sec:finding_Em}

We have demonstrated that the image of the orbit representatives under the orbit operator is a basis of the codespace, and that the dimension of the codespace is given by the number of orbit representatives. In previous sections we have assumed that we have been given $E$, the Z-support of the simultaneous $+1$ eigenspace of the $\mathbf{S}_Z$, as a starting point. In practice, we generally start with the stabiliser generators of a code, and calculating $E$ from them is an NP-complete problem for XP codes (see Ref.~\cite{xs} Section VII). 

In this section, we show how to calculate the orbit representatives from the canonical generators. Orbit representatives have a specific form which reduces the search space significantly compared to searching for the whole of $E$. We use a graph search algorithm to make finding the orbit representatives tractable for `reasonable' codes.

\subsubsection{Exhaustive Algorithm to find $E$}\label{sec:exhaustive_algorithm_E}
We first show how to find $E$ using an exhaustive inefficient algorithm.  Let  \linebreak $\mathbf{S}_Z = \{B_i = XP_N(p_i|\mathbf{0}|\mathbf{z}_i): 0 \le i < s\}$.  Calculating the $+1$ eigenspace of $\mathbf{S}_Z$ is equivalent to solving the following equation in binary vectors $\mathbf{e}$:
\begin{align}
B_i|\mathbf{e}\rangle = \omega^{p_i+2\mathbf{e}\cdot \mathbf{z}_i}|\mathbf{e}\rangle &= \omega^0|\mathbf{e}\rangle, \forall B_i\in \mathbf{S}_Z \label{eq:B_i|e>}
\end{align}
For this equation to have solutions in $\mathbf{e}$, $p_i$ must be divisible by $2$ in $\mathbb{Z}_{2N}$. Let $S_{Zp}$ be the matrix with rows $(\mathbf{z}_i|p_i/2)$. Let $(\mathbf{e}|1)$ represent the column vector $\mathbf{e}$ with an entry of 1 appended. Eq.~\eqref{eq:B_i|e>} can be written in matrix form as:
\begin{align}
  S_{Zp}^T (\mathbf{e}|1) \mod N &= \mathbf{0} \,.
\end{align}
Solutions are of form:
\begin{align}
  (\mathbf{e}|1) = \mathbf{a}K \mod N\,,\label{eq:exhaustive_search}
\end{align}
where $K$ is the Howell basis of $\Ker_{\mathbb{Z}_N}(S_{Zp})$ (see Section~\ref{sec:howell_matrix_form}). To find solutions for binary vectors $\mathbf{e}$, we seek $\mathbf{a} \in \mathbb{Z}_N^n$ such that $(\mathbf{e}|1) = \mathbf{a}K \mod N$ is a vector of zeros and ones.  Linear algebra techniques cannot be used to find the vectors $\mathbf{a}$. We would in principle need to search through all possible values of $\mathbf{a}$ to find valid solutions. Given that there are $N^{|K|}$ possible values of $\mathbf{a}$, this is of exponential complexity.

\subsubsection{Graph Search Algorithm for Orbit Representatives}\label{sec:graph_search}
Rather than searching through all possible values of $\mathbf{a}$ in Eq.~\eqref{eq:exhaustive_search} to find $E$, we employ a more efficient graph search algorithm which uses the special form of the orbit representatives to speed up the the search.


The special form of the orbit representatives is as follows. Let $\mathbf{x}_j$ be the $j$th row of $\mathbf{S}_X$ and let $l_j$ be the \textbf{leading index} of  $\mathbf{x}_j$ - i.e. $\mathbf{x}_j[l_j] = 1$ and $\mathbf{x}_j[k] = 0, \forall k < l_j$\label{sec:leading_index}. In Proposition~\ref{prop:orbit_rep_form}, we show that for any orbit representative $\mathbf{m}_i \in E_m, \mathbf{m}_i[l_j] = 0$ for all leading indices $l_j$, $0 \le j < |S_X|$. In each coset of $S_X$, $\mathbf{m}_i$ is the unique vector with this property. As $S_X$ is in RREF, there are exactly $r=|S_X|$ leading indices $l_j$ where $\mathbf{m}_i[l_j]$ is guaranteed to be zero. 

We modify the exhaustive search algorithm presented in Section~\ref{sec:exhaustive_algorithm_E} to take into account the special form of the orbit representatives:
\begin{enumerate}
\item Find the Howell basis $K$ of $\Ker_{\mathbb{Z}_N}(S_{Zp})$
\item Search for solutions $(\mathbf{e}|1) =  \mathbf{a}K \mod N$ where $\mathbf{e}$ is a binary vector
\item As we only need a orbit representative $\mathbf{m}_i$ for each codeword, we can restrict $\mathbf{e}[l] = 0$ where the $l$ are the indices of the leading entries in $S_X$. 
\end{enumerate}
The search is made more efficient by storing and re-using partial solutions (dynamic programming) and results in a graph object from which the solutions can be generated. The main advantage over the exhaustive algorithm is due to the reduction of the search space by a factor of $2^r$ by using the special form of the orbit representatives. Where all stabiliser generators are diagonal and $r = |S_X| = 0$, the advantage over the exhaustive algorithm is not as significant. However, for `reasonable' codes which have both diagonal and non-diagonal generators, and which encode a relatively low number of logical qubits, the graph search algorithm is efficient in practice.

It is possible that the search algorithm returns an empty set.  In this case, the simultaneous $+1$ eigenspace of the $\mathbf{S}_Z$ is empty and there is no codespace. In this case, the XP stabiliser group does not define a code.

\subsection{Summary of Codewords Algorithm}\label{sec:code_word_algorithm_summary}
In summary, the algorithm for identifying the codespace stabilised by an arbitrary set of XP operators $\mathbf{G}$ is:
\begin{enumerate}
    \item Calculate the canonical generators $\mathbf{S}_Z$ and $\mathbf{S}_X$ such that $\langle \mathbf{G}\rangle = \langle \mathbf{S}_Z, \mathbf{S}_X\rangle$ using the algorithm in Section~\ref{sec:canonical_generators}. If $\omega^qI \in \mathbf{S}_Z$ for $q \ne 0$, the codespace is empty.
    \item Find the orbit representatives $E_m = \{\mathbf{m}_i\}$ using the graph search algorithm in Section~\ref{sec:graph_search}. The dimension of the codespace is $\dim(\mathcal{C}) = |E_m|$. If $E_m = \emptyset$, the codespace is empty as the simultaneous $+1$ eigenspace of the $\mathbf{S}_Z$ is of dimension zero.
    \item A basis of the codespace is given by $\{|\kappa_i\rangle = O_{\mathbf{S}_X}|\mathbf{m}_i\rangle : \mathbf{m}_i \in E_m \}$ using the orbit operator of Eq.~\eqref{eq:orbit_operator}. 
\end{enumerate}

\subsection{Example: Calculating Codewords - Code 1}\label{eg:code1}
We illustrate our algorithm to find the codewords with an example. We will use this same example throughout this paper to illustrate various concepts. The detailed calculations for this example are set out in \href{https://github.com/m-webster/XPFpackage/blob/main/Examples/4.5_code_words.ipynb}{the linked Jupyter notebook}. We start with the following stabiliser generators of precision $N=8$ on $n=7$ qubits: 
\begin{align}
\mathbf{G} &= \begin{matrix}XP_8( 8|0000000|6554444)\\
XP_8( 7|1111111|1241234)\\
XP_8( 1|1110000|3134444)\end{matrix}
\end{align}
\paragraph{Step 1: Canonical Generators}
Using the algorithm in Appendix~\ref{sec:canonical_generator_algorithm_outline}, the canonical generators for this code are:
\begin{align}
\mathbf{S}_Z &= \begin{matrix}XP_8( 8|0000000|2334444)\\
XP_8( 0|0000000|0440000)\end{matrix} \label{eq:SZexample}\\
\mathbf{S}_X &= \begin{matrix}XP_8( 9|1110000|1240000)\\
XP_8(14|0001111|0001234)\end{matrix} \label{eq:SXexample}
\end{align}

Note that a single diagonal generator yields multiple diagonal canonical generators.  This behaviour is typical of XP groups.

\paragraph{Step 2: Orbit Representatives}
The graph search algorithm in Section~\ref{sec:graph_search} yields the following orbit representatives:
\begin{align}
E_m &= \begin{pmatrix}
0000001\\
0000010\\
0000100\\
0000111\\
\end{pmatrix}\label{eq:Em Code 1}
\end{align}

The dimension of the codespace is $\dim(\mathcal{C}) = |E_m| = 4$.

\paragraph{Step 3: Image of orbit representatives under orbit operator is a basis}
Finally, we form an independent set of codewords $\{\kappa_i\}$ by applying the orbit operator $O_{\mathbf{S}_X}$ of Eq.~\eqref{eq:orbit_operator} to the computational basis elements corresponding to the $E_m$:
\begin{align}
\begin{matrix}
| \kappa_0\rangle &= O_{\mathbf{S}_X} |0000001\rangle =|0000001\rangle&+\omega^{6}|0001110\rangle&+\omega^{9}|1110001\rangle&+\omega^{15}|1111110\rangle\\
| \kappa_1\rangle &= O_{\mathbf{S}_X} |0000010\rangle =|0000010\rangle&+\omega^{4}|0001101\rangle&+\omega^{9}|1110010\rangle&+\omega^{13}|1111101\rangle\\
| \kappa_2\rangle &= O_{\mathbf{S}_X}|0000100\rangle =|0000100\rangle&+\omega^{2}|0001011\rangle&+\omega^{9}|1110100\rangle&+\omega^{11}|1111011\rangle\\
| \kappa_3\rangle &= O_{\mathbf{S}_X} |0000111\rangle =|0000111\rangle&+|0001000\rangle&+\omega^{9}|1110111\rangle&+\omega^{9}|1111000\rangle\\
\end{matrix}
\end{align}

\subsection{Calculating Codewords - Discussion and Summary of Results}

Given a set of XP operators $\mathbf{G}$, we can determine a basis for the codespace stabilised by $\langle\mathbf{G}\rangle$. We first determine a set of canonical generators using linear algebra techniques over rings (Section~\ref{sec:canonical_generators}). The method uses the unique vector representation of XP operators of Section~\ref{sec:vector_representation} and can be done efficiently. This mirrors the result for generalised Pauli groups in Ref.~\cite{qudit_codes}.

Once we have the generators in canonical form, we find the orbit representatives $E_m$ using a graph search algorithm (Section~\ref{sec:graph_search}). The codespace dimension corresponds to the number of orbit representatives, and applying the orbit operator defined in Eq.~\eqref{eq:orbit_operator} to the orbit representatives results in a basis of the codespace. The graph search algorithm works for any XP code, but its efficiency depends on the precision of the code and the number of non-diagonal stabiliser generators. In the worst case, where we have only diagonal stabilisers, finding the orbit representatives reduces to an NP-complete problem.

\section{Classification of XP Stabiliser States}\label{chap:hypergraph}

Now that we have some familiarity with the XP stabiliser formalism, it is natural to ask which quantum states can be represented within the formalism. In this chapter, we demonstrate an equivalence between XP stabiliser states and `weighted hypergraph states' - a generalisation of both hypergraph ~\cite{hypergraph} and weighted graph states ~\cite{weightedgraph1}. 

In the Pauli stabiliser formalism, any stabiliser state can be mapped via local Clifford operators to a graph state~\cite{graphstates}. In the XS Formalism~\cite{xs}, the authors show that the phases of an XS stabiliser state are described by a phase function which is a polynomial of maximum degree 3. In this chapter, we generalise these results to the XP formalism. 

In Section~\ref{sec:whg_definitions}, we introduce definitions for weighted hypergraph states. In Section~\ref{sec:xp_phase_function}, we describe the  phases which are possible for XP stabiliser states. In Section~\ref{sec:xp2whg}, we show how to represent any XP stabiliser state as a weighted hypergraph state. Finally, in Section~\ref{sec:whg2xp} we show how to represent any weighted hypergraph state as an XP stabiliser state. In general, this requires us to embed the weighted hypergraph state into a larger Hilbert space.

\subsection{Weighted Hypergraph State Definitions}\label{sec:whg_definitions}

In this section, we introduce the concept of weighted hypergraph states - a class of states which includes graph, hypergraph and weighted graph states. A \textbf{generalised controlled phase operator} $CP(p/q, \mathbf{v})$ is specified by a rational number $p/q$ and a binary vector $\mathbf{v}$ of length $r$. The action of the operator on a computational basis state $|\mathbf{e}\rangle, \mathbf{e}\in \mathbb{Z}_2^r$ is:
\begin{align}
CP(p/q,\mathbf{v})|\mathbf{e}\rangle := \begin{cases}\exp(i2\pi p/q)|\mathbf{e}\rangle: \mathbf{e}\mathbf{v} = \mathbf{v}\\
|\mathbf{e}\rangle: \text{ Otherwise}\end{cases}\,.
\end{align}
Multiplication of vectors in the above equation is component wise. We construct a \textbf{weighted hypergraph state} by applying a series of generalised controlled phase operators to the $|+\rangle^{\otimes r}$ state. 

In the Pauli stabiliser formalism, all stabiliser states can be mapped to graph states by applying a set of local Clifford unitaries. A \textbf{graph state} on $r$ vertices is specified by a set of edges $E = \{(i,j): i < j \in [0\dots r-1]\}$. The graph state is formed by applying controlled $Z$ operators corresponding to the edges to $|+\rangle^{\otimes r}$ i.e. $|\phi\rangle = (\prod_{(i,j) \in E}CZ_{ij})|+\rangle^{\otimes r}$.  

We now show how graph states generalise to weighted hypergraph states. For a binary vector $\mathbf{v}$, the support of $\mathbf{v}$ defines an edge (i.e. $\text{supp}(\mathbf{v}) := \{i \in [0\dots r-1]: \mathbf{v}[i] = 1\}$). Graph states have edges composed of 2 vertices only so $\text{wt}(\mathbf{v}) = |\text{supp}(\mathbf{v})| = 2$. Generalised controlled phase operators can have edges involving between $1$ and $r$ vertices. The condition $\mathbf{e}\mathbf{v} = \mathbf{v}$ means that we apply the phase when $\text{supp}(\mathbf{v}) \subset \text{supp}(\mathbf{e})$.

For graph states, only phases of $\pm 1$ are possible as we apply controlled Z operators. Generalised controlled phase operators, on the other hand, can apply any phase of form $\exp(i2\pi p/q)$. Where $p/q=1/2$, the operator acts as a \textbf{generalised controlled $Z$ operator} because it applies a phase of $\exp(i\pi) = -1$ if $\mathbf{e}\mathbf{v} = \mathbf{v}$.

\subsection{Phase Functions of XP States}\label{sec:xp_phase_function}

In this section, we describe which relative phases are possible for XP stabiliser states. The \textbf{phase function} of an XP stabiliser state $|\phi\rangle$ of precision $N$ is an integer valued function $f$ on vectors $\mathbf{e}\in \ZSupp(|\phi\rangle)$ such that:
\begin{align}
    |\phi\rangle := \sum_{\mathbf{e}\in \ZSupp(|\phi\rangle)} \omega^{f(\mathbf{e})}|\mathbf{e}\rangle = \sum_{\mathbf{e}\in \ZSupp(|\phi\rangle)} \omega^{f(e_0\dots e_{n-1})}|\mathbf{e}\rangle\,.
\end{align}
We generally consider $f$ to be a function of the binary variables $e_i := \mathbf{e}[i], 0 \le i < n$. In this chapter, phase functions are defined by a vector $\mathbf{q} \in \mathbb{Z}^{2^n}$ and are polynomials of form:
\begin{align}
    f(e_0\dots e_{n-1}) = \sum_{s \subset [0\dots n-1]}\mathbf{q}[s]\prod_{j \in s}e_j\label{eq:phase_function_polynomial}
\end{align}
For phase functions of this form, we can identify each term of the polynomial with a generalised controlled phase operator. The term $\mathbf{q}[s]\prod_{j \in s}e_j$ corresponds to the controlled phase operator $CP(\mathbf{q}[s]/2N, \mathbf{v})$ where $\mathbf{v}[j] = 1$ if $j\in s$ or $0$ otherwise. It is known that for Pauli stabiliser states ($N = 2 = 2^1)$, the phase function is a polynomial of the form in Eq.~\eqref{eq:phase_function_polynomial} of degree at most $2$ in the variables $e_i$, whilst for XS codes ($N=4=2^2$), the maximum degree is $3$. Our aim is to generalise these results to XP codes.

We first show how to express the Z-support of any XP stabiliser state in terms of a set of binary variables $\{u_i\}$, which are a subset of the $\{e_i\}$ variables defined above. We will then express the form of the phase function of an XP stabiliser state in terms of the $\{u_i\}$. Due to the results of Section~\ref{sec:code_word_algorithm_summary}, any XP stabiliser state can be written in the following canonical form:
\begin{align}
    |\phi\rangle = O_{\mathbf{S}_X}|\mathbf{m}\rangle = \sum_{\mathbf{u} \in \mathbb{Z}_2^r}\mathbf{S}_X^\mathbf{u}|\mathbf{m}\rangle\label{eq:xp_canonical}
\end{align}
for non-diagonal canonical generators $\mathbf{S}_X$, $r = |\mathbf{S}_X|$ and where $\mathbf{m}$ is the single orbit representative. The orbit operator $O_{\mathbf{S}_X}$ is defined in Eq.~\eqref{eq:orbit_operator} and the generator product $\mathbf{S}_X^\mathbf{u}$ is defined in Eq.~\eqref{eq:generator_product}. Let $\mathbf{S}_X = \{XP_N(p_i|\mathbf{x}_i|\mathbf{z}_i): 0 \le i < r\}$. We can  write $|\phi\rangle$ in terms of the binary variables $u_i := \mathbf{u}[i], 0 \le i < r$ as follows:
\begin{align}
    |\phi\rangle = \sum_{u_i \in \mathbb{Z}_2}\prod_{0 \le i < r}XP_N(u_i p_i|u_i \mathbf{x}_i|u_i \mathbf{z}_i)|\mathbf{m}\rangle\,.\label{eq:phi_ui}
\end{align}
The sum in the above equation ranges over all possible values of $u_i\in \mathbb{Z}_2$, for $0 \le i < r$. The Z-support of $|\phi\rangle$ can be expressed in terms of the $\{u_i\}$ as follows:
\begin{align}
    \ZSupp(|\phi\rangle) = \{\mathbf{m} \oplus \bigoplus_{0 \le i < r}u_i \mathbf{x}_i: u_i \in \mathbb{Z}_2\} \,.\label{eq:zsupp_ui}
\end{align}
We now show that the binary variable $u_i$ can be identified with the value of a particular component $\mathbf{e}[l_i]$ of the vectors $\mathbf{e} \in \ZSupp(|\phi\rangle)$. Let $l_i$ be the leading index of $\mathbf{x}_i$ (see~\ref{sec:leading_index}). Because $\mathbf{S}_X$ is in canonical form and $\mathbf{m}$ is an orbit representative, we have $\mathbf{x}_j[l_i] = \delta_{ij}, \mathbf{m}[l_i] = 0$. Hence for $\mathbf{e} := \mathbf{m} \oplus \bigoplus_{0 \le i < r}u_i \mathbf{x}_i \in \ZSupp(|\phi\rangle), \mathbf{e}[l_i] = u_i$. 

In the following Proposition, we express the phase function for an XP stabiliser state in terms of the $\{u_i\}$ and describe its form:

\begin{proposition}[Phase Functions of XP States]\label{prop:xp_phase_function}
Let $|\phi\rangle = O_{\mathbf{S}_X}|\mathbf{m}\rangle = \sum_{\mathbf{u}\in \mathbb{Z}_2^r}\mathbf{S}_X^\mathbf{u}|\mathbf{m}\rangle$ be an XP stabiliser state in the canonical form of Eq.~\eqref{eq:xp_canonical} with $r:= |\mathbf{S}_X|$. Let $u_i, 0 \le i < r$ be binary variables such that $u_i := \mathbf{u}[i]$. Then:
\begin{enumerate}[label=(\alph*)]
\item The phase function is of the following form for some vector $\mathbf{q} \in \mathbb{Z}^{2^r}$ indexed by the subsets $s$ of $[0\dots r-1]$: 
\begin{align}
    f(u_0,u_1,\dots,u_{r-1}) = \sum_{s \subset [0\dots r-1]}\mathbf{q}[s] 2^{|s|-1}\prod_{j \in s}u_j. 
\end{align} 
\item For $N = 2^t$, the maximum degree of the phase function is $t+1$.
\end{enumerate}
\end{proposition}
Proof of Proposition~\ref{prop:xp_phase_function} is in Appendix~\ref{app:hypergraph}.

\subsection{Representation of XP States as Weighted Hypergraph States}\label{sec:xp2whg}

We now demonstrate a method for determining the phase function for a given XP state. This allows us to represent XP stabiliser states as weighted hypergraph states. 

\subsubsection{Algorithm: Weighted Hypergraph Representation of a Given XP State}
\paragraph{Input:} An XP state $|\phi\rangle =O_{\mathbf{S}_X}|\mathbf{m}\rangle$  of precision $N$. Let $\mathbf{S}_X = \{XP_N(p_i|\mathbf{x}_i|\mathbf{z}_i): 0 \le i < r\}$.

\paragraph{Output:} A set of generalised controlled phase operators $\{CP(p_i/2N,\mathbf{v}_i)\}$  such that $|\phi\rangle = \Big(\prod_i CP(p_i/2N,\mathbf{v}_i)\Big)\sum_{\mathbf{e}\in \ZSupp(|\phi\rangle)}|\mathbf{e}\rangle$.

\paragraph{Method}
\begin{enumerate}
\item Let $l_i$ be the leading index of $\mathbf{x}_i$ and define the $r \times n$ binary matrix $L$ by setting $L_{ij} = 1$ if $j = l_i$ and $0$ otherwise.
\item Let $\mathbf{p}$ be a vector whose entries are indexed by rows $\mathbf{u}\in \mathbb{Z}_2^r$ such that $\mathbf{p}[\mathbf{u}]$ is the phase component of $\mathbf{S}_X^{\mathbf{u}}|\mathbf{m}\rangle$
\item For $\mathbf{u}$ in $\mathbb{Z}_2^r$, ordered by weight then lexicographic order:
    \begin{enumerate}
    \item If $\mathbf{p}[\mathbf{u}] \ne 0$, add the operator $CP(\mathbf{p}[\mathbf{u}]/2N,\mathbf{u}L)$ to the list of operators
    \item For all $\mathbf{v} \in \mathbb{Z}_2^r$ such that $\mathbf{v}\mathbf{u} = \mathbf{u}$, set $\mathbf{p}[\mathbf{v}] = (\mathbf{p}[\mathbf{v}]  - \mathbf{p}[\mathbf{u}] ) \mod 2N$
    \end{enumerate}
\end{enumerate}

If the precision $N=2^t$ is a power of $2$, we only need to consider rows of weight at most $t+1$ due to  Proposition~\ref{prop:xp_phase_function}. By multiplying the $1 \times r$ vector $\mathbf{u}$ by the $r \times n$ matrix L in step 3(a), we create a $1 \times n$ vector $\mathbf{v}$ such that $\mathbf{v}[l_i] = u_i$.

\begin{example}[Weighted Hypergraph Representation of XP State - Union Jack State]
The following example illustrates the operation of the algorithm to determine the phase function of an XP stabiliser state and hence the weighted hypergraph representation. For precision $N=4$, let $|\phi\rangle =O_{\mathbf{S}_X}|\mathbf{m}\rangle$ with:
\begin{align}
\mathbf{m} &= \mathbf{0}\\
\mathbf{S}_X &= \begin{matrix}
XP_4(0|1000111000|0112000033)\\
XP_4(0|0100100110|1001003000)\\
XP_4(0|0010010101|1001003000)\\
XP_4(0|0001001011|2110330000)
\end{matrix}
\end{align}
Calculating the phase component of $\mathbf{S}_X^\mathbf{u}|\mathbf{m}\rangle$ for all values of $\mathbf{u} \in \mathbb{Z}_2^r$, we find that the phase is zero for all values of $\mathbf{u}$, apart from $\mathbf{u} = 1011$ and $1101$ where the phase is $-1$. For precision $N=4$, $\omega = \exp(i\pi/4)$ so $ \omega^4 = -1$. Hence, the phase function of $|\phi\rangle$ is $f(u_0,u_1,u_2,u_3) = 4u_0u_2u_3 + 4u_0u_1u_3$. The degree of $f$ is 3, which is the maximum degree for states of precision $N=4=2^2$. We can write $|\phi\rangle$ as a weighted hypergraph state by applying $CP(1/2,1011000000)$ and 
$CP(1/2,1101000000)$ to $\sum_{\mathbf{e}\in \ZSupp(|\phi\rangle)}|\mathbf{e}\rangle$.

The state $|\phi\rangle$ is the unit cell of the `Union Jack' state introduced in Ref.~\cite{unionjack}. This is a hypergraph state with 2-dimensional Symmetry Protected Topological Order (SPTO) and is a universal resource for quantum computation using only single qubit measurements in the X, Y, and Z basis - see Figure~\ref{fig:whg}. Detailed working for this example is available in the \href{https://github.com/m-webster/XPFpackage/blob/main/Examples/5.1_xp_to_whg.ipynb}{linked Jupyter notebook}.
\end{example}

\captionsetup[subfigure]{margin=5pt}
\begin{figure}[hbt!]
\centering
\begin{subfigure}[t]{.5\textwidth}
  \centering
\begin{tikzpicture}[scale=2]
  \draw[help lines] (-0.5,-0.5) grid (1.5,1.5);
    \filldraw[color=white, fill=red!30, very thick] (0,1)
  -- (0,0)
  -- (1,0)
  -- cycle;
      \filldraw[color=white, fill=blue!30, very thick] (0,1)
  -- (1,1)
  -- (1,0)
  -- cycle;
  \draw[above left] (0,1) node(s){$0$};
  \draw[above right] (1,1) node(t){$1$};
  \draw[below left] (0,0) node(s){$2$};
  \draw[below right] (1,0) node(t){$3$};
  \fill (0,0) circle (0.06cm) (0,1) circle (0.06cm) (1,0) circle (0.06cm) (1,1) circle (0.06cm);
\end{tikzpicture}
    \subcaption{Unit Cell of the Union Jack State of Ref.~\cite{unionjack} which is a hypergraph state. Qubits on the corners of each of the shaded triangles represent the edge size 3 operators $CP(1/2,1101)$ and $CP(1/2,1011)$ which act on $|+\rangle^{\otimes 4}$. 
    }
    \end{subfigure}%
\begin{subfigure}[t]{.5\textwidth}
  \centering
\begin{tikzpicture}[scale=2]
  \draw[help lines] (-0.5,-0.5) grid (1.5,1.5);
  \draw[above left] (0,1) node(s){$0$};
  \draw[above right] (1,1) node(t){$1$};
  \draw[below left] (0,0) node(s){$2$};
  \draw[below right] (1,0) node(t){$3$};
  \fill (0,0) circle (0.06cm) (0,1) circle (0.06cm) (1,0) circle (0.06cm) (1,1) circle (0.06cm);
  \draw[very thick] (0,1)--(1,1);
  \draw[very thick] (0,0)--(1,0);
  \draw[dashed] (0,1)--(1,0);
  \draw[dashed] (0,0)--(1,1);
\end{tikzpicture}
    \subcaption{Unit Cell of the weighted graph state in Ref.~\cite{weightedgraphuniversal}. Qubits connected by bold lines are acted on by controlled Z operators and those by dashed lines by controlled S operators. The weighted graph state is formed by the operator $CP(1/2,1100) CP(1/2,0011) CP(1/4,1001) CP(1/4,0110)$ acting on $|+\rangle^{\otimes 4}$ and can be represented as a precision 4 XP stabiliser state. }
    \end{subfigure}%
     \caption{Examples of Weighted Hypergraph States which can be be represented as XP stabiliser states.}
     \label{fig:whg}
\end{figure}
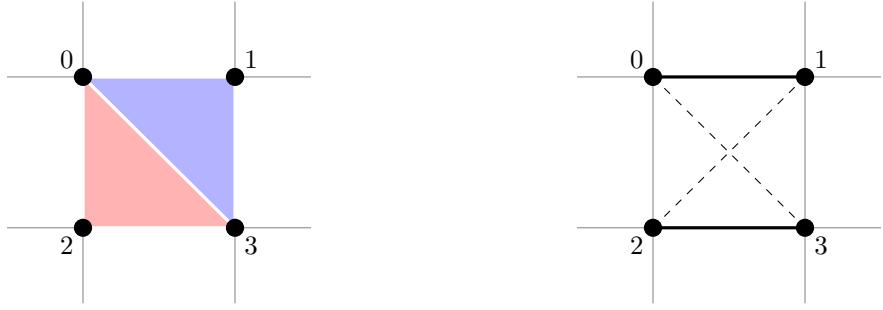

\subsection{Representation of Weighted Hypergraph States as XP Stabiliser States}\label{sec:whg2xp}

In this section, we show how to represent any weighted hypergraph state as an XP stabiliser state. In general, this involves embedding the state into a larger Hilbert space using an embedding operator. We first demonstrate this for a single generalised controlled phase operator i.e. $|\phi\rangle = CP(p/q,\mathbf{v})|+\rangle^{\otimes r}$. There are two possible cases depending on the weight of the vector $\mathbf{v}$.

\paragraph{Case 1 $\text{wt}(\mathbf{v}) = 1$: }
Let $\text{wt}(\mathbf{v}) = 1$ and let $i$ be the single non-zero component of $\mathbf{v}$. For a given precision $N$, we can identify $CP(1/2N, \mathbf{v})$ with the operator $\sqrt{P} = \text{diag}(1,\omega)$ acting on qubit $i$ because for a computational basis vector $|\mathbf{e}\rangle, \sqrt{P}_i|\mathbf{e}\rangle = \omega|\mathbf{e}\rangle$ if $\mathbf{e}[i] =1$ and $|\mathbf{e}\rangle$ otherwise.  Hence, to create a phase of $\exp(\frac{p}{q} 2\pi i)$, we can let $2N = q$. We also need $N$ to be an integer $\ge 2$ so we set $N$ as follows:
\begin{align}
    N = \begin{cases}q/2: \text{ if } q > 2 \text{ and } q\mod 2 = 0\\q: \text{ Otherwise}.\end{cases}
\end{align}

The state $|+\rangle^{\otimes r}$ is stabilised by $\{X_j: j \in [0\dots r-1]\}$ where $X_j$ is the Pauli X operator acting on qubit $j$. We can also write $X_j = XP_N(0|\mathbf{x}_j|\mathbf{0})$ where  $\mathbf{x}_j$ is the $j$th row of $I_r$.  The result of conjugating $X_j$ by $\sqrt{P}_i$ is $\omega X_jP_i$ if $i=j$ and $X_j$ otherwise. This can also be written $\sqrt{P}_iX_j\sqrt{P}_i^{-1} = X_jD_N(\mathbf{x}_j\mathbf{v})$ as is a generalisation of the identity $SXS^{-1} = X(iZ) = Y$ for Pauli operators.  Hence the state $|\phi\rangle$ is stabilised by $\mathbf{S}_X = \{XP_N(0|\mathbf{x}_j|\mathbf{0})D_N(\frac{p}{q}2N\mathbf{x}_j\mathbf{v})\}$.

\paragraph{Case 2 $\text{wt}(\mathbf{v}) \ge 2$: }
For $m := \text{wt}(\mathbf{v}) \ge 2$, we in general need to embed the weighted graph state $|\phi\rangle$ into a larger Hilbert space to represent it as an XP stabiliser state. The \textbf{embedding operator} is defined in terms of $M^r_m$, the binary matrix whose columns are the bit strings of length $r$ of weight between $1$ and $m$ inclusive. We order the columns of $M^r_m$ first by weight then by lexicographic order. The embedding operator $\mathcal{E}^r_m$ acts on computational basis vectors as follows: \begin{align}
    \mathcal{E}^r_m|\mathbf{e}\rangle = |\mathbf{e}M^r_m \mod 2\rangle\,.\label{eq:embedding_operator}
\end{align}
Our aim is to find a precision $N$ and a set of stabiliser generators $\mathbf{S}_X, \mathbf{S}_Z$ which stabilise the embedded state $|\psi\rangle := \mathcal{E}^r_m|\phi\rangle$.

We set the precision $N = q2^{m-2}$ - this is because phase function terms of degree $m$ include a factor of $2^{m-1}$ modulo $2N$ (see Proposition~\ref{prop:xp_phase_function}) and we need to allow for phases of form $\exp(2\pi i/q)$. If $N$ is odd, we multiply it by $2$ so that we can form the diagonal stabiliser generators (see Eq.~\eqref{eq:whg_sz} below).

The non-diagonal stabiliser generators $\mathbf{S}_X$ determine the phase function and are defined as follows. The X-components of the $\mathbf{S}_X$ are the rows $\mathbf{x}_j$ of $M^r_m$. The Z-components are obtained by multiplying the $\mathbf{x}_j$ by an `alternating vector' $\mathbf{a}$ and an `inclusion vector' $\mathbf{w}$. The vector $\mathbf{a}$ is indexed by the columns $\mathbf{u}$ of $M^r_m$ and is $1$ if the weight of $\mathbf{u}$ is even, and $-1$ otherwise:
\begin{align}
    \mathbf{a}[\mathbf{u}] := (-1)^{\text{wt}(\mathbf{u})}\,.\label{eq:alternating_vector}
\end{align}
The inclusion vector $\mathbf{w}$ with respect to $\mathbf{v}$ is:
\begin{align}
    \mathbf{w}[\mathbf{u}] := \begin{cases}1: \mathbf{u} \mathbf{v} = \mathbf{u}\\0: \text{ Otherwise}\end{cases}\label{eq:inclusion_vector}\,.
\end{align}
Multiplying the rows $\mathbf{x}_j$ of $M^r_m$ by the inclusion vector $\mathbf{w}$ ensures that we only consider columns of $M^r_m$ whose support is a subset of the support of $\mathbf{v}$. As a result of Proposition~\ref{prop:whg2xp}, the following operators generate the required phase function:
\begin{align}
    \mathbf{S}_X = \{XP_N(0|\mathbf{x}_j|\frac{p}{q}\frac{2N}{2^{m-1}}\mathbf{a}\mathbf{x}_j\mathbf{w}): 0 \le j < r\} \,.
\end{align}

We now show how to construct the diagonal stabiliser generators $\mathbf{S}_Z$. We calculate a basis of $\Ker_{\mathbb{Z}_2}(M^r_m)$ as follows. Because the columns of weight 1 occur first, $M^r_m$ is of form $M^r_m = \begin{pmatrix}I | A\end{pmatrix}$ for some binary matrix $A$. Hence the kernel of $M^r_m$ over $\mathbb{Z}_2$ is spanned by:
\begin{align}
    K^r_m := \begin{pmatrix}A^T | I\end{pmatrix}\,.\label{eq:ker_M^r_m}
\end{align}
It is easy to see that $K^r_m (M^r_m)^T \mod 2 = 0$. Let $\mathbf{z}_j$ be the $j$th row of $K^r_m$ so that $\mathbf{x}_i\cdot\mathbf{z}_j \mod 2 = 0$. The following operators commute with the elements of $\mathbf{S}_X$:
\begin{align}
    \mathbf{S}_Z := \{XP_2(0|\mathbf{0}|\mathbf{z}_j)\} =  \{XP_N(0|\mathbf{0}|N\mathbf{z}_j/2)\}\,.\label{eq:whg_sz}
\end{align}
This can be seen by using the COMM rule of Table~\ref{tab:algebraic_identities} and noting that $N\mathbf{x}_j\cdot  \mathbf{z}_j \mod 2N = 0$:
\begin{align}
[XP_N(0|\mathbf{x}_j|\mathbf{w}_j),XP_N(0|\mathbf{0}|N\mathbf{z}_j/2)] &= D_N(N\mathbf{x}_j\mathbf{z}_j) 
= XP_N(N\mathbf{x}_j\cdot  \mathbf{z}_j|\mathbf{0}|\mathbf{0}) = I\,.
\end{align}
We are now in a position to state the algorithm for weighted hypergraph states with multiple generalised controlled phase operators.

\subsubsection{Algorithm: Representation of Weighted Hypergraph States as XP Stabiliser States}
\paragraph{Input:} A weighted hypergraph state $|\phi\rangle = \big(\prod_i CP(p_i/q_i,\mathbf{v}_i)\big)|+\rangle^{\otimes r}$ with $p_i, q_i$ mutually prime and $\mathbf{v}_i$ of weight $m_i \ge 0$.

\paragraph{Output:} A precision $N$, an embedding operator $\mathcal{E}^r_m$ and stabiliser generators $\mathbf{S}_X, \mathbf{S}_Z$ of an XP code whose codespace is spanned by $|\psi\rangle := \mathcal{E}^r_m|\phi\rangle$, a state with the same phase function as $|\phi\rangle$.

\paragraph{Method:}
\begin{enumerate}
\item Let $m = \max(\{m_i\})$ - we use the embedding operator $\mathcal{E}^r_m$. Note that when $m = 1$, the embedding operator is trivial as $M^r_1 = I_r$.
\item We set the precision of the code as $N := \text{LCM}(2, \{N_i\})$ where we define the $N_i$ as follows:
    \begin{itemize}
        \item If $m_i \ge 2$: set $N_i = q_i 2^{m_1-2}$.
        \item If $m_i = 1$: if $q_i > 2$ and $q_i\mod 2 = 0$ set $N_i = q_i/2$; otherwise set $N_i = q_i$.   
    \end{itemize}
\item If $m=1$, $\Ker(I_r) = \emptyset$ so we do not require any diagonal stabiliser generators. For $m> 1$, the diagonal stabiliser generators are $\mathbf{S}_Z := \{XP_N(0|\mathbf{0}|N/2 \mathbf{z}_j)\}$ where $\mathbf{z}_j$ is the $j$th row of $K^r_m$ as in Eq.~\eqref{eq:ker_M^r_m}.
\item The non-diagonal stabiliser generators $\mathbf{S}_X := \{A_j\}$ are determined as follows:
    \begin{enumerate}
    \item Set $A_j = XP_N(0|\mathbf{x}_j|\mathbf{0})$ for $j \in [0\dots r-1]$ and $\mathbf{x}_j$ the $j$th row of $M^r_m$. 
    \item Update the $A_j$ for each of the operators $CP(p_i/q_i,\mathbf{v}_i)$:
        \begin{itemize}
            \item For $m_i = 1$, $A_j := A_j D_N(\frac{p_i}{q_i}2N\mathbf{x}_j\mathbf{w}_i)$. 
            \item For $m_i \ge 2$,  $A_j := A_j XP_N(0|\mathbf{0}|\frac{p_i}{q_i}\frac{2N}{2^{m_i-1}}\mathbf{a}\mathbf{x}_j\mathbf{w}_i)$ where $\mathbf{a}$ is the alternating vector of Eq.~\eqref{eq:alternating_vector} and $\mathbf{w}_i$ is the inclusion vector of Eq.~\eqref{eq:inclusion_vector} with respect to $\mathbf{v}_i$.
        \end{itemize}
    \end{enumerate}
\end{enumerate}

The algorithm can be optimised by only including qubits which for some operator $A_j$ has a non-trivial Z-component. In Proposition~\ref{prop:whg2xp_optimised}, we show that we can further optimise  for generalised controlled Z operators with $p_i/q_i = 1/2$ by replacing the factor $\mathbf{a}\mathbf{x}_j$ in step (b) by $\mathbf{v}_i[j]\mathbf{a}(\mathbf{x}_j - 1)$. This has the effect of clearing the Z-component of $A_j$ indexed by column $\mathbf{v}_i$ of $M^r_m$. This implies, for instance, that we can represent graph states, which are created using CZ operators, with a trivial embedding.

\begin{example}[Representing Weighted Hypergraph States as XP Stabiliser States]
In Ref.~\cite{weightedgraphuniversal}, an example of a weighted graph state is given which is a universal resource for measurement-based quantum computation; see Figure~\ref{fig:whg}. The unit cell of this state is \begin{align}
    |\phi\rangle = CP(1/2,1100) CP(1/2,0011) CP(1/4,1001) CP(1/4,0110)|+\rangle^{\otimes 4}\,.
\end{align}
The weighted graph state $|\phi\rangle$ can be represented as an embedded state $|\psi\rangle = \mathcal{E}_2^4|\phi\rangle$ stabilised by the following XP code of precision $N=4$ on $10$ qubits:
\begin{align}
\mathbf{S}_X &= \begin{matrix}
XP_4(0|1000111000|1000201000)\\
XP_4(0|0100100110|0100200100)\\
XP_4(0|0010010101|0010000102)\\
XP_4(0|0001001011|0001001002)
\end{matrix}\;\;\;
\mathbf{S}_Z = \begin{matrix}
XP_4(0|\mathbf{0}|2200200000)\\
XP_4(0|\mathbf{0}|2020020000)\\
XP_4(0|\mathbf{0}|2002002000)\\
XP_4(0|\mathbf{0}|0220000200)\\
XP_4(0|\mathbf{0}|0202000020)\\
XP_4(0|\mathbf{0}|0022000002)
\end{matrix}
\end{align}
Using the optimised method of Proposition~\ref{prop:whg2xp_optimised} and deleting redundant qubits, we find a more compact representation on $6$ qubits as follows:
\begin{align}
\mathbf{S}_X &= \begin{matrix}
XP_4(0|100010|320010)\\
XP_4(0|010001|230001)\\
XP_4(0|001001|003201)\\
XP_4(0|000110|002310)
\end{matrix}\;\;\;
\mathbf{S}_Z = \begin{matrix}
XP_4(0|\mathbf{0}|200220)\\
XP_4(0|\mathbf{0}|022002)
\end{matrix}
\end{align}

Detailed working for this example is available in the \href{https://github.com/m-webster/XPFpackage/blob/main/Examples/5.2_whg_to_xp.ipynb}{linked Jupyter notebook}.
\end{example}

\subsection{Discussion and Summary of Results}
In this chapter we have shown an equivalence between XP stabiliser states and weighted hypergraph states. For any XP stabiliser state, we can write a weighted hypergraph representation and vice-versa. A very wide range of states can be represented within the XP stabiliser formalism, including all weighted graph states and hypergraph states. 

These results may prove useful in implementing fault-tolerant versions of quantum algorithms. The Grover and Deutsch-Jozsa algorithms both employ real equally weighted (REW) pure states. In Ref.~\cite{hypergraph}, the authors showed that each REW state has an associated hypergraph state. As we can represent any hypergraph state within the XP formalism, this could be an interesting application.

\section{Logical Operators and the Classification of XP Codes}\label{chap:LO}
The objective of this section is to understand the logical operator structure of a given XP code. Our aim is to determine all XP operators that act as logical operators on the codespace (``logical XP operators''), and classify the logical actions that arise. 

We start by setting out definitions for logical XP operators and introduce the notion of a phase vector which allows us to describe the logical action of diagonal operators (Section~\ref{sec:LOdef}). 

In the Pauli stabiliser formalism, the stabiliser group $\langle \mathbf{G}\rangle$ is unique for a given codespace and a Pauli operator acts as a trivial logical operator if and only if it is an element of $\langle \mathbf{G} \rangle$.  In the XP formalism, there may be many different stabiliser groups for a given codespace. In Section~\ref{sec:LI}, we show how to find a set of XP operators $\mathbf{M}$  in the canonical form of Section~\ref{sec:canonical_generators} that generates the set of trivial logical XP operators and which uniquely defines the codespace.

In Section~\ref{sec:LO}, we show how to find a set of non-trivial logical XP operators $\mathbf{L}$ which together with $\mathbf{M}$ generates all logical XP operators. Using the example of Reed-Muller codes, we show that in some cases Pauli stabiliser codes can be viewed more naturally as XP codes and that we can systematically determine all possible logical XP operators for such codes using our techniques. Figure~\ref{fig:groupdiagram} explains how the various groups described above relate to each other.

In Section~\ref{sec:quantum_numbers}, we show how to assign quantum numbers to the codewords of Section~\ref{sec:code_word_algorithm_summary} based on the logical operator structure of the code. This in turn leads to a natural classification of XP codes into XP-regular and non-XP-regular codes, which we discuss in Section~\ref{sec:classification}. We show that each XP-regular code can be mapped to a CSS code which has the same diagonal logical operators and similar non-diagonal logical operators.

The algorithms for determining the generators for the logical XP group require the codewords of Section~\ref{sec:code_word_algorithm_summary} as input. In Section~\ref{sec:LO_Modified}, we demonstrate modified algorithms which take the canonical generators and orbit representatives of Section~\ref{sec:find_code_space} rather than the codewords as input. These methods are  more efficient than using the codewords as a starting point.

In Section~\ref{sec:LO_Action}, we describe a framework for analysing the action of diagonal logical XP operators based on the codeword quantum numbers. We show how to determine all possible diagonal logical actions for a given code and how to calculate an operator with a desired logical action. In Section~\ref{sec:LO_Classification}, we use this framework to classify diagonal logical XP operators into core and regular operators and demonstrate that complex logical operators arise in non-XP-regular codes.

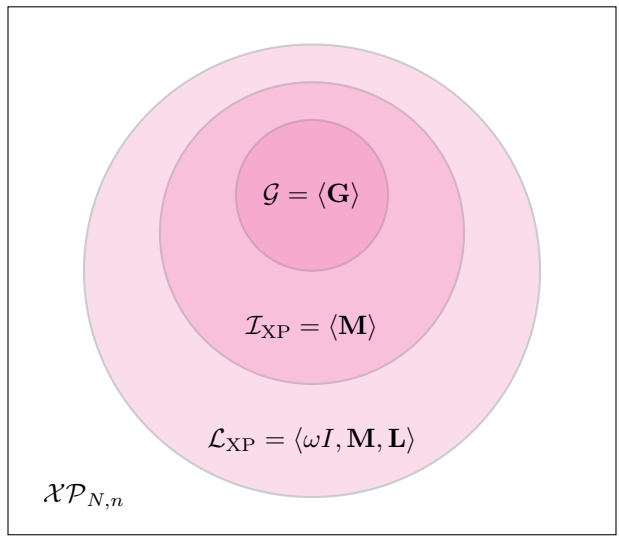
\begin{figure}[hbt!]
\centering
\begin{tikzpicture} 
[set/.style = {thick,
    fill=Rhodamine,
    opacity = 0.2,
    text opacity = 1}]
    
\filldraw[set] (0,1) circle (1);
\filldraw[set] (0,0.5) circle (2) ; 
\filldraw[set] (0,0) circle (3);
\draw (-4,-3.5) rectangle (4,3.5) ;

\draw (-3,-3) node {$\mathcal{XP}_{N,n}$};
\draw (0,-2.25) node  {$\mathcal{L}_\text{XP} = \langle \omega I, \mathbf{M}, \mathbf{L}\rangle$};
\draw (0,-0.75) node  {$\mathcal{I}_\text{XP} = \langle \mathbf{M}\rangle$};
\draw (0,1) node  {$\mathcal{G} = \langle \mathbf{G} \rangle$};
\end{tikzpicture}
\caption{\textbf{Relationship between XP Operator Groups}:  Here, $\mathcal{XP}_{N,n}$ is the group of all XP operators of precision $N$ on $n$ qubits.  The stabiliser group $\mathcal{G} = \langle \mathbf{G} \rangle$ is the same as the group generated by the canonical generators $\mathbf{S}_X,\mathbf{S}_Z$. The logical identity group $\mathcal{I}_\text{XP} = \langle \mathbf{M}\rangle$ fixes all elements of the codespace $\mathcal{C}$. It contains but is in general not equal to $\mathcal{G}$. The logical operator group $\mathcal{L}_\text{XP} = \langle \omega I, \mathbf{M}, \mathbf{L}\rangle$ is the set of XP operators that preserve the codespace.   }\label{fig:groupdiagram}
\end{figure}

\subsection{Definitions: Logical XP Operators}\label{sec:LOdef}
The \textbf{logical XP identity group}, denoted $\mathcal{I}_\text{XP}$, is the group of XP operators that fixes all elements in the codespace $\mathcal{C}$. The codewords $\{|\kappa_i\rangle\}$ of Section~\ref{sec:code_word_algorithm_summary} are a basis of the codespace, so we use the following definition:
\begin{align}
\mathcal{I}_\text{XP} := \{A \in \mathcal{XP}_{N,n} : A| \kappa_i\rangle = | \kappa_i\rangle, \forall i\} \,.
\end{align}
In the Pauli stabiliser formalism, the logical identity group for a stabiliser code is the same as the stabiliser group for the code. In the XP formalism, this is not necessarily the case and we show an example of this in \ref{eg:LI_code1}. Determining the logical identity group is non-trivial, so we present an algorithm to construct it which takes the codewords as input.

The \textbf{logical XP group}, denoted $\mathcal{L}_\text{XP}$, is the group of XP operators that preserves the codespace - that is, the set of XP operators $A \in \mathcal{XP}_{N,n}$ such that $A(\mathcal{C}) = \mathcal{C}$. In the Pauli stabiliser formalism, an operator $A$ is a logical operator if and only if it commutes with all the stabiliser generators. In the XP formalism, $A$ is a logical operator if and only if the group commutator of $A$ with any logical identity operator is in the logical identity group (see Proposition~\ref{prop:LOtest}). 

The action of a logical XP operator $A$ on the codewords $\{|\kappa_i\rangle\}$ can be described in terms of a vector $\mathbf{f}$ and a permutation $\pi$ of the codewords as follows:
\begin{align}
A|\kappa_i\rangle = \omega^{\mathbf{f}[i]}|\kappa_{\pi(i)}\rangle\label{eq:LO} \,.
\end{align}
The vector $\mathbf{f} \in \mathbb{Z}_{2N}^{\dim \mathcal{C}}$ is referred to as the \textbf{phase vector} and tells us which phase is applied by $A$ to each codeword. The permutation $\pi$ describes the non-diagonal action and $\pi^2 = 1$ (see Proposition~\ref{prop:LXP_action}). We call a logical operator \textbf{non-trivial} if it is not a logical identity (i.e. $A \in \mathcal{L}_\text{XP} \setminus \mathcal{I}_\text{XP}$).  If the phase vector for a diagonal operator $A$ is constant, say $c$, then the logical action of $A$ is $\omega^cI$. A diagonal operator is trivial iff its phase vector is zero.

\subsection{Determining the Logical Identity Group}\label{sec:LI}
We  present an algorithm to construct the generators of the logical identity group for an XP code.  This algorithm takes the codewords of the XP code, as in Section~\ref{sec:code_word_algorithm_summary}, as input. The result is a list of XP operators $\mathbf{M}$ in the canonical form of Proposition~\ref{prop:canonical_generators} that generate the logical identity group, or FALSE if the complex span of codewords is not the codespace of any XP code of precision $N$. Detailed proofs of the results in this section are found in Appendix~\ref{app:LO}.

We first demonstrate how to find generators for the diagonal logical identity group, then turn to the non-diagonal generators. 

\subsubsection{Diagonal Logical Identity Group Generators $\mathbf{M}_Z$}\label{sec:M_Z}

Our algorithm for finding the diagonal logical identity group is as follows:
\begin{enumerate}
\item Let $E$ be the Z-support of the codewords (as defined in Section~\ref{sec:code_words_notation}) and let $E_M$ be the matrix formed by taking $\{(\mathbf{e}|1):\mathbf{e} \in E\} \subset \mathbb{Z}_2^n \times \mathbb{Z}_2$ as rows.
\item Determine the Howell basis $K_M$ of $\Ker_{\mathbb{Z}_N}(E_M)$ (see Section~\ref{sec:howell_matrix_form}).
\item Let $(\mathbf{z}_k| q_k)$ denote the rows of $K_M$ and let $\mathbf{M}_Z = \{XP_N(2 q_k|\mathbf{0}|\mathbf{z}_k)\}$.
\item Calculate the orbit representatives corresponding to $\mathbf{M}_Z$ (see Section~\ref{sec:graph_search}). If the number of orbit representatives does not match the number of codewords, return FALSE.
\end{enumerate}
We show that this algorithm produces a generating set of diagonal logical identity operators in Proposition~\ref{prop:M_Z}.

\subsubsection{Non-diagonal Logical Identity Group Generators $\mathbf{M}_X$ }\label{sec:M_X}

We now set out the algorithm to find the non-diagonal generators, which consists of two steps:  first, we find the X components of the generators; second, we find the phase and Z components that are consistent with the relative phases between the computational basis elements in the codewords.

\paragraph{Step 1: Identifying the X-Components.}
Let $E_i$ be the Z-support of the codeword $|\kappa_i\rangle$ as defined in Section~\ref{sec:code_words_notation}. Take any element $\mathbf{e}_i \in E_i$ and let $T$ be the binary matrix with rows formed from  $\{ \mathbf{e}_i \oplus \mathbf{e}: \mathbf{e} \in E_i\}$. Let $S_X = \RREF_{\mathbb{Z}_2}(T)$ and let $\mathbf{m}_i = \res_{\mathbb{Z}_2}(T,\mathbf{e}_i)$ using the residue function defined in Eq.~\eqref{eq:residue_function}. Verify that $E_i = \mathbf{m}_i + \langle S_X \rangle$ for each codeword. If not, return FALSE.

\paragraph{Step 2: Determining the Phase and Z-Components.}
We use linear algebra modulo $N$  (see Section~\ref{sec:linalgmodN}) to find valid phase and Z-components for the generators identified in Step 1. Assume that the codewords $|\kappa_i\rangle$ are written in orbit form as in Eq.~\eqref{eq:orbitform}. For each $\mathbf{x}$ in $S_X$ we complete the following steps:
\begin{enumerate}
    \item Let $\mathbf{e}'_{ij} = \mathbf{e}_{ij} \oplus \mathbf{x}$, and let $p'_{ij}$ be the phase of $\mathbf{e}'_{ij}$ in the codewords
    \item Let $p''_{ij} = (p'_{ij} - p_{ij}) \mod 2N$. For there to be a valid solution, the $p''_{ij}$ are all either even (i.e. divisible by 2 in $\mathbb{Z}_{2N}$) or odd. Let $a = p_{00}'' \mod 2$ be an adjustment factor. Let $p'''_{ij} = (p''_{ij} - a)/2$ so that $p'''_{ij} \in \mathbb{Z}_N$. Let $\mathbf{p}'''$ be the vector formed from the $p'''_{ij}$
    \item Find a solution $(\mathbf{z}|q) \in \mathbb{Z}_N^{n}\times \mathbb{Z}_N$ such that $E_M^T (\mathbf{z}|q) = \mathbf{p}'''$ using linear algebra modulo $N$  (see Section~\ref{sec:linalgmodN}). 
    \item If there is no such solution, return FALSE
    \item Otherwise, let $K_M$ be the Howell basis  of $\Ker_{\mathbb{Z}_N}(E_M)$. Let $(\mathbf{z}'|p') = \res_{\mathbb{Z}_N}(K_M,(\mathbf{z}|q))$ and add the operator $XP_N(a+2p'|\mathbf{x}|\mathbf{z}')$ to $\mathbf{M}_X$.
\end{enumerate}
We show that this algorithm produces a non-diagonal logical identity operator in Proposition~\ref{prop:M_X}. 

\subsubsection{Properties of the Canonical Logical Identity Generators}\label{sec:code_space_test}
The algorithm to find the logical identity generators take the codewords as input and results in a set of generators $\mathbf{M}$ in canonical form. Hence, the set $\mathbf{M}$ uniquely identifies the codespace. Any XP group stabilising the codewords is composed of operators that act as logical identities on the codespace and so is a subgroup of $\langle\mathbf{M}\rangle$. In Proposition~\ref{prop:modified_LI}, we demonstrate an algorithm for determining $\mathbf{M}$ that does not require the codewords as input, and so results in a test for whether two sets of XP operators stabilise the same codespace.

The logical identity algorithm can also be used to determine if a given quantum state is a stabiliser state of an XP code of precision $N$. We simply apply the algorithm - if it fails, the state is not a valid XP stabiliser state. If it succeeds, the set of operators $\mathbf{M}$ is a set of stabiliser generators for the state. The method can also be used for codespaces with dimensions greater than $1$. We illustrate the algorithm for the logical identity generators using our main example:

\begin{example}[Logical Identity Group - Code 1]\label{eg:LI_code1}

Taking the codewords calculated for the code in Example~\ref{eg:code1} as input, we find that the canonical generators of the logical identity group are:
\begin{align}
    \mathbf{M}_X =
    \begin{matrix}
    XP_8( 9|1110000|0070000)\\
    XP_8(14|0001111|0001234)
    \end{matrix}\;, \quad    \mathbf{M}_Z =
    \begin{matrix}
    XP_8( 0|0000000|1070000)\\
    XP_8( 0|0000000|0170000)\\
    XP_8( 8|0000000|0004444)
    \end{matrix}
\end{align}
Compare these to the canonical generators:
\begin{align}
    \mathbf{S}_X =
    \begin{matrix}
    XP_8( 9|1110000|1240000)\\
    XP_8(14|0001111|0001234)
    \end{matrix}\,, \quad   \mathbf{S}_Z =
    \begin{matrix}
    XP_8( 8|0000000|2334444)\\
    XP_8( 0|0000000|0440000)
    \end{matrix}
\end{align}
By definition, any stabiliser of the codespace is a logical identity so $\langle \mathbf{S}_X,\mathbf{S}_Z\rangle \subset \langle \mathbf{M}_X,\mathbf{M}_Z\rangle$. In this example, the diagonal generators $\mathbf{S}_Z$ are not the same as $\mathbf{M}_Z$, but they have the same simultaneous $+1$ eigenspace. In all cases, $\langle\mathbf{S}_Z\rangle \subset \langle\mathbf{M}_Z\rangle$. For this example, we observe that none of the operators in $\mathbf{M}_Z$ are in $\langle\mathbf{S}_Z\rangle$, but all of the operators in $\mathbf{S}_Z$ are in $\langle\mathbf{M}_Z\rangle$. For example:
\begin{align}
XP_8( 8|0|2334444) &= XP_8( 0|0|1070000)^2 XP_8( 0|0|0170000)^3 XP_8( 8|0|0004444)
\end{align}

The non-diagonal generators $\mathbf{S}_X$ are the same as $\mathbf{M}_X$, up to a product of elements of $\langle\mathbf{M}_Z\rangle$.  Full working for this example is in the \href{https://github.com/m-webster/XPFpackage/blob/main/Examples/6.1_logical_identity_generators.ipynb}{linked Jupyter notebook}.
\end{example}

\subsection{Determining the Logical Operator Group}\label{sec:LO}
We now present an algorithm that will identify the logical XP operator group of an XP code.  This algorithm again takes as input the codewords of Section~\ref{sec:code_word_algorithm_summary}, and is similar to the algorithm for the logical identity group. The result is a list of XP operators $\mathbf{L}$ that generate the logical XP operator group together with $\mathbf{M}$ and $\omega I$. Detailed proofs of these results are in Appendix~\ref{app:LO}.

We first demonstrate how to find generators for the diagonal logical operator group, then turn to the non-diagonal generators.

\subsubsection{Diagonal Logical Operator Algorithm}\label{sec:L_Z}

The following algorithm gives a list of operators $\mathbf{L}_Z$ which together with $\omega I$ and $\mathbf{M}_Z$ generate all diagonal logical XP operators. We assume we have the codewords $\{ |\kappa_i\rangle \}$ in orbit form as in Eq.~\eqref{eq:orbitform} expressed as a linear combination of  computational basis elements $|\mathbf{e}_{ij}\rangle$. The key to finding the diagonal logical operators is to form a matrix from the binary vectors corresponding to the basis vectors for each codeword, and an index which indicates which codeword the vector is from:
\begin{enumerate}
\item For each $\mathbf{e}_{ij}$, let $\mathbf{i}$ be a binary vector of length $\dim(\mathcal{C})$ which is all zeros apart from the $i$th component which is $1$. 
\item Let $E_L$ be the matrix formed by taking  $(\mathbf{e}_{ij}|\mathbf{i}) \in \mathbb{Z}_2^n \times \mathbb{Z}_2^{\dim(\mathcal{C})}$ as rows.
\item Determine the Howell basis $K_L$ of  $\Ker_{\mathbb{Z}_N}(E_L)$ modulo $N$  (see Section~\ref{sec:howell_matrix_form}). 
\item Let $(\mathbf{z}_k| \mathbf{q}_k) \in \mathbb{Z}_N^n \times \mathbb{Z}_N^{\dim(\mathcal{C})}$ be a row of $K_L$ and let $M_Z$ be the matrix formed from the Z-components of the diagonal logical identity operators $\mathbf{M}_Z$ of Section~\ref{sec:M_Z}. Let $K$ be the Howell basis of the matrix formed from the residue of each row of the $\mathbf{z}_k$ over $\mathbb{Z}_N$ with respect to $M_Z$.  
\item Let $\mathbf{z}_k$ be a row of $K$ and let $B_k = XP_N(0|\mathbf{0}|\mathbf{z}_k)$. Then the set $\mathbf{L}_Z = \{B_k\}$ is a set of non-trivial diagonal logical operators which together with $\mathbf{M}_Z$ and $\omega I$ generate all diagonal logical XP operators. 
\end{enumerate}
Note that the operator $XP(1|\mathbf{0}|\mathbf{0}) = \omega I$ is always a logical operator for any XP code - it has the effect of applying a phase $\omega$ to all codewords. By convention, we do not include it in $\mathbf{L}_Z$. We show that the algorithm produces a generating set of diagonal logical XP operators in Proposition~\ref{prop:L_Z}.

\subsubsection{Non-diagonal Logical Operators}\label{sec:L_X}
In this section, we demonstrate how to find a generating set of non-diagonal logical XP operators. We first show how to find the X-components of a generating set of non-diagonal logical operators. We then demonstrate how to find valid phase and Z-components for a logical operator with a given X component.

\paragraph{Step 1: Identifying the X-Components.}\label{sec:LXx}
Here, we identify the valid X-components for all non-diagonal logical operators by calculating a coset decomposition of the orbit representatives $E_m$. In Section~\ref{sec:coset_structure_of_E}, we showed that we can write the Z-support of the codewords, $E$, in coset form $E = E_m +\langle S_X\rangle$. We can decompose the orbit representatives themselves into cosets $E_m = E_q + \langle L_X\rangle$. To find the matrix $L_X$ and set of vectors $E_q$, we use the following result:
\begin{proposition}[Coset Decomposition of $E_m$]\label{prop:Em_Decomposition}
Given a set of binary vectors $E_m \subset \mathbb{Z}_2^n$  there exists a unique binary matrix $L_X$ in RREF such that:
\begin{align}
    E_m &= E_q + \langle L_X \rangle\nonumber\\
    E_q &:= \{\text{Res}_{\mathbf{Z}_2}(L_X,\mathbf{e}): \mathbf{e} \in E_m\}\\ 
    \mathbf{x} \in \langle L_X \rangle &\iff \mathbf{x} \oplus E_m = E_m\nonumber
\end{align}

\end{proposition}
\begin{proof}
Let $T = \{\mathbf{x} \in \mathbb{Z}_2^n: \mathbf{x} \oplus \mathbf{e} \in E_m, \forall \mathbf{e} \in E_m\}$. Choose an arbitrary $\mathbf{e}_0 \in E_m$. Then $T$ is a subset of $E_m' = \{\mathbf{e} \oplus \mathbf{e}_0 : \mathbf{e} \in E_m\}$. We can check whether $\mathbf{x} \in E_m'$ is in $T$ by checking if $\mathbf{x} \oplus \mathbf{e} \in E_m,\forall \mathbf{e} \in E_m$. If $\mathbf{x} \in T$ then it permutes the elements of $E_m$ so $\mathbf{x} \oplus E_m = E_m$.

$T$ is a group under component-wise addition modulo $2$ because $\mathbf{x}_1, \mathbf{x}_2 \in T \implies \mathbf{x}_1 \oplus \mathbf{x}_2 \in T$, so $L_X = \RREF_{\mathbb{Z}_2}(T)$ generates $T$ under component wise addition modulo $2$.

Finally, the residue function is an equivalence relation partitioning $E_m$ into cosets of $\langle L_X\rangle$. 
\end{proof}

We can think of $T = \langle L_X\rangle$ as the group of all vectors $\mathbf{x}$ such that $\mathbf{x} \oplus E_m = E_m$. The X-component of any logical XP operator must be in $\langle L_X\rangle + \langle S_X\rangle$ because logical operators preserve the codespace. In Proposition~\ref{prop:L_X_X-components}, we show that logical operators with X-components in $L_X$ together with $\mathbf{L}_Z$, $\mathbf{S}_X$, $\mathbf{S}_Z$ and $\omega I$ generate the full set of logical XP operators.

\paragraph{Step 2: Valid Phase and Z-components.}
Assume that the codewords $|\kappa_i\rangle$ are written in orbit form as in Eq.~\eqref{eq:orbitform}. The algorithm for finding the phase and Z-components of the operators for a given $\mathbf{x} \in L_X$ is as follows:
 \begin{enumerate}
    \item Let $\mathbf{e}'_{ij} = \mathbf{e}_{ij} \oplus \mathbf{x}$, and let $p'_{ij}$ be the phase of $\mathbf{e}'_{ij}$in the codewords
    \item Let $p''_{ij} = (p'_{ij} - p_{ij}) \mod 2N$. For there to be a valid solution, for fixed $i$ the $p''_{ij}$ are all either even or odd. Let $a_i = p''_{i0} \mod 2$ be an adjustment factor
    \item Let $p'''_{ij} = (p''_{ij} - a_i)/2$ so that $p'''_{ij} \in \mathbb{Z}_N$, and let $\mathbf{p}'''$ be the binary vector with the $p'''_{ij}$ as components
    \item Find a solution $(\mathbf{z}|\mathbf{q}) \in \mathbb{Z}_N^{n}\times \mathbb{Z}_N^{\dim(\mathcal{C})}$ such that $E_L^T (\mathbf{z}|\mathbf{q}) \mod N = \mathbf{p}'''$ using linear algebra modulo $N$  (see Section~\ref{sec:linalgmodN}). Then $A = XP_N(0|\mathbf{x}|\mathbf{z})$ is a logical operator.
    \end{enumerate}
    
In Proposition~\ref{prop:L_X}, we show that the above algorithm generates a valid logical operator with X-component $\mathbf{x}$, or returns FALSE if this is not possible. The resulting operator is non-diagonal, but is not necessarily a logical X operator. For $A$ to be a logical X operator, then $A^2$ should be a logical identity operator. Applying the SQ rule of Section~\ref{sec:algebraic_identities}, $A^2$ is diagonal. Hence we require $A^2 \in \langle \mathbf{M}_Z \rangle$. We show how to adjust the phase and Z component to ensure this in Section~\ref{prop:L_X_X-components}.
	
\subsubsection{Examples: Logical Operators}
We now illustrate this algorithm for the logical XP operator group for two example XP codes:

\begin{example}[Logical Operators - Code 1]

The orbit representatives for Code 1 of Example~\ref{eg:code1} are given by Eq.~\eqref{eq:Em Code 1}. We calculate the coset decomposition of $E_m$ which gives us the X-components of the non-diagonal logical generators. Using the technique in Proposition~\ref{prop:Em_Decomposition}, we find that $L_X$ has $2$ rows and $E_q$ is of size $1$:
\begin{align}
L_X &=  \begin{pmatrix}
0000101\\
0000011\\
\end{pmatrix}\\
E_q &= \begin{pmatrix}0000001\end{pmatrix}
\end{align}
By adding elements of $\langle L_X \rangle$ to any orbit representative  $\mathbf{m}_i$, we can reach any other orbit representative $\mathbf{m}_j$. The logical operator group generators are:
\begin{align}
    \mathbf{L}_X =
    \begin{matrix}
XP_8( 2|0000101|0000204)\\
XP_8( 1|0000011|0000034)
    \end{matrix}\,,  \quad    \mathbf{L}_Z =
    \begin{matrix}
XP_8( 0|0000000|0002226)\\
XP_8( 0|0000000|0000404)\\
XP_8( 0|0000000|0000044)
    \end{matrix}
\end{align}
Full working for this example is in the \href{https://github.com/m-webster/XPFpackage/blob/main/Examples/6.2_logical_operator_generators.ipynb}{linked Jupyter notebook}. We look at the logical action of these logical operators in Example~\ref{eg:code1_action}.
\end{example}

\begin{example}[Logical Operators - Code 2]\label{eg:code2}
We now introduce our second main example, which is the code given by the following canonical stabiliser generators:
\begin{align}
\mathbf{S} = &\begin{matrix}XP_8(0|0000000|1322224)\\
XP_8(12|1111111|1234567)\end{matrix}
\end{align}
Using the graph search algorithm of Section~\ref{sec:graph_search}, the orbit representatives for Code 2 are:
\begin{align}
E_m &= \begin{pmatrix}0000000\\0000111\\0001011\\0001101\\0011110\\0011001\\0010101\\0010011\end{pmatrix}
\end{align}
The algorithm in Section~\ref{sec:code_word_algorithm_summary} yields the following codewords:
\begin{align}
&\begin{array}{rrr}
|\kappa_0\rangle =& |0000000\rangle&+\omega^{12}|1111111\rangle\\
|\kappa_1\rangle =& |0000111\rangle&\;\;\;+|1111000\rangle\\
|\kappa_2\rangle =& |0001011\rangle&+\omega^{14}|1110100\rangle\\
|\kappa_3\rangle =& |0001101\rangle&+\omega^{12}|1110010\rangle\\
|\kappa_4\rangle =& |0011110\rangle&\;\;\;+|1100001\rangle\\
|\kappa_5\rangle =& |0011001\rangle&+\omega^{8}|1100110\rangle\\
|\kappa_6\rangle =& |0010101\rangle&+\omega^{10}|1101010\rangle\\
|\kappa_7\rangle =& |0010011\rangle&+\omega^{12}|1101100\rangle
\end{array}
\end{align}
Calculating the coset decomposition of $E_m = E_q + \langle L_X \rangle$ using Proposition~\ref{prop:Em_Decomposition} we find that $L_X$ contains only one row but $E_q$ has $4$ elements:
\begin{align}
L_X = & \begin{pmatrix}
0011110\end{pmatrix}\\
E_q = &\begin{pmatrix}0000000\\0000111\\0001011\\0001101\end{pmatrix}
\end{align}
From the starting point $\mathbf{m}_i$, we can only reach orbit representatives $\mathbf{m}_j = \mathbf{m}_i \oplus 0011110$ by adding elements of $\langle L_X\rangle$ modulo $2$. We could, however, apply a unitary $U$ which permutes codewords as follows.
\begin{align}
| \kappa_0\rangle &\rightarrow | \kappa_1\rangle \rightarrow | \kappa_2\rangle \rightarrow | \kappa_3\rangle \rightarrow | \kappa_0\rangle\\
| \kappa_4\rangle &\rightarrow | \kappa_5\rangle \rightarrow | \kappa_6\rangle \rightarrow | \kappa_7\rangle \rightarrow | \kappa_4\rangle
\end{align}
Let the codewords in orbit form be $| \kappa_i\rangle = \sum_{0 \le j < 2^r}\omega^{p_{ij}}|\mathbf{e}_{ij}\rangle$. Let $k = (i + 1) \mod 4$ then $U$ is the operator given by:
\begin{align}
U &= \sum_{\mathbf{e}_{ij} \in E}|\mathbf{e}_{kj}\rangle\langle \mathbf{e}_{ij}| + \sum_{\mathbf{e} \in \mathbb{Z}_2^n \setminus E}|\mathbf{e}\rangle\langle \mathbf{e}|
\end{align}
$U$ cannot, however,  be written as an XP operator.

Applying the algorithms in sections~\ref{sec:L_Z} and~\ref{sec:L_X} we find the following generators for the logical XP group:
\begin{align}
\mathbf{L}_X =& XP_8(2|0011110|0012304)\,, \quad
\mathbf{L}_Z =\begin{matrix}
XP_8(0|0000000|0211112)\\
XP_8(0|0000000|0022220)\\
XP_8(0|0000000|0004004)\\
XP_8(0|0000000|0000404)\\
XP_8(0|0000000|0000044)\\
\end{matrix}
\end{align}
Full working for this example is in the \href{https://github.com/m-webster/XPFpackage/blob/main/Examples/6.2_logical_operator_generators.ipynb}{linked Jupyter notebook}.
\end{example}

\begin{example}[Reed Muller Codes]\label{eg:reedmuller}
In this example, we look at Reed Muller codes. These can be viewed as XP codes, and the algorithms of this chapter make it straightforward to determine their full logical operator structure. By varying the parameters of these codes, we show that they give rise to transversal logical operators at any level of the Clifford hierarchy.

We can write the 15-qubit Reed Muller code as a precision 2 code (i.e. using Pauli group operators) in terms of diagonal ($\mathbf{S}_Z$) and non-diagonal ($\mathbf{S}_X$) stabiliser generators:
\begin{align}
\mathbf{S}_Z &= \begin{array}{l}
XP_2(0|\mathbf{0}|100011100011101)\\
XP_2(0|\mathbf{0}|010010011011011)\\
XP_2(0|\mathbf{0}|001001010110111)\\
XP_2(0|\mathbf{0}|000100101101111)\\
XP_2(0|\mathbf{0}|000010000011001)\\
XP_2(0|\mathbf{0}|000001000010101)\\
XP_2(0|\mathbf{0}|000000100001101)\\
XP_2(0|\mathbf{0}|000000010010011)\\
XP_2(0|\mathbf{0}|000000001001011)\\
XP_2(0|\mathbf{0}|000000000100111) 
\end{array},\;\;\;
\mathbf{S}_X = \begin{array}{l}
XP_2(0|100011100011101|\mathbf{0})\\
XP_2(0|010010011011011|\mathbf{0})\\
XP_2(0|001001010110111|\mathbf{0})\\
XP_2(0|000100101101111|\mathbf{0})
\end{array}\label{eq:reedmuller}
\end{align}
The logical operators for precision $N=2$ are:

\begin{align}
\bar X =XP_2(0|000011111100001|\mathbf{0}),\;\;\;\bar Z =
XP_2(0|\mathbf{0}|000000000011111)
\end{align}

We can rescale this code to be of precision $N=8$ by multiplying the Z-components of the generators by 4. Applying the algorithm in Section~\ref{sec:LO}, we find additional diagonal operators as follows:

\begin{align}
\bar S^\dag = XP_8( 0|\mathbf{0}|000022222200002),\;\;\;\bar T^\dag = XP_8( 0|\mathbf{0}|111111111111111)
\end{align}
Note that $\bar S^\dag$ has an $S$ operator on same qubits as $\bar X$, whilst $\bar T^\dag$ has a $T$ operator on all 15 qubits. If we again rescale to precision $N=16$, we do not obtain any additional logical operators. This suggests that the code has a `natural precision' of $8$.
\end{example}

In the XPF, the stabiliser group for a given codespace is not unique. We can write a  more compact generating set of precision 4 operators that stabilise the same codespace. The operators are symmetrical in X and S and generate a different stabiliser group to those in Eq.~\eqref{eq:reedmuller}:
\begin{align}
\mathbf{S}_Z &= \begin{array}{l}
XP_4(0|\mathbf{0}|100011100011101)\\
XP_4(0|\mathbf{0}|010010011011011)\\
XP_4(0|\mathbf{0}|001001010110111)\\
XP_4(0|\mathbf{0}|000100101101111)
\end{array},\;\;\;
\mathbf{S}_X = \begin{array}{l}
XP_4(0|100011100011101|\mathbf{0})\\
XP_4(0|010010011011011|\mathbf{0})\\
XP_4(0|001001010110111|\mathbf{0})\\
XP_4(0|000100101101111|\mathbf{0})
\end{array}
\end{align}
Note here that the X-components of the non-diagonal generators and the Z-components of the diagonal generators are the rows of the binary matrix $M^4_4$ as defined in Section~\ref{sec:whg2xp}. 

In Proposition \ref{prop:reed-muller}, we generalise this example and show the Reed-Muller code on $2^r-1$ qubits can be written as the codespace of a precision $N = 2^{r-2}$ code whose stabiliser generators are symmetric in X and P with X and Z-components the rows of the matrix $M^r_r$ respectively. Pleasingly, this gives a self-dual set of stabiliser generators for all codes in this family. This generalises the known result for the Steane code, where $r = 3$ and $N=2$. Furthermore, to stabilise the code space of $n = 2^r - 1$ qubits, we require only $2r$ stabilisers, a logarithmic scaling. The code has natural precision of $2^{r-1}$ and a transversal logical $\text{diag}(1,\exp(2\pi i/2^{r-1}))$ operator.


You can explore which logical operators arise in codes with different parameters for Reed Muller codes in the linked  \href{https://github.com/m-webster/XPFpackage/blob/main/Examples/6.4_reed_muller.ipynb}{Jupyter notebook}.

\subsection{Assigning Quantum Numbers to the Codewords}\label{sec:quantum_numbers}
In this section, we demonstrate a natural way of assigning quantum numbers to the codewords of Section~\ref{sec:code_word_algorithm_summary}. This view of the codewords  gives rise to a classification of XP codes (Section~\ref{sec:classification}). It also allows us to develop more efficient algorithms to determine the logical operator group (Section~\ref{sec:LO_Modified}) and to analyse the logical action of operators (Sections~\ref{sec:LO_Action},~\ref{sec:LO_Classification}). 

The assignment of quantum numbers is based on the coset structure of the  Z-support of the codewords $E$ (see Section~\ref{sec:code_words_notation}). Recall from Section~\ref{sec:L_X} that $L_X$ is a set of binary vectors such that $E_m = E_q + \langle L_X \rangle$. Hence, we can write $E$ in coset form as:
\begin{align}
    E &= E_q + \langle S_X\rangle + \langle L_X \rangle\label{eq:Emcoset}\\
    &= \{(\mathbf{q}_l + \mathbf{u}S_X + \mathbf{v}L_X)\mod 2: \mathbf{q}_l \in E_q, \mathbf{u}\in \mathbb{Z}_2^r, \mathbf{v}\in \mathbb{Z}_2^k\}\label{eq:Emconstr}
\end{align}
We refer to $E_q$ as the \textbf{core} of the code. We can index elements of $E$ by writing:
\begin{align}
    \mathbf{e}_{l, \mathbf{u}, \mathbf{v}} &:= (\mathbf{q}_l + \mathbf{u}S_X + \mathbf{v}L_X)\mod 2
\end{align}
We can assign quantum numbers to the orbit representatives $\mathbf{m} \in E_m$ of Section~\ref{sec:code_words_notation} as follows:
\begin{align}
\mathbf{m}_{l, \mathbf{v}} &:= (\mathbf{q}_l + \mathbf{v} L_X)\mod 2
\end{align}
and these also apply to the codewords of Section~\ref{sec:code_word_algorithm_summary}:
\begin{align}
|\kappa_{l, \mathbf{v}}\rangle &:= O_{\mathbf{S}_X}|\mathbf{m}_{l, \mathbf{v}}\rangle\,.
\end{align} 
We refer to $l$ as the \textbf{Core Index}, $\mathbf{u}$ as the \textbf{Stabiliser Index} and $\mathbf{v}$ as the \textbf{Logical Index}.  The \textbf{orbit distance} is used to develop more efficient versions the logical identity and logical operator algorithms (see Section~\ref{sec:LO_Modified}), and is defined as:
\begin{align}
  \dist(\mathbf{e}) := \wt(\mathbf{u}) + \wt(\mathbf{v})\,.
\end{align}

\begin{example}[Quantum Numbers - Code 1]
For Code 1 of Example~\ref{eg:code1}, the orbit representatives can be written $E_m = E_q + \langle L_X \rangle$ where $E_q = \{0000001\}, L_X = \{0000100,0000010\}$. The full decomposition of $E_q + \langle S_X\rangle + \langle L_X \rangle$ and associated quantum numbers as per Section~\ref{sec:quantum_numbers} is:
\begin{align*}
\begin{array}{c | c || r: r r : r}
  &    &\multicolumn{4}{l}{\text{Stabiliser Index}}\\
\text{Logical Index}&\text{Core Index}&           00&            10&            01&            11\\
\hline
\hline
00&   0& 0000001& 1110001& 0001110&1111110\\
\hdashline
10&   0&0000100& 1110100&0001011& 1111011\\
01&   0& 0000010&1110010& 0001101& 1111101\\
\hdashline
11&   0&0000111& 1110111&0001000& 1111000
\end{array}
\end{align*}
The vectors in the first column of the table with Stabiliser Index $\mathbf{u} = \mathbf{0}$ correspond to the orbit representatives $E_m$. The vectors in row $i$ of the table are the Z-support of the $i$th codeword $|\kappa_i\rangle$ as in Eq.~\eqref{eq:kappa}. Each codeword can be identified by the quantum numbers $(l,\mathbf{v})$ where $\mathbf{v}$ is the Logical Index and $l$ is the Core Index. In this case, the size of the core is 1, so the core index is the same for all codewords. The dashed lines group together vectors with the same orbit distance. 
\end{example}

\begin{example}[Quantum Numbers - Code 2]
For code 2 of Example~\ref{eg:code2}, the full decomposition of $E_q + \langle S_X\rangle + \langle L_X \rangle$ and associated indexing as per Section~\ref{sec:quantum_numbers} is:
\begin{align*}
\begin{array}{c | c||r : r}
  &    &\multicolumn{2}{l}{\text{Stabiliser Index}}\\
 \hline
\text{Logical Index}&\text{Core Index}&   0&             1\\
\hline
\hline
 0&   0& 0000000&1111111\\
 0&   1& 0000111& 1111000\\
 0&   2& 0001011&1110100\\
 0&   3& 0001101&1110010\\
 \hdashline
 1&   0&0011110&1100001\\
 1&   1&0011001& 1100110\\
 1&   2&0010101& 1101010\\
 1&   3& 0010011& 1101100
 \end{array}
\end{align*}
The vectors with orbit distance 0, i.e. $\mathbf{u}=\mathbf{0}$ and $\mathbf{v}=\mathbf{0}$ correspond to the core $E_q$. In this case, the size of the core is 4, so we need both the Logical Index and the core index to specify the codewords. These examples are illustrated in the linked  \href{https://github.com/m-webster/XPFpackage/blob/main/Examples/6.5_quantum_numbers.ipynb}{Jupyter notebook}.
 \end{example}

\subsection{Classification of XP Codes}\label{sec:classification}
In this section, we present a way of classifying XP codes into XP-regular and non-XP-regular codes.  The main result is that each XP-regular code can be mapped via a diagonal unitary operator to a CSS code that has a very similar logical operator structure. We will see in Section~\ref{sec:LO_Classification} that non-XP-regular codes have a much richer logical operator structure that is distinct from PSF codes and so offer the possibility for interesting new classes of codes.

In Section~\ref{sec:xp-regular-def}, we introduce the concept of XP-regular codes and give some examples and elementary properties. In Section~\ref{sec:CoreClass}, we demonstrate that each XP-regular code can be mapped to a CSS code with identical diagonal logical operators, and similar non-diagonal logical operators. 

\subsubsection{Definition of XP-Regular and Non-XP-Regular Codes}\label{sec:xp-regular-def}

Consider an XP code, and let $E_q$ be the core of a code as defined in Section~\ref{sec:quantum_numbers}. If $|E_q|=1$, the code is \textbf{XP-regular}. Otherwise, the code is \textbf{non-XP-regular}.

One major difference between XP-regular and non-XP-regular codes is the codespace dimension. The codespace dimension for an XP code is $\dim(\mathcal{C}) = |E_q|2^k$ where $k = |L_X|$ (see Section~\ref{sec:quantum_numbers}). For an XP-regular code, $|E_q|=1$ so the codespace dimension is a power of $2$ and it encodes $k$ \textbf{logical qubits}. The codespace dimension of non-XP-regular codes may or may not be a power of $2$. Non-XP-regular codes are not additive, and their structure resembles that of the codeword stabilised (CWS) quantum codes of Ref.~\cite{cws}. The CWS class is very broad and includes all Pauli stabiliser and qudit stabiliser codes. There are examples of CWS codes which have better error correction properties than any known additive code with the same number of physical qubits.~\cite{562}.

In Ref.~\cite{xs}, a `regular code' is defined as one where the diagonal stabiliser generators of the code are elements of $\langle -I, Z\rangle^{\otimes n}$. All regular codes are XP-regular and in Section~\ref{sec:CoreClass} we will show a link between the two definitions. The examples below illustrate our definition of XP-regular codes:

\begin{example}[XP-Regular and Non-XP-Regular Codes]
Examples of  XP-Regular and Non-XP-Regular codes include:
\begin{enumerate}
\item All XP stabiliser states (i.e. XP codes with one-dimensional codespaces) are XP-regular; 
\item All Pauli codes (i.e. XP codes with precision $N=2$) are XP-regular;
\item Code 1 of Example~\ref{eg:code1} is XP-regular as $|E_q|=1$, though it is not regular according to the definition in Ref.~\cite{xs};
\item Code 2 of Example~\ref{eg:code2} is non-XP-regular, as $|E_q|=4$;
\end{enumerate}
\end{example}

\subsubsection{Mapping XP-Regular Codes to CSS Codes}\label{sec:CoreClass}

In this section, we show that each XP-regular code can be mapped via a diagonal unitary operator to a CSS code with a very similar logical operator structure. This is significant because it shows that the logical operator structure of an XP-regular code is no more complex than the corresponding CSS code. 

The algorithm for mapping a regular code whose canonical generators are $\mathbf{S}_X, \mathbf{S}_Z$ is:
\begin{enumerate}
    \item Determine a set of diagonal Pauli operators $\mathbf{R}_Z$ with the same simultaneous $+1$ eigenspace as the $\mathbf{S}_Z$
    \item If $\mathbf{S}_X = \{(p_i|\mathbf{x}_i|\mathbf{z}_i)\}$ then let $\mathbf{R}_X = \{(0|\mathbf{x}_i|\mathbf{0})\}$
    \item The mapped CSS code has stabiliser generators $\mathbf{R}_X,\mathbf{R}_Z$.
\end{enumerate}

We first show how to calculate a set of diagonal Pauli operators $\mathbf{R}_Z$ which have the same simultaneous $+1$ eigenspace as $\mathbf{S}_Z$. This links to the definition of `regular code' in Ref.~\cite{xs}, where regular codes were defined as those in which all diagonal generators are diagonal Paulis and implies that all `regular codes' are XP-regular.

\begin{lemma}
[Regular Diagonal Generators]\label{prop:regdiaggens}
If a code is XP-regular with diagonal canonical stabilisers $\mathbf{S}_Z$, then there exist diagonal Pauli operators whose simultaneous $+1$ eigenspace is the same as $\mathbf{S}_Z$.
\end{lemma}

\begin{proof}

If the code is XP-regular, it has core size 1. Let $\mathbf{q}$ be the sole element in the core. The Z-support of the simultaneous $+1$ eigenspace of the $\mathbf{S}_Z$ can be written as:
\begin{align}
E = \mathbf{q} + \langle S_X \rangle + \langle L_X \rangle
\end{align}
Let $G_X$ be the matrix formed from the rows of $L_X$ and $S_X$. The rows of $G_X$ are independent. Find the Howell basis $K$ of $\Ker_{\mathbb{Z}_2}(G_X)$. Because there are $|S_X| + |L_X| = r+k$ independent rows in $G_X$, there are $n-r-k$ independent rows in $K$.

Let $\mathbf{R}_Z = \{XP_2(-2\mathbf{q} \cdot \mathbf{z},\mathbf{0},\mathbf{z}): \mathbf{z} \in K\}$ and let the Z-support of the simultaneous $+1$ eigenspace of the $\mathbf{R}_Z$ be $E'$. Operators in $\mathbf{R}_Z$ stabilise all elements $\mathbf{e} \in E$ so $E \subset E'$.

Because $\mathbf{R}_Z$ has $n-k-r$ independent diagonal Pauli operators, $|E'| = 2^{r+k} = |E|$. Hence the simultaneous $+1$ eigenspaces of $\mathbf{S}_Z, \mathbf{R}_Z$ are the same.
\end{proof}

We are now in a position to prove the main result of this section:

\begin{proposition}[Mapping XP-Regular Codes to CSS Codes]\label{prop:cssmapping}

Given a XP-regular code $\mathbf{C}$ with canonical generators $\mathbf{S}_X, \mathbf{S}_Z$, there is a mapping to a CSS code $\mathbf{C}'$ with generators $\mathbf{R}_X, \mathbf{R}_Z$ such that:

\begin{enumerate}
\item If $| \kappa_i\rangle = \sum_{0 \le j < 2^r} \omega^{p_{ij}} |\mathbf{e}_{ij}\rangle$ is  codeword of $\mathbf{C}$ then $| \kappa_i\rangle' = \sum_{0 \le j < 2^r} |\mathbf{e}_{ij}\rangle$ is a codeword of $\mathbf{C}'$;
\item $\mathbf{C}'$ has the same diagonal logical operators as $\mathbf{C}$;
\item If $XP_N(p|\mathbf{x}|\mathbf{z})$ is a non-diagonal logical operator of $\mathbf{C}$, then $XP_N(0|\mathbf{x}|\mathbf{0})$ is a logical operator of $\mathbf{C}'$.
\end{enumerate}
\end{proposition}

\begin{proof}

Let $\mathbf{S}_X$ and $\mathbf{S}_Z$ be the canonical generators for $\mathbf{C}$. Because $\mathbf{C}$ is XP-regular, by Proposition~\ref{prop:regdiaggens} we can find a decomposition $E = \mathbf{q} + \langle S_X \rangle + \langle L_X \rangle$ and  $\mathbf{R}_Z$ which from $\langle -I, Z\rangle^{\otimes n}$ and which have the same simultaneous $+1$ eigenspace as the $\mathbf{S}_Z$. Let $\mathbf{R}_X = \{XP_N(0|\mathbf{x}|\mathbf{0}) : \mathbf{x} \in L_X  \}$ and let $\mathbf{C}'$ be the code defined by the stabiliser generators $\mathbf{R}_X, \mathbf{R}_Z$.

By construction, the simultaneous $+1$ eigenspaces of the diagonal generators are the same (and so are the respective Z-supports i.e. $E = E'$). The matrices formed from the X-components of the non-diagonal canonical generators are the same (i.e. $S_X = R_X$). Hence, the orbit representatives are the same (i.e. $E_m' = E_m = \{\mathbf{m}_i\}$), and so $| \kappa_i\rangle' = O_{\mathbf{R}_X} |\mathbf{m}_i\rangle = \sum_j |\mathbf{e}_{ij}\rangle$ form a basis for the codespace of $C'$.

The algorithm for computing the diagonal logical operators in Section~\ref{sec:L_Z} depends only on the Z-support of the codewords. As $E' = E$, the diagonal logical operators are the same.

In Section~\ref{sec:L_X}, we showed that the X components of the logical operators depend only on $E=E'$, so for any logical operator $XP_N(p|\mathbf{x}|\mathbf{z})$ of $\mathbf{C}$, $XP_N(0|\mathbf{x}|\mathbf{0})$ is a logical operator of $\mathbf{C}'$.
\end{proof}

We can transform the codespaces  $\mathcal{C}'$ to $\mathcal{C}$ by applying the diagonal unitary $U$ specified by:
\begin{align}
U = \sum_{\mathbf{e}_{ij} \in E}\omega^{p_{ij}}|\mathbf{e}_{ij}\rangle\langle \mathbf{e}_{ij}| + \sum_{\mathbf{e} \in \mathbb{Z}_2^n \setminus E}|\mathbf{e}\rangle\langle \mathbf{e}|\,.
\end{align}

 \begin{example}[Map XP-Regular Code to CSS - Code 1]
Code 1 of Example~\ref{eg:code1} is an XP-regular code. Applying Proposition~\ref{prop:cssmapping}  regular generators are given by:
\begin{align}
\mathbf{R}_Z &=
\begin{array}{l}
XP_2( 0|0000000|1010000)\\
XP_2( 0|0000000|0110000)\\
XP_2( 2|0000000|0001111)
\end{array}
\,\ \quad
\mathbf{R}_X =
\begin{array}{l}
XP_2( 0|1110000|0000000)\\
XP_2( 0|0001111|0000000)
\end{array}\,.
\end{align}
Full working for this example is in the \href{https://github.com/m-webster/XPFpackage/blob/main/Examples/6.8_mapping_XP-regular_code.ipynb}{linked Jupyter notebook}.
\end{example}

\subsubsection{Error Correction for XP-Regular Codes}

We now briefly consider error correction for XP codes. One of the complexities associated with this problem is that the stabiliser generators of XP codes are not guaranteed to commute, so simultaneous measurements may not be possible. A possible error-correction routine for \textbf{XP-regular codes} is as follows. In this case, we calculate the regular diagonal generators $\mathbf{R}_Z$ which are diagonal Pauli operators. By measuring the $\mathbf{R}_Z$ first and correcting for errors, we  guarantee that we are in a subspace where all stabiliser generators commute and can complete the error-correction process by measuring the non-diagonal generators $\mathbf{S}_X$. A similar process is outlined in section V of Ref.~\cite{xs}. 

Another approach would be to adopt the error correction methods for codeword stabilised (CWS) quantum codes presented in Ref.~\cite{cws}. In this approach, the codewords of a Pauli stabiliser code are viewed as translations of a graph state $|S\rangle$ by tensors of Pauli Z operators $XP_2(0|\mathbf{0}|\mathbf{w}_i)$. The vectors $\mathbf{w}_i$ are called \textbf{word operators} and form a classical code over $\mathbb{Z}_2^n$. Any single-qubit error can be propagated along the edges of the graph state by applying stabilisers of $|S\rangle$ and converted into a tensor of Pauli Z operators $XP_2(0|\mathbf{0}|\mathbf{z})$. This allows decoding using classical techniques using the bit string $\mathbf{z}$. To use this approach, we would need to be able to convert an arbitrary XP code into a graph state translated by tensors of Pauli Z operators. We have seen in Chapter~\ref{chap:hypergraph} that XP states can be represented as weighted hypergraph states. The orbit representatives $E_m$ of XP codes play a similar role to the word operators of CWS codes so this approach may have potential. 

\subsection{Modified Logical Operator Algorithms}\label{sec:LO_Modified}
The  logical identity and logical operator algorithms presented in sections~\ref{sec:LI} and ~\ref{sec:LO} require as input the codewords in the orbit form of Eq.~\eqref{eq:orbitform}. In the Pauli stabiliser formalism, there are algorithms for determining the logical Z and X operators that do not require us to first calculate the codewords. Can we find such algorithms in the XP formalism?

In this section, we demonstrate modified algorithms for determining the logical operator group which do not require the codewords as input. Instead, we take the the canonical generators and orbit representatives of Chapter~\ref{chap:codewords} as a starting point. We show that the modified logical identity algorithm is significantly more efficient than the original version. As a result, the modified algorithm can be used to determine if two different XP groups have the same codespace. 

\subsubsection{Modified Logical Identity Algorithm}\label{sec:LI_Modified}
The modified algorithm for determining the logical XP identity group generators can be used where the precision of the code is a power of $2$. The logical identity algorithm uses the Z-support $E$ of the codewords (see Section~\ref{sec:code_words_notation}). If $N = 2^t$, we only need to consider elements of $E$ at most orbit distance $t$ from the core (see Section~\ref{sec:quantum_numbers},  Proposition~\ref{prop:modified_LI}).

The main steps of the modified algorithm are as follows. 
\begin{enumerate}
\item Given a set of generators $\mathbf{G}$ for the stabiliser group, we calculate the canonical generators $\mathbf{S}$ and orbit representatives $E_m$ (see Chapter~\ref{chap:codewords}). 
\item From $\mathbf{S}, E_m$, we can efficiently calculate $E_q, S_X$ and $L_X$ without calculating $E$ in full.  
\item The elements of $E$ at most orbit distance $t$ from the core are $E_t = \{(\mathbf{q} + \mathbf{u}S_X + \mathbf{v}L_X) \mod 2 : \mathbf{q} \in E_q, \mathbf{u} \in \mathbb{Z}_2^r, \mathbf{v} \in \mathbb{Z}_2^k; \wt(\mathbf{u}) + \wt(\mathbf{v}) \le t\}$. We can then use $E_t$ instead of $E$ to determine $\mathbf{M}$ in the algorithm set out in Section~\ref{sec:LI}.
\end{enumerate}

Next, we look at the computational complexity of the modified logical identity algorithm versus that of the original version. Let the coset decomposition of $E = E_q + \langle S_X\rangle + \langle L_X\rangle$ as in  Eq.~\eqref{eq:Emcoset}. The number of elements in $E$ is $|E| = q 2^{k+r}$ where $q = |E_q|, r = |S_X|, k = |L_X|$. Hence, the complexity of the original logical identity algorithm, which requires us to perform row operations on a matrix of size $|E|$, is $\mathcal{O}(q 2^{k+r})$.

We now consider the computational complexity of the modified logical identity algorithm for codes of various precisions. For PSF codes, $N = 2 = 2^1$ and we only need to consider elements of $E_1$ which are at most orbit distance 1 from the core to find the logical identity group. The total number of elements to consider is $q + q(k+r)$. Because $q=1$ for all Pauli stabiliser codes, the run time of the modified algorithm is $\mathcal{O}(k+r)$ in this case.

For XS stabiliser codes, $N=4 = 2^2$. We need to consider elements of $E_2$ which are up to orbit distance $2$ from the core. Hence, the logical identity algorithm is $\mathcal{O}(q + q(k+r) + q\binom{k+r}{2}) = \mathcal{O}(q (k+r)^2)$

In general, we see that for precision $N=2^t$, the modified logical identity algorithm is $\mathcal{O}(q (k+r)^t)$.  For small $N$ and large $r$ or $k$, this can result in significantly faster run time compared to the original version which is $\mathcal{O}(q 2^{k+r})$. Topological codes tend to have a small number of logical qubits but a large number of stabiliser generators, so this is an important improvement. We can also express the complexity in terms of the number of qubits $n$ as $\mathcal{O}(n^t)$ because $2^n \ge |E| = q 2^{k+r}$.

\subsubsection{Modified Logical Operator Algorithm}
Similarly, we can use a modified version of the logical operator algorithm where the precision of the XP code is a power of $2$ (say $N=2^t$). Instead of calculating the codewords, we only need to consider the elements of $E$ up orbit distance $t$ from the orbit representatives $E_m$. Proof of this claim is in Section~\ref{app:modified_LI+LO}.

\subsection{Diagonal Logical Actions Arising in an XP Code}\label{sec:LO_Action}
In the Pauli stabiliser formalism, methods exist to find the logical $Z$ operators for a code. In the XP formalism, a code may have logical operators with a wider range of actions - for instance, logical $S, T$ or $\sqrt{T}$ operators, as well as logical controlled phase operators - for example logical $CZ, CCZ$ and $CT$ operators.  In this section, we show how to describe all possible logical actions a diagonal logical operator can apply for a given code. We show how to determine whether a particular logical action is possible, and if so how to calculate a logical operator with this action.

We first show how to describe all possible actions which can be applied by the diagonal logical operators of a code. We use the phase vectors of Section~\ref{sec:LOdef} to describe the logical action of a diagonal operator. Let $\mathbf{L}_Z'$ be the set of diagonal logical operators plus $\omega I$ - the logical operator that applies a phase of $\omega$ to each codeword. Let $F_Z$ be the matrix whose rows are the phase vectors of $\mathbf{L}_Z'$. We can calculate the Howell basis $F_D = \How_{\mathbb{Z}_{2N}}(F_Z)$ by using the techniques in Appendix~\ref{app:linalg} so that $F_D = U F_Z$ for some matrix $U$. The phase vectors which can be applied by a diagonal logical operator of the code are given by $\Span_{\mathbb{Z}_{2N}}(F_D)$. 

We next show how to find a logical operator whose action is given by a phase vector in $\Span_{\mathbb{Z}_{2N}}(F_D)$. Let $\mathbf{f}_i,\mathbf{u}_i$ be rows of $F_D, U$ respectively. Then the operator $L_i = \mathbf{L}_Z'^{\mathbf{u}_i}$ has the logical action given by $\mathbf{f}_i$, using the generator product notation of Eq.~\eqref{eq:generator_product}. The operators $\mathbf{L}_D = \{L_i : 0 \le i < |F_D|\}$ generate diagonal logical operators with all possible phase vectors. Let the required phase vector be $\mathbf{f} \in \Span_{\mathbb{Z}_{2N}}(F_D)$ so that $\mathbf{f} = \mathbf{u} F_D \mod 2N $ for some $\mathbf{u} \in \mathbb{Z}_{2N}^{|F_Z|}$, then the XP operator $\mathbf{L}_D^{'\mathbf{u}}$ has phase vector $\mathbf{f}$.

 \begin{example}[Logical Action - Code 1 and Code 2]\label{eg:code1_action}
For \textbf{Code 1}  of Example~\ref{eg:code1}, the phase vectors for each logical operator in $\mathbf{L}$ are:
\begin{align}
\mathbf{L}_Z &=
\begin{array}{l r r r r r}
XP_8(0|0000000|0002226)&\mathbf{f}:&12&4&4&4\\
XP_8(0|0000000|0000404)&\mathbf{f}:&8&8&0&0\\
XP_8(0|0000000|0000044)&\mathbf{f}:&8&0&8&0\\
\end{array}\\
\mathbf{L}_X &=
\begin{array}{l r r r r r}
XP_8(10|0000101|0000600)&\mathbf{f}:&0&0&0&0\\
XP_8(9|0000011|0004434)&\mathbf{f}:&0&0&0&0\\
\end{array}
\end{align}
Hence operator $XP_8( 0|\mathbf{0}|0002226)$ applies a phase of $\omega^{12}$ on the first codeword and $\omega^4$ on the other codewords. Calculating $F_Z$ and $F_D$ we find:
\begin{align}
F_Z = &\begin{pmatrix}1&1&1&1\\12&4&4&4\\8&8&0&0\\8&0&8&0\end{pmatrix},\;\;\;
F_D = \How_{\mathbb{Z}_8}(F_Z) = \begin{pmatrix}1&1&1&1\\0&8&0&0\\0&0&8&0\\0&0&0&8\end{pmatrix}
\end{align}
The following diagonal operators generate diagonal logical operators with all possible phase vectors:
\begin{align}
\mathbf{L}_D &=
\begin{array}{r r r r r r}
XP_8(1|0000000|0000000)&\mathbf{f}:&1&1&1&1\\
XP_8(4|0000000|0006266)&\mathbf{f}:&0&8&0&0\\
XP_8(12|0000000|0002262)&\mathbf{f}:&0&0&8&0\\
XP_8(4|0000000|0002666)&\mathbf{f}:&0&0&0&8
\end{array}
\end{align}
The action of operator $XP_8( 4|\mathbf{0}|0006266)$ is to apply $\omega^8 = -1$ to the second codeword only. As we have two logical qubits, this is a logical CZ operation.

In fact, the logical effects we can obtain are generated by $\omega I$, and CZ on logical indices $01, 10$ and $11$. We can make combinations of these operators to generate the Z operators on the first and second logical qubit:

\begin{align}
\begin{array}{l r r r r r r}
\Bar{Z}_{10} &= XP_8(8|0000000|0000044)&\mathbf{f}:&0&8&0&8\\
\Bar{Z}_{01} &= XP_8(0|0000000|0004040)&\mathbf{f}:&0&0&8&8\\
\end{array}
\end{align}

For \textbf{Code 2} of Example~\ref{eg:code2}, we have the following logical operators and corresponding phase vectors:
\begin{align}
\mathbf{L}_Z &=\begin{array}{l r r r r r | r r r r}
XP_8(0|0000000|0211112)&\mathbf{f}:&0&8&8&8&8&8&8&8\\
XP_8(0|0000000|0022220)&\mathbf{f}:&0&8&8&8&0&8&8&8\\
XP_8(0|0000000|0004004)&\mathbf{f}:&0&8&0&0&8&0&8&8\\
XP_8(0|0000000|0000404)&\mathbf{f}:&0&0&8&0&8&8&0&8\\
XP_8(0|0000000|0000044)&\mathbf{f}:&0&0&0&8&8&8&8&0\\
\end{array}\\
\mathbf{L}_X &= \begin{array}{l r r r r r | r r r r}
XP_8(14|0011110|0074160)&\mathbf{f}:&0&0&0&0&0&0&0&0
\end{array}
\end{align}
Note that the bar in the phase vector groups together codewords with the same logical index, but with different core indices. Calculating the Howell basis of $F_Z$, we obtain the generators:
\begin{align}
\mathbf{L}_D &= \begin{array}{l r r r r r | r r r r}
XP_8( 1|0000000|0000000)&\mathbf{f}:&1&1&1&1&1&1&1&1\\ 
XP_8( 0|0000000|0237336)&\mathbf{f}:&0&8&0&0&0&0&8&8\\
XP_8( 0|0000000|0026260)&\mathbf{f}:&0&0&8&0&0&0&8&0\\
XP_8( 0|0000000|0026620)&\mathbf{f}:&0&0&0&8&0&0&0&8\\
XP_8( 0|0000000|0277772)&\mathbf{f}:&0&0&0&0&8&0&0&0\\
XP_8( 0|0000000|0673332)&\mathbf{f}:&0&0&0&0&0&8&8&8\\
\end{array}
\end{align}
The logical Z operator on the single logical qubit is:
\begin{align}
\Bar{Z}_1 &= \begin{array}{l r r r r r | r r r r}
XP_8( 0|0000000|0062224)&\mathbf{f}:&0&0&0&0&8&8&8&8\\
\end{array}
\end{align}
Full working for these examples is in the \href{https://github.com/m-webster/XPFpackage/blob/main/Examples/6.9_logical_action.ipynb}{linked Jupyter notebook}.
\end{example}

\begin{example}[Logical Action - Hypercube Codes]
In Example~\ref{eg:reedmuller}, we saw that the Reed Muller code on $2^r - 1$ qubits has a transversal logical $\text{diag}(1,\exp(i 2\pi/2^{r-1})$ operator. In this example, we view the Hypercube code of dimension $D$ as an XP code of precision $N = 2^D$. We show that it has transversal generalised controlled Z logical operators at the $(D-1)$st level of the Clifford hierarchy. This fact was first pointed out in Ref.~\cite{hypercube}, but is easily verified using the techniques of this section and we calculate the corresponding XP operators in the \href{https://github.com/m-webster/XPFpackage/blob/main/Examples/6.10_hypercube_codes.ipynb}{linked Jupyter Notebook}.
\end{example}
\subsection{Classification of Logical Operators}\label{sec:LO_Classification}
We now introduce a classification scheme for diagonal logical operators based on the quantum numbers assigned to the codewords in Section~\ref{sec:quantum_numbers}. A \textbf{regular XP logical operator} is an XP operator that applies the same phase to codewords with the same \textbf{logical index}. A \textbf{core logical XP operator} is an XP operator that applies the same phase to codewords with the same \textbf{core index}.

To determine if an operator is regular or core, we reshape the phase vector for the operator so that rows correspond to codewords with the same logical index and columns to codewords with the same core index.

 \begin{example}[Regular and Core Operators - Code 2]
For Code 2 of Example~\ref{eg:code2}, which is a non-XP-regular code, consider the operator \linebreak $A = XP_8( 0|\mathbf{0}|0062224)$. The phase vector for $A$ is:
\begin{align*}
\begin{array}{l|l}
	&\text{Core Index}\\
	\text{Logical Index}&
		\begin{array}{l l l l}
		0&1&2&3
		\end{array}
	\\
	\hline
	\begin{array}{l}0\\1\end{array}&
	\begin{array}{l l l l}0&0&0&0\\8&8&8&8\end{array}
\end{array}
\end{align*}
The operator $A$ applies a phase of $\omega^8 = -1$ for logical index 1, so it is a regular operator.
Now consider operator $B = XP_8( 0|\mathbf{0}|0026620)$. The phase vector for $B$ is:
\begin{align*}
\begin{array}{l|l}
	&\text{Core Index}\\
	\text{Logical Index}&
		\begin{array}{l l l l}
		0&1&2&3
		\end{array}
	\\
	\hline
	\begin{array}{l}0\\1\end{array}&
	\begin{array}{l l l l}0&0&0&8\\0&0&0&8\end{array}
\end{array}
\end{align*}
The operator $B$ applies a phase of $-1$ to codewords with core index 3, so it is a core operator. Now consider $C = XP_8( 0|\mathbf{0}|0277772)$ which has phase vector:
\begin{align*}
\begin{array}{l|l}
	&\text{Core Index}\\
	\text{Logical Index}&
		\begin{array}{l l l l}
		0&1&2&3
		\end{array}
	\\
	\hline
	\begin{array}{l}0\\1\end{array}&
	\begin{array}{l l l l}0&0&0&0\\8&0&0&0\end{array}
\end{array}
\end{align*}
One might think that all logical operators are either core operators, regular operators, or products of core and regular operators. However, the operator $C$ is a counterexample to this hypothesis and demonstrates that more complex logical operators arise in non-XP-regular codes. The operator $C$ applies a phase of $-1$ when the core index is $0$ \textbf{and} the Logical Index is $1$. Because it applies a phase of $-1$ to one of the 8 codewords, the operator can be thought of as a CCZ gate. Full working for these examples is in the \href{https://github.com/m-webster/XPFpackage/blob/main/Examples/6.11_logical_operator_classification.ipynb}{linked Jupyter notebook}. 
\end{example}

\subsection{Logical Operators - Summary and Discussion}

In this chapter, we have shown how to determine the logical operator structure for any XP code. We have presented algorithms to calculate generators for the logical operator and logical identity groups using linear algebra techniques (Sections~\ref{sec:LI} and~\ref{sec:LO}). In contrast to the Pauli and qudit stabiliser formalism, XP codespaces are not uniquely identified by the stabiliser group. Two XP codes have the same codespace if and only if they have the same logical identity generators. The efficiency of these algorithms depends on the precision $N$ of the XP code. Where $N=2^t$, the algorithms are of $\mathcal{O}(n^t)$ complexity where $n$ is the number of qubits (see Section~\ref{sec:LO_Modified}). In the worst case, we need to determine the codewords in full to determine the logical operator group.

By allocating quantum numbers to the codewords, we can analyse the logical action of diagonal XP operators and fully classify which logical actions arise. We can determine all possible logical actions applied by operators of XP form, which can include logical operators at various levels of the Clifford hierarchy. These techniques give a more complete picture of the logical operator structure than previous methods, even when looking at Pauli stabiliser codes.

The coset decomposition of the orbit representatives $E_m$ yields the core $E_q$ of the code (Section~\ref{sec:coset_structure_of_E}), and allows us to determine the non-diagonal logical operators. The size of the core $E_q$ leads to a classification of XP codes into XP-regular and non-XP-regular codes. Any XP regular code can be mapped to a CSS code which has a similar logical operator structure via a unitary transformation. Non-XP regular codes have a more complex logical operator structure. Though not fully developed in this paper, there appear to be several possible approaches for error-correction of XP codes, despite the fact that stabiliser generators do not commute in general.

The main limitation of the above algorithms is that we consider only logical operators of XP form. An area for further investigation would be to develop algorithms to find the non-XP unitary operators which act as logical operators. 

\section{Measurements in the XP Formalism}\label{chap:Measurements}
Determining the extent to which computations on a quantum computer can be classically simulated is one of the central questions in the field of quantum information.  In the Pauli stabiliser formalism, the Gottesman-Knill theorem states that stabiliser circuits can be classically simulated efficiently. In particular, given a Pauli stabiliser code, we can efficiently simulate the measurement of any Pauli operator on the codespace, including both exact calculation of the Born rule probabilities for such measurements as well as the update rule to determine the post-measurement state.  In this chapter, we look at whether a similar result holds in the XPF - i.e. can the measurement of XP operators can be simulated efficiently in the XPF? 

In Section~\ref{sec:MeasDef}, we set out our assumptions and criteria for an XP operator to be `XP-measurable' on an XP code. In Section~\ref{sec:meas-outcome-probabilities}, we show how to determine the outcome probabilities for measurement of arbitrary XP operators on an XP codespace. In Section~\ref{sec:measure_diag_pauli}, we present an efficient stabiliser algorithm for measuring diagonal Pauli operators on XP codes. We consider whether we can do the same for precision $4$ diagonal XP operators in Section~\ref{sec:meas_precision_4} and show that estimating outcome probabilities is not tractable for these. We also give examples showing that the measurement of some XP operators takes us outside the XP formalism.

\subsection{Measurement Definitions}\label{sec:MeasDef} 

We first define what we mean by an operator being \emph{measurable} within the XP formalism.

Let $\mathbf{C}$ be an XP code with canonical stabiliser generators $\mathbf{S}$. Assume the system is described by the density operator $\rho$ which is proportional to the projector onto the codespace defined by $\mathbf{S}$. We can write $\rho$ in terms of the codewords and the Z-support $E$ of the codewords of Section~\ref{sec:code_words_notation} as follows:
\begin{align}
\rho := \frac{1}{|E|} \sum_i |\kappa_i\rangle\langle \kappa_i|\label{eq:rho}
\end{align}
Let $A$ be an XP operator, and $A_\lambda$ the projector onto the $\lambda$ eigenspace of $A$. The operator $A$ is \textbf{XP-measurable} on $\mathbf{C}$ if for each eigenvalue $\lambda$ of $A$:
\begin{enumerate}
\item We can calculate the probability of obtaining outcome $\lambda$ which is given by 
\linebreak $\Pr(\lambda) = \text{Tr}(A_\lambda \rho A_\lambda)$; and
\item We can find a set of XP operators $\mathbf{S}^\lambda$ such that the projector onto the codespace defined by $\mathbf{S}^\lambda$ is proportional to $A_\lambda \rho A_\lambda$.
\end{enumerate}
Note that in the above definition, we are only concerned with whether the above tasks can be done in principle, not whether they can be done efficiently. 

\subsection{Outcome Probabilities for Measurements of XP Operators}\label{sec:meas-outcome-probabilities}
In this section, we demonstrate how to calculate the outcome probabilities for measurement of arbitrary XP operators on an XP code. We assume that we are given the codewords in orbit format as input (as in Eq.~\eqref{eq:orbitform}). 

For diagonal operators, we can calculate the outcome probabilities by looking at the Z-support of the codewords $E$ (see Section~\ref{sec:code_words_notation}). Let $A$ be the diagonal XP operator we wish to measure and assume it has $+1$ as an eigenvalue. Let $E^+$ be the set of binary vectors $E^+ = \{\mathbf{e}\in E: A|\mathbf{e}\rangle =  |\mathbf{e}\rangle\}$. We show in Proposition~\ref{prop:prob_diag} that the probability of obtaining outcome $+1$ when measuring $A$ is $\Pr(+1) = \frac{|E^+|}{|E|}$.

We can also calculate outcome probabilities for non-diagonal operators if given the codewords. The method is set out in Proposition~\ref{prop:prob_nondiag} and is somewhat more complex than for diagonal XP operators. 

Hence, given the codewords, we can in principle determine the outcome probabilities when measuring any XP operator. This does not necessarily mean that the probabilities can be estimated efficiently or that we can express the resulting state of the system as an XP code as we see in Section~\ref{sec:meas_precision_4}

\subsection{Stabiliser Method for Measurement of Diagonal Pauli Operators}\label{sec:measure_diag_pauli}
We now consider the measurement of diagonal Pauli operators - i.e. elements of $\langle -I, Z\rangle^{\otimes n}$. In this special case, we can estimate the outcome probabilities efficiently and we can always express the resulting state as an XP code. 

In this section, we demonstrate an efficient, stabiliser-based update algorithm to simulate measurements of diagonal Pauli operators. The algorithm provides update rules for the \emph{core form} of an XP code after measurement. We first describe the core form of an XP code. We then state the algorithm for measuring diagonal Pauli operators and give some examples which illustrate the algorithm.

\subsubsection{Core Form of an XP Code}\label{sec:core_formalism}
In Ref.~\cite{simulation}, measurements are simulated by determining update rules for the stabiliser generators, logical Z operators, logical X operators and anti-commutators. 

Our algorithm gives update rules for XP codes in \textbf{core form}. The core form of an XP code consists of the following data: the core $E_q$, which was introduced in Section~\ref{sec:quantum_numbers} and is a set of binary vectors; the non-diagonal canonical generators $\mathbf{S}_X$ of Section~\ref{sec:canonical_generators}; and the logical X operators $\mathbf{L}_X$ - of Section~\ref{sec:L_X}. 

The core form encapsulates the key properties of the code in a compact way. From $E_q$, $\mathbf{S}_X$ and $\mathbf{L}_X$ we can generate the orbit representatives and the codewords (see Section~\ref{sec:quantum_numbers}). If necessary, we can efficiently calculate the diagonal logical operators and logical identities using the algorithms in sections~\ref{sec:LI} and~\ref{sec:LO}.

\subsubsection{Algorithm for Measuring Diagonal Paulis}\label{sec:meas_diag_paulis}
Assume we have an XP code in core form (i.e. the non-diagonal stabiliser generators $\mathbf{S}_X$ generating non-diagonal logical operators $\mathbf{L}_X$ and core $E_q$ as in Section~\ref{sec:core_formalism}). As per our discussion in Section~\ref{sec:vector_representation}, if the precision of the code $N$ is not a multiple of 2, then $Z$ operators do not exist. We can if necessary re-scale the code to be of precision $2N$ by doubling the phase and Z components of all stabiliser generators. Assume we wish to measure the diagonal Pauli operator $A = XP_2(0|\mathbf{0}|\mathbf{z})$. Note that $A$ has zero phase component but the algorithm can be generalised to operators with non-trivial phase components very easily. Our aim is to determine an XP code in core format representing the post-measurement system, as well as the probability of measuring each eigenvalue of $A$ (in this case, $\pm 1$).

The algorithm uses the \textbf{parity function} of a binary vector $\mathbf{x}$ with respect to the binary vector $\mathbf{z}$ which is defined as: 
\begin{align}
\text{Par}_{\mathbf{z}}(\mathbf{x}) := \mathbf{x}\cdot\mathbf{z} \mod 2\label{eq:parity_function}
\end{align}
\paragraph{Step 1:} Determine if there exists $B \in \mathbf{S}_X$ or failing that, $B \in \mathbf{L}_X$ with X-component $\mathbf{x}$ such that $\text{Par}_{\mathbf{z}}(\mathbf{x}) = 1$. If $B$ does not exist, go to Step 2. If $B$ exists, we update $\mathbf{S}_X \cup \mathbf{L}_X$ and $E_q$ via the following steps:
\begin{itemize}
\item Remove $B$ from $\mathbf{S}_X \cup \mathbf{L}_X$
\item For any $C \in \mathbf{S}_X \cup \mathbf{L}_X$ with X-component $\mathbf{y}$ such that $\text{Par}_{\mathbf{z}}(\mathbf{y}) = 1$, replace $C$ with $BC$.
\item Update $E_q$ by setting $E_q = E_q \cup \{\mathbf{q}_l \oplus \mathbf{x} : \mathbf{q}_l \in E_q\}$
\end{itemize}
\paragraph{Step 2:} Split $E_q$ into two sets:
\begin{align}
E_q^+ &:= \{\mathbf{q} \in E_q : \text{Par}_{\mathbf{z}}(\mathbf{q}) = 0\}\\
E_q^- &:= \{\mathbf{q} \in E_q : \text{Par}_{\mathbf{z}}(\mathbf{q}) = 1\}
\end{align}
The probability of obtaining the outcome $+1$ is $\Pr(+1) = \frac{|E_q^+|}{|E_q|}$ and the post-measurement core  is $E_q^+$. The probability of obtaining $-1$ is $\Pr(-1) = \frac{|E_q^-|}{|E_q|}$ and the updated core is $E_q^-$. 

In Appendix~\ref{app:measurement}, we explain in detail why the algorithm works. Essentially this is because the parity function of Eq.~\eqref{eq:parity_function} commutes with the addition of vectors modulo $2$. Once we have the code in core format, the above algorithm simulates measurement of diagonal Paulis in $\mathcal{O}(|E_q|+|\mathbf{S}_X|+|\mathbf{L}_X|)$ time complexity. 

The algorithm generalises the method of simulating measurements in the Pauli stabiliser formalism (e.g. in Ref.~\cite{simulation}). When simulating measurements in the Pauli stabiliser formalism, we look for stabiliser generators which do not commute with the operator being measured. The parity function serves a similar purpose in our algorithm. Any operator $B$ with X-component $\mathbf{x}$ commutes with $A$ if  $\text{Par}_{\mathbf{z}}(\mathbf{x}) = 0$ and anticommutes otherwise. Any computational basis vector $|\mathbf{e}\rangle$ with $\text{Par}_{\mathbf{e}}(\mathbf{x}) = q$ is in the $(-1)^q$ eigenspace of $A$. In Step 1, we need to remove at most 1 non-diagonal operator from $\mathbf{S}_X \cup \mathbf{L}_X$, which is also the case for the Pauli stabiliser formalism. 

One significant difference between the XPF and the Pauli stabiliser formalism is the possible outcome probabilities which arise. In the Pauli stabiliser formalism, outcome probabilities are always a multiple of $\frac{1}{2}$ when measuring a single operator. This is not the case in the XPF because for non XP-regular codes, the sizes of $E_q, E_q^+$ and $E_q^-$ may not be powers of $2$. Outcome probabilities in the XPF may be irrational numbers - in particular when measuring non-diagonal operators (see Proposition~\ref{prop:prob_nondiag}).

\subsubsection{Examples - Measurement of Diagonal Paulis}
Below are examples of the measurement of diagonal Paulis for Code 2  of Example~\ref{eg:code2} to illustrate the operation of the algorithm. Full working for these examples is in the \href{https://github.com/m-webster/XPFpackage/blob/main/Examples/7.1_measure_diagonal_Pauli.ipynb}{linked Jupyter notebook}.

 \begin{example}[No update to $\mathbf{S}_X \cup \mathbf{L}_X$ - Code 2, Non-XP-Regular Code]

Code 2 expressed in core form is:
\begin{align}
E_q =& \{0000000,0000111,0001011,0001101\}\\
\mathbf{S}_X =& XP_8(12|1111111|0334567)\\
\mathbf{L}_X =& XP_8(14|0011110|0012340)
\end{align}
In this example, we measure $A = XP_2( 0|0000000|0111111)$. Looking at elements of $\mathbf{S}_X \cup \mathbf{L}_X$, the parity of all operators is 0, so we do not need to update them in Step 1 of the algorithm.

Moving to Step 2, we calculate $E_q^+, E_q^-$ as follows:
\begin{align}
E_q^+ =& \{0000000\}\\
E_q^- =& \{0000111,0001011,0001101\}
\end{align}
The probability of measuring $+1$ is $\Pr(+1) = |E_q^+|/|E_q| = 1/4$, whilst the probability of measuring $-1$ is $\Pr(-1) = |E_q^-|/|E_q| = 3/4$. These probabilities do not arise when measuring a single operator in the Pauli stabiliser formalism.
\end{example}

 \begin{example}[Update $\mathbf{S}_X \cup \mathbf{L}_X$ - Code 2, Non-XP-Regular Code]

In this example, we measure $A = XP_2( 0|0000000|0000100)$ on Code 2.  

For Step 1, we find for $B = XP_8(12|1111111|0334567)$ the X-component $\mathbf{x} = 1111111$ has parity $+1$. We remove $B$ from $\mathbf{S}_X$. In $\mathbf{L}_X$, we also find $C = XP_8(14|0011110|0012340)$ with X-component $\mathbf{y} = 0011110$ has parity 1. Replace $C$ with $BC = XP_8(14|1100001|0700003)$. Update $E_q$ by adding $\mathbf{x} \oplus \mathbf{q}$ for $\mathbf{q} \in E_q$ to $E_q$ so the updated code in core format is:
\begin{align}
E_q = \{&0000000,0000111,0001011,0001101,\nonumber\\
&1111111,1111000,1110100,1110010\}\\
\mathbf{S}_X =& \emptyset\\
\mathbf{L}_X =& XP_8(14|1100001|0700003)
\end{align}

For Step 2, we calculate $E_q^+, E_q^-$:
\begin{align}
    E_q^+ =& \{0000000,0001011,0010011,0011001\}\\
    E_q^- =& \{0000111,0001101,1111111,1110100\}
\end{align}

We obtain measurement outcomes $+1$ or $-1$ with equal probability of $\frac{1}{2}$. This is always the case when we update $\mathbf{S}_X \cup \mathbf{L}_X$ in Step 1 because for the binary vector $\mathbf{q}$, exactly one of $\mathbf{q},\mathbf{q}\oplus \mathbf{x}$ has parity $+1$.
\end{example}

\subsection{Measuring Precision $4$ XP Operators}\label{sec:meas_precision_4}
In the previous section, we showed that diagonal Paulis are XP-measurable on any XP code. In this section, we look at the measurement of precision $4$ XP operators, which can be considered the next most complex case. We show that determining the outcome probabilities for diagonal precision $4$ operators is in general computationally complex. We look at two examples which illustrate that precision 4 XP operators are not in general XP-measurable because the post-measurement states cannot be expressed as an XP codespace.

\subsubsection{Estimating Outcome Probabilities of Diagonal Precision 4 Operators is Intractable}
Consider measuring a diagonal precision 4 XP operator $A = XP_4(0|\mathbf{0}|\mathbf{z})$. The probability of obtaining outcome $+1$ when measuring $A$ is $\frac{|E^+|}{|E|}$ where $E$ is the Z-support of the pre-measurement codewords and $E^+ = \{\mathbf{e}\in E : A|\mathbf{e}\rangle = |\mathbf{e}\rangle\}$ (see Proposition~\ref{prop:prob_diag}).  Hence, determining probability outcomes reduces to the problem of finding the simultaneous $+1$ eigenspace of the XP operators $\mathbf{S}_Z \cup \{A\}$.

Now consider simulating measurements on an XP code stabilising $|+\rangle^{\otimes m}$ for some large value of $m$. Assume we have a series of diagonal XP operators $A_i$ of precision 4 which all share $+1$ as an eigenvalue. Proposition~\ref{prop:prob_diag} states that the probability of obtaining the result $+1$ after measuring the series of operators depends on the dimension of the simultaneous $+1$ eigenspace of the $A_i$. Determining $E^+$ is known to be an NP-complete problem~\cite{xs}. No matter which algorithm we use, we must calculate $|E^+|$ so this complexity seems unavoidable in the general case. 

\subsubsection{XP Formalism is not Closed under Measurement of XP Operators}
The following examples illustrate that when measuring precision 4 operators, it is not always possible to represent the post-measurement system as the codespace of an XP code. Full working for these examples is in the \href{https://github.com/m-webster/XPFpackage/blob/main/Examples/7.3_measure_precision_4.ipynb}{linked Jupyter notebook}:

\begin{example}[Measurement of Diagonal Precision 4 Operator]

Let us measure the diagonal operator $A = S_1S_2^3S_3^3$ on the code defined by the stabiliser generators $\mathbf{S}_X = \{X_1, X_2, X_3\}, \mathbf{S}_Z = \emptyset$ where $X_i$ denotes the Pauli $X$ operator applied to the $i$th qubit. The codespace is one dimensional and spanned by $|\kappa\rangle = |+\rangle^{\otimes 3}$. 

The operator $A$ has 4 eigenvalues, with Z-supports of the corresponding eigenspaces as follows:
\begin{align}E^{+1} &= \{000,101,110\}&E^{+i} &= \{100\}\\E^{-1} &= \{011\}&E^{-i} &= \{001,010,111\}
\end{align}
The probability of obtaining each measurement result is:
\begin{align}
\Pr(+1) &= 3/8&\Pr(+i) &= 1/8\\\Pr(-1) &= 1/8&\Pr(-i) &= 3/8 \,.
\end{align}
In the case of outcome $+1$, the post-measurement state is
\begin{align}
|\kappa^+\rangle &= |000\rangle + |010\rangle + |111\rangle \,.
\end{align}
In Section~\ref{sec:coset_structure_of_E}, we demonstrated that the codewords of XP codes have Z-support of size $2^r$ for some integer $r$ so this state cannot be written as the codespace of any XP code.

\end{example}

Similarly, for non-diagonal precision $4$ operators, we can easily find examples of operators where the post-measurement system cannot be represented as the codespace of an XP code:

\begin{example}[Measurement of Non-diagonal Precision 4 Operator]
Now let us measure the non-diagonal operator $B = XP_4(2|111|123)$ on an XP code stabilising $|\kappa\rangle = |+\rangle^{\otimes 3}$. The square of $B$ is $B^2 = I$, so the eigenvalues are $\pm 1$. 

Let $\omega = \exp(i\pi/4)$ and note that $\frac{1 + \omega^2}{2} = \frac{1 + i}{2} = \frac{\omega}{\sqrt{2}}$ and $\frac{1 + \omega^6}{2} = \frac{1 - i}{2} = \frac{\omega^7}{\sqrt{2}}$. 

Applying the projector $B^+$ to $|\kappa\rangle$ using Proposition~\ref{prop:xp_projectors}, we obtain:
\begin{align}
    |\kappa^+\rangle &=\frac{1}{2}(|\kappa\rangle + B|\kappa\rangle)\\
    &=\frac{1}{2}\Big(|000\rangle +|001\rangle + |010\rangle+ |011\rangle+ |100\rangle + |101\rangle +|110\rangle +|111\rangle\nonumber\\
    &+\omega^6|000\rangle +|001\rangle + \omega^2|010\rangle -|011\rangle -|100\rangle +\omega^6|101\rangle +|110\rangle +\omega^2|111\rangle\Big)\\
    &=\frac{1}{\sqrt{2}}\Big(\omega^7|000\rangle +\sqrt{2}|001\rangle + \omega|010\rangle +\omega^7|101\rangle +\sqrt{2}|110\rangle +\omega|111\rangle\Big)
    \end{align}
Calculating the probability of obtaining outcome $+1$:
\begin{align}
    \Pr(+1) &= \frac{\langle\kappa^+|\kappa^+\rangle}{\langle\kappa|\kappa\rangle}\\
    &= \frac{1 + 2 + 1 +1 + 2+ 1}{2\cdot 8} =  \frac{1}{2}
    \end{align}
Similarly
\begin{align}
    |\kappa^-\rangle &=\frac{1}{2}(|\kappa\rangle - B|\kappa\rangle)\\
    &=\frac{1}{\sqrt{2}}\Big(\omega|000\rangle  + \omega^7|010\rangle+\sqrt{2}|011\rangle+\sqrt{2}|100\rangle +\omega|101\rangle  +\omega^7|111\rangle\Big)\\
    \Pr(-1) &= \frac{1 + 1 +2+2+1 +  1}{2\cdot 8} =  \frac{1}{2}
\end{align}
The state $|\kappa^+\rangle$ cannot be represented as the codespace of an XP code because:
\begin{itemize}
    \item The size of the Z-support $\ZSupp(|\kappa^+\rangle)$ is not a power of 2,
    \item The coefficients of the computational basis elements have different moduli,
    \item It is not possible to find a set of diagonal XP operators $\mathbf{S}_Z^+$ where the Z-support of the simultaneous $+1$ eigenspace of the $\mathbf{S}_Z^+$ is equal to $\ZSupp(|\kappa^+\rangle)$.
\end{itemize}

\end{example}

\subsection{Measurement in the XP Formalism - Summary of Results}
We have demonstrated that it is not in general possible to efficiently simulate measurement of XP operators in the XP formalism, apart from the special case of measuring diagonal Pauli operators. Firstly, measurement of an XP operator on an XP code may result in a state which cannot be described as an XP codespace. We have seen two examples which illustrate this. Secondly, calculating outcome probabilities when measuring a series of diagonal operators requires us to determine the simultaneous +1 eigenspace of these operators. This is known to be an NP-complete problem when the operators are of precision $4$.

Hence, there appears to be no obvious generalisation of the Gottesman-Knill theorem to XP codes. This suggests that XP codes can describe states which display computationally complex, non-classically simulable behaviour.
 
\section{Discussion and Open Questions}\label{chap:openquestions}

In this paper, we have set out the foundations for the XP formalism. We have formulated XP versions of many of the algorithms available in the Pauli stabiliser formalism - for instance determining a basis for the codespace, generators for the logical operator group and simulating the measurement of diagonal Pauli operators. The computational complexity of these algorithms depends on the precision $N$ of the XP code, and certain edge cases have exponential complexity. We have given examples of XP operators which cannot be measured within the formalism or where estimating outcome probabilities is NP-complete. Hence, there appears to be no obvious generalisation of the Gottesman-Knill theorem to XP codes. XP codes are on the boundary of what is classically simulable, and so there are good reasons to believe that XP codes allow us to engineer states which exhibit useful, non-classically simulable behaviour.

The rich logical operator structure of XP codes may make them useful for applications such as magic state distillation, which requires codes with non-Clifford logical operators. In Ref.~\cite{css_optimal}, the authors showed  that triorthogonal CSS codes have optimal error correction parameters for Pauli stabiliser codes with a transversal logical T operator. We note that this result only applies to XP codes of precision $2$, and in particular does not apply to non-XP-regular codes. XP codes with transversal logical non-Clifford operations (for instance T or CCZ gates) could be used for fault-tolerant preparation of magic states. We have focused so far on XP codes where the precision is a power of 2. Where the precision is not a power of 2, XP codes may have logical operators which are outside the Clifford hierarchy. 

We have described the states and phase functions which arise within the XP formalism. As part of this, we have shown that two important classes of states, hypergraph and weighted graph states, can be represented as XP stabiliser states. Hence, we can use the algorithms presented in this paper to analyse these. One of the main benefits of looking at them as XP stabiliser codes is that we can very easily determine the symmetries of the states as these are just the elements of the logical identity group. In Figure~\ref{fig:union_jack_generators} and the the \href{https://github.com/m-webster/XPFpackage/blob/main/Examples/8.1_union_jack_symmetries.ipynb}{linked Jupyter notebook}, we illustrate how the algorithms in this paper make it easy to determine the $\mathbb{Z}_2$ symmetries of the Union Jack state of Ref.~\cite{unionjack}.

\captionsetup[subfigure]{margin=5pt}
\begin{figure}[hbt!]
\centering
\begin{subfigure}[t]{.5\textwidth}
  \centering
  \includegraphics[width=\linewidth]{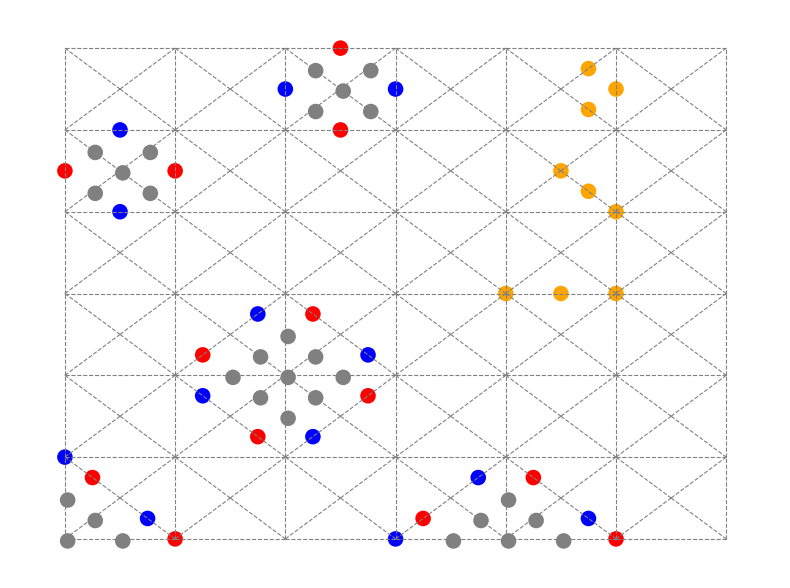}
    \subcaption{\textbf{Stabiliser Generators for Union Jack State}:  Using the techniques in this paper, we can represent the Union Jack State as an XP code. This makes it easier to study the $\mathbb{Z}_2$ symmetries of the state. In this figure, we illustrate a sample set of stabiliser generators for the XP code. There are qubits on each vertex and each edge of the cellulation. The different coloured dots represent the application of operators for the stabiliser generator as follows. Grey: $X$; yellow $Z$; red $S$; blue: $S^3$.  }
    \end{subfigure}%
\begin{subfigure}[t]{.5\textwidth}
  \centering
      \includegraphics[width=\linewidth]{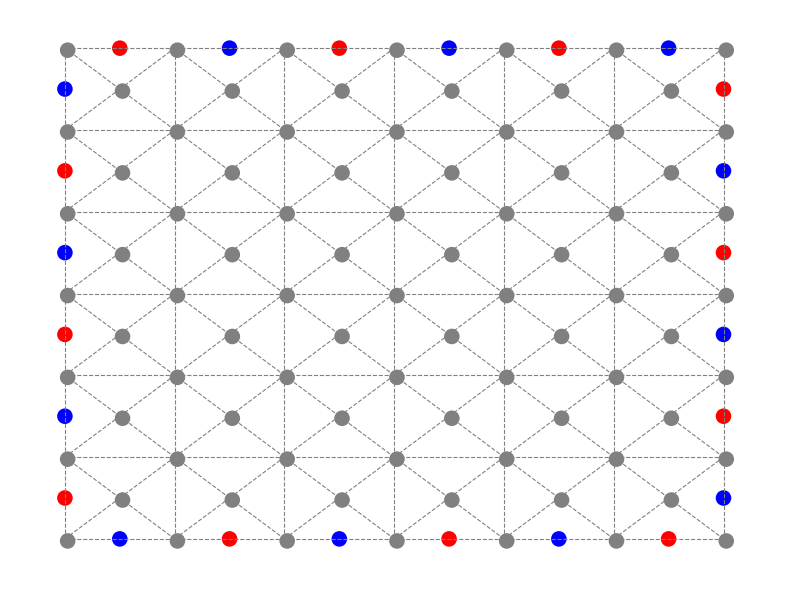}
    \subcaption{\textbf{Onsite Symmetry for Union Jack State}:  By multiplying together all non-diagonal stabiliser generators, we find an onsite symmetry which is preserved apart from a phase component on certain qubits on the boundary. This is indicative of symmetry protected topological order. In the \href{https://github.com/m-webster/XPFpackage/blob/main/Examples/8.1_union_jack_symmetries.ipynb}{linked Jupyter notebook}, we show that there are 3 onsite symmetries for the Union Jack state on a square plane with open boundary, two of which are are preserved on the lattice and one which is not.  }
     \end{subfigure}
     \caption{Finding $\mathbb{Z}_2$ Symmetries of the Union Jack state by representing it as  an XP code.}
     \label{fig:union_jack_generators}
\end{figure}

In the Pauli and qudit stabiliser formalisms, operators commute up to a phase. In the XP stabiliser formalism, operators commute up to a diagonal operator. This suggests that non-Abelian anyon models arise naturally in the XP formalism, and the existence of such codes in the XS stabiliser formalism had already been established in Ref.~\cite{xs}. Such models may be useful in achieving fault-tolerant quantum computation~\cite{anyons-kitaev} as well as understanding physical phenomena such as the fractional Quantum Hall Effect~\cite{anyons-hall}. In Ref.~\cite{Abelian_anyons}, the authors showed that all twisted quantum doubles (TQDs) with Abelian anyons can be represented as qudit stabiliser codes. It would be interesting to see if there is a similar result for  non-Abelian TQDs represented as XP codes. You can explore the logical operator structure of various topological codes, including Twisted Quantum Doubles, in the \href{https://github.com/m-webster/XPFpackage/blob/main/Examples/8.2_topological_codes.ipynb}{linked Jupyter notebook}.

One of the most exciting implications of this work is that no-go theorems which apply to Pauli stabiliser codes and commuting stabiliser codes do not necessarily apply to XP codes. Such results exist, for instance, in the area of self-correcting quantum memories (Ref.~\cite{finitetemp} on page 23). Rather than rely on active error correction, a self-correcting quantum memory is a system where large scale errors would be suppressed by a macroscopic energy barrier. As XP codes involve stabilisers which are not Paulis and which may not commute, they are worth investigating for potential use as self-correcting memories.

\begin{acknowledgements} 
This research was supported by the Australian Research Council via the Centre of Excellence in Engineered Quantum Systems (EQUS) project number CE170100009. BJB also received support from the European Union’s Horizon 2020 research and innovation programme under the Marie Skłodowska-Curie grant agreement No. 897158. Mark Webster was supported by the Sydney Quantum Academy, Sydney, NSW, Australia. We would also like to thank Paul Webster who reviewed early drafts of this paper and provided valuable feedback.
\end{acknowledgements}

\appendix

\section{Linear Algebra over Rings}\label{app:linalg}
Most readers will be familiar with linear algebra techniques for vector spaces over a field  - for instance, solving linear equations or finding a basis of a subspace. In this paper, we work with vectors over $\mathbb{Z}_N$, which is a ring in general rather than a field.  Linear algebra over rings is not covered in standard linear algebra textbooks, so we give an introduction here (for more background, see Refs.~\cite{storjohann1} and ~\cite{buchmann}).

We start with basic concepts from ring theory in Section~\ref{sec:ring_concepts}. We then consider the row span of a matrix over a ring. For vector spaces, we can calculate the Reduced Row Echelon Form (RREF) of a matrix and this gives us a basis for the subspace spanned by the rows of the matrix. In Section~\ref{sec:howell_matrix_form}, we introduce the Howell matrix form which is a generalisation of the RREF for rings. Calculation of the Howell basis is central to many of the algorithms in this paper. 

We show how to solve linear equations modulo $N$ in Section~\ref{sec:linalgmodN} and how to find the intersection of spans in Section~\ref{sec:affinespan}. These techniques are used when calculating the logical operators of XP codes.

\subsection{Ring Concepts}\label{sec:ring_concepts}
A ring $R$ is a set of elements with addition and multiplication binary operations. Elements of a ring are not guaranteed to have multiplicative inverses.  Elements which do have inverses are called \textbf{units.} In contrast, a \textbf{zero divisor} is an element $a \ne 0$ where $ab = 0$ for some $b\ne 0$.
 
We  define an equivalence relation such that $a \sim b \iff a = ub$  for some unit $u$ (or $a$ is an \textbf{associate} of $b$ ). For each element a of $\mathbb{Z}_N$, we can  calculate a \textbf{minimal associate} $m_a := \GCD(N,a)$. We can show that $a \sim b \iff m_a = m_b$. 

\begin{example}[Ring Definitions for $\mathbb{Z}_8$]
Using the ring $\mathbb{Z}_8$ as an example: 
\begin{enumerate}
\item The units  are: $\{1,3,5,7\}$. 
\item The zero divisors are: $\{2,4,6\}$
\item The minimal associates are:
\end{enumerate}
\begin{align}
\begin{matrix}\mathbf{a}&0&1&2&3&4&5&6&7\\\mathbf{m_a}&0&1&2&1&4&1&2&1\\\end{matrix}
\end{align}
Notice that $m_a = 1$ if and only if $a$ is a unit.
\end{example}

\subsection{Spans of Matrices and the Howell Matrix Form}\label{sec:howell_matrix_form}
We can define the span of a matrix $B \in R^{m \times n}$ with rows $\mathbf{b}_i$ over a ring $R$ as:
\begin{align}
\Span_R(B) := \{\sum_i a_i \mathbf{b}_i: a_i \in R\} = \{\mathbf{a}B : \mathbf{a} \in R^{m}\} \subseteq R^{n} 
\end{align}
Spans over rings can be considered subgroups of $R^n$ under component-wise addition over $R$ (i.e. $\Span_R(B) = \langle B \rangle_{(R,+)} \leq R^n = \langle I \rangle_{(R,+)}$).

Where $R$ is a field (e.g.  $R = \mathbb{Z}_q$ where $q$ is prime), the RREF gives us a basis of the span and two matrices have the same span if and only if they have the same RREF. The Howell form plays an analogous role for spans over rings.

The Howell matrix form or Howell basis of $B \in R^{m \times n}$ over a ring $R$, denoted $\How_R(B)$, is an $n \times n$ matrix of form:

\begin{align}\How_R(B) = \begin{pmatrix} &a & * &\cdot& * &\cdot\\&\\  & & &b&*&\cdot\\ \\  &  &  & & &c\end{pmatrix}
\end{align}
\ \\
The entries of $\How_R(B)$ are subject to the following constraints:
\begin{itemize}
    \item Diagonal entries  (marked $a, b, c$ above) are minimal associates in $\mathbb{Z}_N$.
    \item If the diagonal entry is zero, the whole row is zero
    \item Entries below the diagonal entries are zero
    \item Entries above the diagonal entries are strictly less than the diagonal entry (marked $\cdot$ above), unless the diagonal entry  is $0$ in which case they are unrestricted (marked $*$ above).
\end{itemize}
The Howell matrix is the unique matrix of the above form which has the \textbf{Howell property}. Let $S_i$ be the subset of $\Span_R(B)$ where the first $i$ entries are zero. The $n\times n$ matrix $H$ has the Howell property if $S_i$ is the span of the last $n-i$ rows of $H$. Two matrices have the same span over $R$ if and only if they have the same Howell matrix. Methods of constructing the Howell matrix form of a given matrix are set out in Refs. ~\cite{howell}, \cite{storjohann1} and ~\cite{buchmann}, and we use a simplified more intuitive algorithm in the XPF package. You can view examples of calculating the Howell matrix form in the \href{https://github.com/m-webster/XPFpackage/blob/main/Examples/A.1_howell_matrix.ipynb}{ linked Jupyter notebook}.

Where $R$ is a field, $\How_R(B) = \RREF_R(B)$. In the special case of binary matrices (i.e. where $R=\mathbb{Z}_2$), the RREF has a particular form. In particular, for any leading index $l$ the entries in column $l$ are strictly less than $1$ and hence are zero. This fact is useful when determining the special form of orbit representatives (see Section~\ref{sec:graph_search}).

\subsection{Solving Linear Equations over Rings}\label{sec:linalgmodN}
In this section, we show how to solve linear equations over rings using the Howell matrix form. Given the matrix  $A \in R^{m \times n}$ and a constant vector $\mathbf{c} \in R^{m }$, we wish to solve for the vector $\mathbf{x} \in R^{n}$ in the linear equation:
\begin{align}
\mathbf{x}A^T + \mathbf{c} = 0 \label{eq:Ax+c}
\end{align}

We calculate the Howell form of the transpose $A^T$ so that:
\begin{align}
T &:= \How_{R}(A^T) = S A^T\label{eq:SAt}
\end{align}
where $T\in R^{m\times m}$ and $S \in R^{m\times n}$. As part of this calculation, we also obtain the Howell basis $K$ for $\Ker_{R}(A)$ such that $KA^T = 0$.
We can solve Eq.~\eqref{eq:Ax+c} if $\mathbf{c}\in \Span_{R}(T)$ so that for some vector $\mathbf{v} \in R^{1\times m}$:
\begin{align}
\mathbf{c} &= \mathbf{v}T
\end{align}
In this case, the solutions for $\mathbf{x}$ are given by:
\begin{align}
\mathbf{x} &=  - \mathbf{v}S + \mathbf{a}K
\end{align}
where $\mathbf{a}$ ranges over all values of $R^{n}$. That $\mathbf{x}$ is a solution to Eq.~\eqref{eq:Ax+c} can easily be verified by substitution. We can also write the solution set as an affine span $\mathbf{x} \in - \mathbf{v}S + \Span_R(K)$ (see Eq.~\eqref{eq:affine_span} below).

\subsection{Intersections of Spans and Affine Spans}\label{sec:affinespan}
When calculating the codewords (Chapter~\ref{chap:codewords}) and logical operators (Chapter~\ref{chap:LO}) of an XP code, we work with affine spans which can be identified with cosets of row spans. To find the logical X operators, we need to determine the intersection of affine spans (see Section~\ref{sec:diagonal_component_L_X}). In this section, we explain how to compute the intersection of two affine spans and introduce the residue function which identifies which coset a vector belongs to.

\subsubsection{Affine Span Definition}
Given an offset $\mathbf{a} \in R^n$ and a matrix $B \in R^{r\times n}$, the \textbf{affine span} is defined as:
\begin{align}
    \mathbf{a} + \Span_{R}(B) :=  \{\mathbf{a}+\mathbf{b}: \mathbf{b} \in \Span_{R}(B) \}\label{eq:affine_span}
\end{align}
Because $\Span_{R}(B)$ is a subgroup of $R^n$, we can also consider affine spans to be cosets (i.e. $ \mathbf{a} + \Span_{R}(B) = \mathbf{a} + \langle B\rangle_{(R,+)}$). The \textbf{residue function} identifies the coset of a vector and can be used to determine if two vectors are in the same affine span. The residue function is defined as:
\begin{align}
    \mathbf{m} &= \res_R(B,\mathbf{a})\label{eq:residue_function} \text{, where:}
    \begin{pmatrix}1&\mathbf{m}\\0&B\end{pmatrix} = \How_R\begin{pmatrix}1&\mathbf{a}\\0&B\end{pmatrix}
\end{align}
A vector is in a span if and only if its residue is zero - i.e. $\mathbf{a} \in \Span_{R}(B) \iff \res_R(B,\mathbf{a}) = 0$. Vectors are in the same coset if and only their residues are the same - i.e. $\mathbf{a} + \Span_{R}(B) = \mathbf{b} + \Span_{R}(B) \iff \res_R(B,\mathbf{a}) = \res_R(B,\mathbf{b})$.
\subsubsection{Algorithm for Intersection of Spans}\label{sec:intersections_of_spans}
Given two matrices $A \in R^{r\times n}$ and $B\in R^{s\times n}$, we can find the intersection of the respective spans ($\Span_{R}(A) \cap \Span_{R}(B)$) as follows:
\begin{enumerate}
    \item Form the $(r+s)\times n$ matrix $C = \begin{pmatrix}A\\B\end{pmatrix}$
    \item Calculate the Howell basis $\begin{pmatrix}K_A&K_B\end{pmatrix}$ of $\Ker(C^T)$ where $K_A$ is $(r+s)\times r$ and $K_B$ is $(r+s)\times s$ so that $K_A A + K_B B = 0$.
    \item Let $D = K_A A = -K_B B$.
    \item The intersection is $\Span_{R}(A) \cap \Span_{R}(B) = \Span_R(\How_{R}(D))$
\end{enumerate}

\subsubsection{Algorithm for Intersection of Affine Spans}\label{sec:affine_span_intersection}
The intersection of two affine spans is either empty or an affine span. Assume we are given two affine spans $\mathbf{a} + \Span_{R}(A)$ and $\mathbf{b} + \Span_{R}(B)$. There exists a vector $\mathbf{c}$ in the intersection if and only if we can find vectors $\mathbf{u} \in R^r$ and $\mathbf{v} \in R^s$  such that:
\begin{align}
    \mathbf{c} = \mathbf{a}+\mathbf{u}A = \mathbf{b}+\mathbf{v}B
\end{align}
This is possible only when $\mathbf{a}-\mathbf{b} \in \Span_{R}(A) \cup \Span_{R}(B)$ which is true if and only if: 
\begin{align}
\How_{R}\begin{pmatrix}1&\mathbf{a}-\mathbf{b}\\0&A\\0&B\end{pmatrix} = \begin{pmatrix}1&\mathbf{0}\\0&A\\0&B\end{pmatrix}
\end{align}
In this case, we set $\mathbf{c} = (\mathbf{b} +\res_{R}(A,\mathbf{a}-\mathbf{b})$ and the intersection is given by:
\begin{align}
    (\mathbf{a} + \Span_{R}(A)\cap ( \mathbf{b} + \Span_{R}(B)) = \mathbf{c} + \Span_{R}(A)\cap \Span_{R}(B)
\end{align}

\section{Canonical Generator Algorithm - Proof of Result}\label{sec:canonical_generator_proof}\label{sec:canonical_generator_algorithm_outline}

In this appendix, we provide a proof of Proposition~\ref{prop:canonical_generators}. This is an important result which states that we can calculate a set of canonical generators of unique form for any XP group. This allows us to  determine whether two sets of XP operators generate the same group and also identify a set of generators for the diagonal subgroup.  The proof of the proposition is constructive and relies on the following algorithm:
\\
\\
\textbf{Canonical Generator Algorithm}
 \newline
The algorithm for producing the canonical generators from an arbitrary set of XP operators $\mathbf{G}$ is:
\begin{enumerate}
    \item \textbf{Simplify X Components:} let $G_X$ be the binary matrix whose rows are the X components of the operators in $\mathbf{G}$. We can put $G_X$ into RREF by using row operations over $\mathbb{Z}_2$. These row operations correspond to group operations between elements of $\mathbf{G}$ and we update $\mathbf{G}$ accordingly.
    \item \textbf{Split $\mathbf{G}$ into diagonal and non-diagonal operators:} let $\mathbf{S}_X$ be the non-diagonal operators and $\mathbf{S}_Z$ be the diagonal operators. 
    \item \textbf{Add squares and commutators of $\mathbf{S}_X$:} squares and commutators of operators in $\mathbf{S}_X$ are diagonal - add these to $\mathbf{S}_Z$.
    \item \textbf{Add commutators between $\mathbf{S}_Z$ and $\mathbf{S}_X$:} add to $\mathbf{S}_Z$ all possible commutators between elements of $\mathbf{S}_Z$ and elements of $\mathbf{S}_X$. Where $N=2^t$ is a power of $2$, we do this step $t-1$ times. 
    \item \textbf{Simplify $\mathbf{S}_Z$:} Let $S_{Zp}$ be matrix whose rows are the image of  $\mathbf{S}_Z$ under the Zp map of  Section~\ref{sec:xp_group_structure} - i.e. $\text{Zp}(XP_N(p|\mathbf{0}|\mathbf{z})) = (2\mathbf{z}|p)$ . The final set of diagonal generators are the XP operators corresponding to the rows of $H_{Zp} := \How_{\mathbb{Z}_{2N}}(S_{Zp})$ i.e. $\mathbf{S}_Z = \text{Zp}^{-1}(H_{Zp})$.
    \item \textbf{Simplify Z Components of $\mathbf{S}_X$:} Let $A = XP_N(p|\mathbf{x}|\mathbf{z}) \in \mathbf{S}_X$ and let \linebreak $(2\mathbf{z}'|p') = \text{Res}_{\mathbb{Z}_{2N}}(H_{Zp}, (2\mathbf{z}|p))$ (see Eq.~\eqref{eq:residue_function}). Replace $A$ with $A' = XP(p'|\mathbf{x}|\mathbf{z}')$.
\end{enumerate}
\  \\
We restate  Proposition~\ref{prop:canonical_generators} here for clarity:
\ \\
\ \\
\textbf{Proposition~\ref{prop:canonical_generators} }(Canonical Generators of an XP Group)
\newline
For any set of XP Operators $\mathbf{G} = \{G_1, \dots, G_m\}$, there exists a unique set of diagonal operators $\mathbf{S}_Z := \{B_j : 0 \le j < s\}$ and non-diagonal operators $\mathbf{S}_X := \{A_i : 0 \le i < r\}$ with the following form:
\begin{enumerate}
\item Let $S_X$ be the binary matrix formed from the X-components of the $\mathbf{S}_X$.  $S_X$ is in Reduced Row Echelon Form (RREF).
\item Let $S_{Zp}$ be matrix whose rows are the image of  $\mathbf{S}_Z$ under the Zp map of Section~\ref{sec:xp_group_structure} (i.e. $\text{Zp}(XP_N(p|\mathbf{0}|\mathbf{z})) = (2\mathbf{z}|p)$). The matrix $S_{Zp}$ is in Howell Matrix Form (see Appendix~\ref{app:linalg}).
\item For $XP_N(p|\mathbf{x}|\mathbf{z}) \in \mathbf{S}_X$, $\begin{pmatrix}1&(2\mathbf{z}|p)\\0&S_{Zp}\end{pmatrix}$ is in Howell Matrix Form.
\end{enumerate}
The following properties hold for the canonical generators:

\paragraph{Property 1:} All group elements $G \in \langle\mathbf{G}\rangle$ can be expressed in the generator product form of Eq.~\eqref{eq:generator_product} $G = \mathbf{S}_X^\mathbf{a} \mathbf{S}_Z^\mathbf{b}$ where $\mathbf{a} \in \mathbb{Z}_2^{|\mathbf{S}_X|}$, $\mathbf{b}  \in \mathbb{Z}_N^{|\mathbf{S}_Z|}$, $\mathbf{S}_X^\mathbf{a} = \prod_{0 \le i < |\mathbf{S}_X|} A_i^{\mathbf{a}[i]}$ and $\mathbf{S}_Z^\mathbf{b} = \prod_{0 \le j < |\mathbf{S}_Z|} B_j^{\mathbf{b}[j]}$
\paragraph{Property 2:} Two sets of XP operators of precision $N$ generate the same group if and only if they have the same canonical generators.

\begin{proof}
Steps 1-2 of the algorithm create a list of non-diagonal operators $\mathbf{S}_X$ whose X-components are in RREF, plus diagonal operators $\mathbf{S}_Z$. We claim that after these operations, $\langle \mathbf{S}_X,\mathbf{S}_Z\rangle = \langle \mathbf{G}\rangle$. Over $\mathbb{Z}_2$, the row operations to convert the matrix $G_X$ into RREF involve either:
\begin{enumerate}
\item Swapping the order of rows; or
\item Adding rows i.e. $\mathbf{r}_i' = (\mathbf{r}_i + \mathbf{r}_j) \mod 2$
\end{enumerate}
The row operations can be translated into group operations on $\mathbf{G}$ as follows:
\begin{enumerate}
\item Swapping the order of generators; or
\item Replacing a generator by a product of generators i.e. $G_i' = G_i G_j$.
\end{enumerate}
For case 2, we need to check if $G_i$ is still in the group after the row operations. Because $G_j$ is unchanged, $G_j^{-1}$ remains in the group so we have $G_i = (G_i G_j) G_j^{-1} = G_i' G_j^{-1}$.

Steps 3-4 of the algorithm ensure that all possible squares and commutators of the generators are added to the list of diagonal operators $\mathbf{S}_Z$. Now we show that where $N = 2^t$, that $t-1$ rounds of adding commutators is sufficient. When we calculate the commutator of operators $A_1, A_2$, the resulting degree (see Section~\ref{sec:degree+fundamental_phase}) is at most half the degrees of $A_1, A_2$: recall the \textbf{COMM} rule of Section~\ref{sec:algebraic_identities}:
\begin{align}A_1 A_2 A_1^{-1}A_2^{-1} = D_N(2\mathbf{x_1}\mathbf{z_2} - 2\mathbf{x_2}\mathbf{z_1}+4\mathbf{x_1}\mathbf{x_2}\mathbf{z_1}-4\mathbf{x_1}\mathbf{x_2}\mathbf{z_2})\end{align}
As $X^2 = P^N = I$, the maximum degree of a precision $N$ operator is $N = 2^t$. Hence, after $t-1$ rounds of taking commutators, a further round of commutators yields operators of degree 1 (i.e. phase multiples of $I$). Note that all of the operators added to $\mathbf{S}_Z$ are in $\langle \mathbf{G} \rangle$ so there is no change to $\langle \mathbf{S}_X,\mathbf{S}_Z\rangle$ in this step.

Step 5 ensures that $S_{Zp}$, the matrix formed from the phase and Z-components of $\mathbf{S}_Z$ under the $\text{Zp}$ map is in Howell matrix form. In Section~\ref{sec:xp_group_structure}, we showed that the Zp map is a group homomorphism so group generators in $\mathbb{Z}_{2N}^{n+1}$, i.e. the rows of the Howell matrix, correspond to diagonal group generators in $\mathcal{XP}_{N,n}$, i.e. $\mathbf{S}_Z$ so there is no change to the group generated by $\mathbf{S}_Z$ in this step.

In Step 6, the residue function of Eq.~\eqref{eq:residue_function} ensures that the Z-components of the non-diagonal canonical generators are of the correct form. The adjustment corresponds to multiplication of elements in $\langle \mathbf{S}_Z \rangle$ so we are assured that the final set of generators meets the invariant $\langle \mathbf{G} \rangle = \langle \mathbf{S}_X, \mathbf{S}_Z\rangle$.

To prove Property 1, we need to show that any element $G \in \langle \mathbf{G} \rangle$ can be expressed n the generator product form of Eq.~\eqref{eq:generator_product} $G = \mathbf{S}_X^\mathbf{a} \mathbf{S}_Z^\mathbf{b}$. We have already shown that $\langle \mathbf{G} \rangle = \langle \mathbf{S}_X, \mathbf{S}_Z\rangle$. Hence, we can write $G$ as a string of operators from $\mathbf{S}_X, \mathbf{S}_Z$. Now assume we have a diagonal operator $B\in \mathbf{S}_Z$ which occurs immediately before a non-diagonal operator $A\in \mathbf{S}_X$. We can write:
\begin{align}
    BA = AB (B^{-1} A^{-1} B A)
\end{align}
The commutator $B^{-1}A^{-1} B A \in \langle \mathbf{S}_Z\rangle$, so we can always move diagonal operators to the right of non-diagonal operators.

Now assume we have two non-diagonal operators $A_j, A_i \in \mathbf{S}_X$ which occur immediately next to each other in the string, but out of order (i.e. $i < j$). We can move $A_j$ to the right of $A_i$ by using commutators as follows:
\begin{align}
    A_j A_i =  A_i A_j (A_j^{-1}A_i^{-1}A_j A_i)
\end{align}
The commutator $(A_j^{-1}A_i^{-1}A_j A_i) \in \langle \mathbf{S}_Z\rangle$, so we can ensure $A_i, A_j$ occur in the correct order with a diagonal operator to the right.

Reordering the non-diagonal operators may result in squares or higher powers of non-diagonal operators arising in the string. As $A^2 \in \langle \mathbf{S}_Z\rangle$  for any $A \in  \mathbf{S}_X$, any power of $A$ can be written as $A^q = A B$ for $q$ odd or $A^q = B$ for some $B \in \langle \mathbf{S}_Z\rangle$. Hence, $G$ can be written with powers of $A$ in $\mathbb{Z}_2$. Accordingly, any $G \in \langle \mathbf{G} \rangle$ can be written as a string of the form in Property 1.

To establish Property 2, note that the Howell matrix form and RREF are unique. Thus, for any operators which generate the same group, the canonical form will be the same.
\end{proof}

\section{Coset and Orbit Structure of Codewords}
In Chapter~\ref{chap:codewords}, we gave an algorithm for generating the codewords by applying the orbit operator to the orbit representatives. In this appendix, we provide proofs underlying the algorithm. The results in this appendix  assume we have the canonical generators $\mathbf{S}_X$ and $\mathbf{S}_Z$ (see Section~\ref{sec:canonical_generators}) and the set of binary vectors $E$, which is the Z-support of the simultaneous $+1$ eigenspace of $\mathbf{S}_Z$ (see Section~\ref{sec:code_words_notation}). Our aim is to calculate a basis of the codespace stabilised by $\mathbf{S}_X, \mathbf{S}_Z$. 

In Proposition~\ref{prop:cworbit}, we show that the image under the orbit operator of any $|\mathbf{e}\rangle$ where $\mathbf{e} \in E$ is stabilised by $\mathbf{S}_X, \mathbf{S}_Z$. Let $S_X$ be the binary matrix formed from the X-components of the $\mathbf{S}_X$ and $E_m = \{\res_{\mathbb{Z}_2}(S_X, \mathbf{e}): \mathbf{e} \in E\}$ be the orbit representatives. In Proposition~\ref{prop:eclosed}, we show that $E$ is closed under addition by elements of the span $\langle S_X \rangle$. In Proposition~\ref{prop:partition}, we show that the cosets $\mathbf{m}_i + \langle S_X \rangle, \mathbf{m}_i \in E_m$ partition $E$. In Proposition~\ref{prop:kappa_are_basis}, we show that the image of $E_m$ under the orbit operator forms a basis of the codespace. Finally, we show that the orbit representatives have a unique form, which is used in the graph search algorithm of Section~\ref{sec:graph_search}.

\begin{proposition}[Codewords as Orbits]\label{prop:cworbit}
Given canonical generators for a code $\mathbf{S}_X$ and $\mathbf{S}_Z$, let $E = \{\mathbf{e}: \mathbf{e} \in \mathbb{Z}_2^n, B|\mathbf{e}\rangle = |\mathbf{e}\rangle, \forall B \in \langle \mathbf{S}_Z \rangle \}$ be the Z-support of the simultaneous $+1$ eigenspace of $\mathbf{S}_Z$. 

Then $O_{\mathbf{S}_X} |\mathbf{e}\rangle$ is stabilised by all elements of  $\langle \mathbf{S}_X, \mathbf{S}_Z\rangle$, for any $|\mathbf{e}\rangle \in E$.
\end{proposition}

\begin{proof}
It is sufficient to prove this for the generators $A_i \in \mathbf{S}_X$ and $B_j \in \mathbf{S}_Z$. Let $B_j \in \mathbf{S}_Z$ be a diagonal generator. Then we have:
\begin{align}
B_j O_{\mathbf{S}_X} |\mathbf{e}\rangle &= \sum_{\mathbf{v} \in \mathbb{Z}_2^r} B_j \mathbf{S}_X ^\mathbf{v} |\mathbf{e}\rangle\\
&= \sum_{\mathbf{v} \in \mathbb{Z}_2^r} \mathbf{S}_X ^\mathbf{v} D_\mathbf{v} |\mathbf{e}\rangle \text{ for }D_\mathbf{v} = (\mathbf{S}_X ^\mathbf{v})^{-1}B_j \mathbf{S}_X ^\mathbf{v}\\
&= \sum_{\mathbf{v} \in \mathbb{Z}_2^r} \mathbf{S}_X ^\mathbf{v}  |\mathbf{e}\rangle \text{ because } D_\mathbf{v} \text{ is diagonal}, D_\mathbf{v}\in \langle\mathbf{S}_Z\rangle \text{ and so } D_\mathbf{v}|\mathbf{e}\rangle = |\mathbf{e}\rangle \\
&= O_{\mathbf{S}_X} |\mathbf{e}\rangle
\end{align}
Let $A_i \in \mathbf{S}_X$ be a non-diagonal operator. 
\begin{align} 
A_i O_{\mathbf{S}_X} |\mathbf{e}\rangle &= \sum_{\mathbf{v} \in \mathbb{Z}_2^r} A_i \mathbf{S}_X ^\mathbf{v} |\mathbf{e}\rangle
\end{align}
We can move $A_i$ to the right by applying commutators. We can then move the commutators to the right. Let $\mathbf{i}$ be the length $r$ binary vector which is all zero, apart from component $i$ which is 1 and let $\mathbf{v}' = \mathbf{v}\oplus \mathbf{i}$. As all commutators are diagonal and so are in $\langle \mathbf{S}_Z\rangle$, we can write:
\begin{align}
A_i O_{\mathbf{S}_X} |\mathbf{e}\rangle &= \sum_{\mathbf{v} \in \mathbb{Z}_2^r} \mathbf{S}_X ^{\mathbf{v}'}  D_\mathbf{v} |\mathbf{e}\rangle, \;\; \exists D_\mathbf{v} \in \langle\mathbf{S}_Z\rangle\\
&= \sum_{\mathbf{v'} \in \mathbb{Z}_2^r} \mathbf{S}_X ^\mathbf{v'}  |\mathbf{e}\rangle \text{ since } D_\mathbf{v}|\mathbf{e}\rangle = |\mathbf{e}\rangle \\
&= O_{\mathbf{S}_X} |\mathbf{e}\rangle
\end{align}
\end{proof}

\begin{proposition}[$E$ closed under addition by $\langle S_X\rangle$]\label{prop:eclosed}
If $\mathbf{e} \in E$, then $\mathbf{e} \oplus \mathbf{x} \in E$ for all $\mathbf{x} \in \langle S_X\rangle$.
\end{proposition}
\begin{proof}
Let $\mathbf{x} := \mathbf{u}S_X \mod 2 \in \langle S_X\rangle$ and $C := \mathbf{S}_X^\mathbf{u}$. Then $C|\mathbf{e}\rangle = \omega^p|\mathbf{e} \oplus \mathbf{x}\rangle$ for some $p \in \mathbb{Z}_{2N}$. Let $B\in \mathbf{S}_Z$ and $D = B^{-1} C^{-1} B C \in \langle \mathbf{S}_Z\rangle$. Then because $B,D \in \langle \mathbf{S}_Z\rangle$,  $BD|\mathbf{e}\rangle =  |\mathbf{e}\rangle$ and so:
\begin{align}
B (C|\mathbf{e}\rangle) &= C B D |\mathbf{e}\rangle =  C|\mathbf{e}\rangle
\end{align}

Hence, $C|\mathbf{e}\rangle = \omega^p|\mathbf{e} \oplus \mathbf{x}\rangle$ is in the simultaneous $+1$ eigenspace of the $\mathbf{S}_Z$ and so $\mathbf{e} \oplus \mathbf{x} \in E$.
\end{proof}
\begin{proposition}[Cosets of $E_m$ partition $E$]\label{prop:partition}
The cosets $\mathbf{m}_i +\langle S_X\rangle$ partition $E$ i.e.:
\begin{align}
&E = \bigcup_i (\mathbf{m}_i +\langle S_X\rangle)&\\
&(\mathbf{m}_i +\langle S_X\rangle) \cap (\mathbf{m}_j +\langle S_X\rangle) = \emptyset, \forall j \ne i&
\end{align}
\end{proposition}
\begin{proof}
By Proposition~\ref{prop:eclosed}, for all $\mathbf{m}_i \in E_m, \mathbf{m}_i +\langle S_X\rangle \subset E$ hence $\bigcup_i (\mathbf{m}_i +\langle S_X\rangle) \subset E$. But by the definition of orbit representatives for any $\mathbf{e} \in E$ we can calculate $\text{Res}_{\mathbf{Z}_2}(S_X, \mathbf{e}) \in E_m$ hence $\mathbf{e} \in \mathbf{m}_i +\langle S_X\rangle$ for some $\mathbf{m}_i \in E_m$. 

The fact that cosets of a subgroup partition the group is a well-known result from group theory. Hence if $(\mathbf{m}_i +\langle S_X\rangle) \cap (\mathbf{m}_j +\langle S_X\rangle) \ne \emptyset$ then $\mathbf{m}_i \in \mathbf{m}_j +\langle S_X\rangle$. But then because the RREF is unique $\mathbf{m}_i = \text{Res}(S_X,\mathbf{m}_i) = \mathbf{m}_j$ so $i = j$.
\end{proof}

\begin{proposition}[The $|\kappa_i\rangle$ are a basis of $\mathcal{C}$]\label{prop:kappa_are_basis}
Let $|\kappa_i\rangle = O_{\mathbf{S}_X}|\mathbf{m}_i\rangle, \mathbf{m}_i \in E_m$. The $|\kappa_i\rangle$ are a basis of the codespace $\mathcal{C}$ stabilised by the canonical generators $\mathbf{S}_X, \mathbf{S}_Z$. 
\end{proposition}
\begin{proof}
First, we show that the $|\kappa_i\rangle$ are independent. The Z-support of the codeword $|\kappa_i\rangle$ is the coset $\ZSupp(|\kappa_i\rangle) = \mathbf{m}_i +\langle S_X \rangle$ and so by Proposition~\ref{prop:partition} the Z-support of the codewords partition $E$. Hence $|\kappa_i\rangle$ are independent. 

Next, we show that the $|\kappa_i\rangle$ span the codespace $\mathcal{C}$. Let $|\psi\rangle \in \mathcal{C}$. Then $|\psi\rangle$ is stabilised by all elements $B \in \mathbf{S}_Z$. Let $\mathbf{e} \in \ZSupp(|\psi\rangle)$ then because $B$ is diagonal, $B|\mathbf{e}\rangle = \omega^p|\mathbf{e}\rangle, \exists p \in \mathbb{Z}_{2N}$ and this implies $B|\mathbf{e} \rangle = |\mathbf{e}\rangle, \forall \mathbf{e} \in \ZSupp(|\psi\rangle)$. Hence $\ZSupp(|\psi\rangle) \subset E$.

Now let $\lambda_i$ be the coefficients of $|\mathbf{m}_i\rangle$ in $|\psi\rangle$ so that:
\begin{align}
    \lambda_i &= \langle \mathbf{m}_i|\psi\rangle \in \mathbb{C}, \mathbf{m}_i \in E_m. 
\end{align}

For $\mathbf{e} \in \ZSupp(|\psi\rangle)$, we now show that the coefficient $\langle \mathbf{e}|\psi\rangle$ is determined by the $\lambda_i$. Because $\mathbf{e} \in E$, there exists unique $i, \mathbf{u}\in \mathbb{Z}_2^r$ such that $\mathbf{e} = (\mathbf{m}_i + \mathbf{u}S_X).$ The operator $\mathbf{S}_X^\mathbf{u} \in \langle \mathbf{S}_X\rangle$ and so $\mathbf{S}_X^\mathbf{u}|\psi\rangle = |\psi\rangle$.  The action of  $\mathbf{S}_X^\mathbf{u}$ on $|\mathbf{m}_i\rangle$ is given by $\mathbf{S}_X^\mathbf{u}|\mathbf{m}_i\rangle = \omega^p|\mathbf{e}\rangle, \exists p \in \mathbb{Z}_{2N}$. Hence the coefficient of $|\mathbf{e}\rangle$ in $|\psi\rangle$ is given by $\langle \mathbf{e}|\psi\rangle = \lambda_i \omega^p = \lambda_i \langle \mathbf{e}|\kappa_i\rangle$. Hence:
\begin{align}
|\psi\rangle &= \sum_{i}\lambda_i|\kappa_i\rangle
\end{align}
Hence any $|\psi\rangle \in \mathcal{C}$ can be written as a linear combination of the $|\kappa_i\rangle$ and the result follows.
\end{proof}

The orbit representatives have a form which is unique for each coset, which proves useful in the graph search algorithm of Section~\ref{sec:graph_search}.

\begin{proposition}\label{prop:orbit_rep_form}
Let $S_X$ be an $r \times n$ binary matrix in RREF. Let $\mathbf{e}$ be a binary vector of length $n$ and let $\mathbf{m} = \res_{\mathbb{Z}_2}(S_X,\mathbf{e})$. Then $\mathbf{m}$ is the unique element of the coset $\mathbf{e} + \langle S_X\rangle$ for which $\mathbf{m}[l] = 0$ for all leading indices $l$ of $S_X$.
\end{proposition}
\begin{proof}
We first show that $\mathbf{m}[l] = 0$ for all leading indices $l$ of $S_X$. By the definition in Eq.~\eqref{eq:residue_function}, the matrix $\begin{pmatrix}1&\mathbf{m}\\\mathbf{0}&S_X\end{pmatrix}$ is a binary matrix in Howell form. Hence, the entries above the leading indices are strictly less than $1$ and so are zero. Therefore $\mathbf{m}[l] =0$ for all leading indices $l$.

Now we show the uniqueness of the property. Assume there exists some binary vector $\mathbf{a}$ in the coset with $\mathbf{a}[l] = 0$ for all leading indices $l$. Then  $\mathbf{b} = \mathbf{a} \oplus \mathbf{m}$ is a vector in $\langle S_X \rangle$ such that $\mathbf{b}[l] =0$ for all leading indices $l$ of $S_X$. The only member of $\langle S_X \rangle$ with this property is $\mathbf{0}$ hence $\mathbf{a} = \mathbf{m}$.
\end{proof}

\section{Proof of Results: Classification of XP Stabiliser States}\label{app:hypergraph}

In this appendix, we provide detailed proofs of the results in Chapter~\ref{chap:hypergraph}. In Section~\ref{app:integer_vectors}, we set out some basic results which are useful for working with integer and binary vectors. In Section~\ref{app:phase_function}, we prove Proposition~\ref{prop:xp_phase_function} regarding the form of the phase function of an XP stabiliser state. Finally in Section~\ref{app:whg2xp}, prove that the algorithm for representing weighted hypergraph states as XP stabiliser states gives the correct result.

\subsection{Operations on Binary and Integer Vectors}\label{app:integer_vectors}
The results in Chapter~\ref{chap:hypergraph} involve operations on binary and integer vectors. This section sets out some basic results for these types of vectors. We use component wise addition and multiplication of vectors and a \textbf{dot product} of vectors over the integers. Given two vectors $\mathbf{a}, \mathbf{b} \in \mathbb{Z}^n$ the dot product is defined as:
\begin{align}
    \mathbf{a}\cdot\mathbf{b} = \sum_{0 \le i < n} \mathbf{a}[i]\mathbf{b}[i] = \sum_{0 \le i < n} (\mathbf{a}\mathbf{b})[i]
\end{align}
Hence, the dot product is the sum of the entries of the component wise product of two vectors. Accordingly, we can write the following rule for dot product over the component wise product:
\begin{align}
\mathbf{a}\cdot \mathbf{b}\mathbf{c} = \sum_{0 \le i < n} (\mathbf{a}(\mathbf{b}\mathbf{c}))[i]= \sum_{0 \le i < n} ((\mathbf{a}\mathbf{b})\mathbf{c})[i] = \mathbf{a} \mathbf{b}\cdot \mathbf{c}\label{eq:dot_prod_over_mul}
\end{align}
We also have a distributive rule for dot product over component wise addition:
\begin{align}
\mathbf{a}\cdot (\mathbf{b} + \mathbf{c})  = \mathbf{a} \cdot \mathbf{b} + \mathbf{a}\cdot \mathbf{c}\label{eq:dot_prod_over_add}
\end{align}
The usual rule for scalar products also applies - for $u \in \mathbb{Z}$:
\begin{align}
    \mathbf{a}\cdot(u\mathbf{b}) = u(\mathbf{a}\cdot\mathbf{b}) = (u\mathbf{a})\cdot\mathbf{b}
\end{align}
The \textbf{weight of a binary vector} can also be thought of as a dot product. Let $\mathbf{1}$ be the all 1 vector of length $n$. Considering the binary vector $\mathbf{a}$ a vector of zeros and ones in $\mathbb{Z}^n$: 
\begin{align}
\text{wt}(\mathbf{a}) = \mathbf{1} \cdot \mathbf{a}
\end{align}
We can look at component wise multiplication and addition modulo 2 of binary vectors in terms of the effect on the support of the vectors. The support of the component wise product of binary vectors is the intersection of the supports:
\begin{align}
    \text{supp}(\mathbf{a} \mathbf{b}) &= \text{supp}(\mathbf{a}) \cap \text{supp}( \mathbf{b}).
\end{align}
The support of the addition of binary vectors modulo 2 is the symmetric difference of the supports:
\begin{align}
    \text{supp}(\mathbf{a} \oplus \mathbf{b}) &= (\text{supp}(\mathbf{a}) \cup \text{supp}( \mathbf{b})) \setminus (\text{supp}(\mathbf{a}) \cap \text{supp}( \mathbf{b})).
\end{align}
We often apply the following identity for binary vectors:
\begin{align}
    \mathbf{a} \oplus \mathbf{b} = \mathbf{a} + \mathbf{b} - 2 \mathbf{a}  \mathbf{b}
\end{align}
where operations on the RHS are over the integers. This generalises to the following identity for binary vectors $\mathbf{x}_j, j \in [0\dots r-1]$ of length $n$:
\begin{align}
\bigoplus_{0 \le i < r}\mathbf{x}_i = \sum_{s \subset [0\dots r-1]}(-2)^{|s|-1}\prod_{j \in s}\mathbf{x}_j\label{eq:venn_diagram}
\end{align}

\subsection{Phase Functions of XP Stabiliser States}\label{app:phase_function}
In this section, we prove Proposition~\ref{prop:xp_phase_function} which classifies the form of the phase functions of XP stabiliser states.

\paragraph{Proposition~\ref{prop:xp_phase_function}} [Phase Functions of XP States]
Let $|\phi\rangle = O_{\mathbf{S}_X}|\mathbf{m}\rangle = \sum_{\mathbf{u}\in \mathbb{Z}_2^r}\mathbf{S}_X^\mathbf{u}|\mathbf{m}\rangle$ be an XP stabiliser state in the canonical form of Eq.~\eqref{eq:xp_canonical} with $r:= |\mathbf{S}_X|$. Let $u_i, 0 \le i < r$ be binary variables such that $u_i := \mathbf{u}[i]$. Then:
\begin{enumerate}[label=(\alph*)]
\item The phase function is of the following form for some vector $\mathbf{q} \in \mathbb{Z}^{2^r}$ indexed by the subsets $s$ of $[0\dots r-1]$: 
\begin{align}
    f(u_0,u_1,\dots,u_{r-1}) = \sum_{s \subset [0\dots r-1]}\mathbf{q}[s] 2^{|s|-1}\prod_{j \in s}u_j. 
\end{align} 
\item For $N = 2^t$, the maximum degree of the phase function is $t+1$.
\end{enumerate}
\begin{proof}
Let $\mathbf{S}_X = \{XP_N(p_i|\mathbf{x}_i|\mathbf{z}_i)\}$ and let $\mathbf{s}_i := \bigoplus_{i < j < r}u_j\mathbf{x}_j$. Using Eq.~\eqref{eq:venn_diagram} and the dot product results of Section~\ref{app:integer_vectors}, the phase component of $\mathbf{S}_X^\mathbf{u}|\mathbf{m}\rangle$ can be written as follows:
\begin{align}
    q &= \sum_{0 \le i < r}\Big[u_ip_i + 2u_i\mathbf{z}_i \cdot (\mathbf{m}  \oplus \mathbf{s}_i) \Big]\\
    &= \sum_{0 \le i < r}\Big[u_ip_i + 2u_i\mathbf{z}_i \cdot (\mathbf{m}  + \mathbf{s}_i - 2\mathbf{m}\mathbf{s}_i ) \Big]\\
    &= \sum_{0 \le i < r}\Big[u_i(p_i + 2\mathbf{z}_i \cdot \mathbf{m})  + 2u_i \mathbf{z}_i (1 - 2\mathbf{m})\cdot \mathbf{s}_i \Big]\\
    &= \sum_{0 \le i < r}(p_i + 2\mathbf{z}_i \cdot \mathbf{m})u_i  + \sum_{\substack{0 \le i < r\\ c \subset [i+1\cdots r-1]}}2u_i \mathbf{z}_i (1 - 2\mathbf{m})\cdot (-2)^{|c|-1}\prod_{j\in c}u_j\mathbf{x}_j\\
    &= \sum_{0 \le i < r}(p_i + 2\mathbf{z}_i \cdot \mathbf{m})u_i  + \sum_{\substack{0 \le i < r\\ c \subset [i+1\cdots r-1]}}\Big[ (-1)^{|c|-1} \mathbf{z}_i (1 - 2\mathbf{m})\cdot \prod_{j\in c}\mathbf{x}_j\Big]2^{|c|} u_i\prod_{j\in c}u_j\label{eq:c_i_0}
\end{align}
The first term in the above equation is linear in $u_i$ and setting $s := c \cup \{i\}$, the second term is also of the required form so part (a) follows.

Now let $N$ be a power of $2$ such that  $N = 2^t$. As $\omega^{2N} = \omega^{2^{t+1}} = 1$, any terms in the phase function with degree $d-1 \ge t+1$ or $d> t+1$ have coefficients which are multiples of $2N$ and hence generate trivial phases. Hence, the maximum degree of the phase function is $t+1$ and part (b) follows.
\end{proof}

\subsection{Representing Weighted Hypergraph States as XP Stabiliser States}\label{app:whg2xp}
In this section, we show that the algorithm of Section~\ref{sec:whg2xp} for representing weighted hypergraph states as XP stabiliser states produces the required results. The embedding operator of Eq.~\eqref{eq:embedding_operator} is defined in terms of $M^r_m$, which is the binary matrix whose columns are the bit strings of length $r$ of weight between $1$ and $m$ inclusive. This construction is similar to that used in classical simplex codes and Reed-Muller codes~\cite{coding_theory}.

In the following proposition, we break $M^r_r$ into blocks $W^r_k$ where the columns all have weight $k$ and calculate the weight of the product of $t$ rows of $W^r_k$. 

\begin{proposition}[Weight of Vector Products]\label{prop:sums_products_mod_2}
Let $W^r_k$ be a matrix whose columns are the bit strings of length $r$ with weight $k$. Let $\mathbf{p}_t = \prod_{0 \le i < t}\mathbf{x}_i$ be the product of the first $t$ rows of $W^r_k$  where $t \le r$. Then:
\begin{align}
    \text{wt}(\mathbf{p}_t) = \begin{cases}0: t > k\\
\binom{r-t}{k-t}: t \le k\end{cases}
\end{align}
\end{proposition}
\begin{proof}
We can think of $\text{wt}(\prod_{0 \le i < t}\mathbf{x}_i)$ as the number of columns $\mathbf{u}$ of $W^r_k$ such that $\prod_{0 \le i < t}\mathbf{u}[i] = 1$.  If $t > k$ then $\prod_{0 \le i < t}\mathbf{u}[i] = 0$ for all columns because $\text{wt}(\mathbf{u}) = k < t$ and the product includes at least one zero.

If $t \le k$ then the first $t$ entries in $\mathbf{u}$ must all be one, implying that the last $r-t$ entries of $\mathbf{u}$ must include $k-t$ values of 1. There are $\binom{r-t}{k-t}$ ways of constructing bit strings of weight $k$ which have the first $t$ entries equal to 1. 
\end{proof}

We now consider dot products with the \textbf{alternating vector} $\mathbf{a}$ of Eq.~\eqref{eq:alternating_vector} which is 1 when the weight of the corresponding column of $M^r_r$ is even and $-1$ when the weight is odd. The vector $\mathbf{a}$ will be used to construct the Z-component of the non-diagonal stabilisers of the XP code.

\begin{proposition}[Dot Product with Alternating Vector]\label{prop:dot_prod_alt}
Let $\mathbf{x}_i$ be the $i$th row of $M^r_r$ and let  $\mathbf{p}_t = \prod_{0 \le j < t}\mathbf{x}_i$ for $1 \le t \le r$. Let $\mathbf{a}$ be the vector such that $\mathbf{a}[j] = (-1)^{\text{wt}(\mathbf{u}_j)}$ where $\mathbf{u}_j$ is the $j$th column of $M^r_r$. Then $\mathbf{a}\cdot \mathbf{p}_t = (-1)^r$ if $t=r$ and $0$ otherwise.
\end{proposition}
\begin{proof}
From Proposition~\ref{prop:sums_products_mod_2}:
\begin{align}
    \mathbf{a}\cdot \mathbf{p}_i &= \sum_{t \le k \le r}(-1)^k\binom{r-t}{k-t}\\
    &= \sum_{0 \le j \le r-t}(-1)^{t+j}\binom{r-t}{j}\,.
\end{align}
If $t = r$ then $\mathbf{a}\cdot \mathbf{p}_i = (-1)^r$. If $t < r$ then $\mathbf{a}\cdot \mathbf{p}_i = (-1)^t(1-1)^{r-t} = 0$.
\end{proof}

We now show that the alternating vector $\mathbf{a}$ allows us to construct an XP stabiliser state which has a phase function corresponding to a generalised controlled phase operator.
\begin{proposition}[Weighted Hypergraph States]\label{prop:whg2xp}
Let $\mathbf{x}_i$ be the $i$th row of $M^r_r$ and let  $A_i = XP_N(0|\mathbf{x}_i|\mathbf{a} \mathbf{x}_i)$ for $N > 2^r$ and $\mathbf{S}_X = \{A_i : 0 \le i < r\}$ and $\mathbf{a}$ as defined in Eq.~\eqref{eq:alternating_vector}. Let $u_i$ be the  variable representing the value of $\mathbf{u}[i]$ for the binary vector $\mathbf{u}$ of length $r$. Then:
\begin{enumerate}[label=(\alph*)]
    \item The operators $A_i, A_j$ commute; and
    \item The phase component of $\mathbf{S}_X^\mathbf{u}|\mathbf{0}\rangle$ is $p = 2^{r-1}\prod_{0 \le i < r}u_i$.
\end{enumerate}
\end{proposition}
\begin{proof}

(a) Using the COMM rule of Table~\ref{tab:algebraic_identities}, the group commutator of $A_i$ and $A_j$ is:
\begin{align}
A_i A_j A_i^{-1}A_j^{-1} &= D(2\mathbf{x}_i\mathbf{z}_j - 2\mathbf{x}_j\mathbf{z}_i + 4\mathbf{x}_i\mathbf{x}_j\mathbf{z}_i - 4\mathbf{x}_i\mathbf{x}_j\mathbf{z}_j)\\
&= D(2\mathbf{x}_i\mathbf{a}\mathbf{x}_j - 2\mathbf{x}_j\mathbf{a}\mathbf{x}_i + 4\mathbf{x}_i\mathbf{x}_j\mathbf{a}\mathbf{x}_i - 4\mathbf{x}_i\mathbf{x}_j\mathbf{a}\mathbf{x}_j)\\
&= D(2\mathbf{a}\mathbf{x}_i\mathbf{x}_j - 2\mathbf{a}\mathbf{x}_i\mathbf{x}_j + 4\mathbf{a}\mathbf{x}_i\mathbf{x}_j - 4\mathbf{a}\mathbf{x}_i\mathbf{x}_j)\\
&=D(\mathbf{0}) = I
\end{align}

(b) Applying Eq.~\eqref{eq:c_i_0} and noting that the phase components of the $A_i$ are all trivial, the phase component of $\mathbf{S}_X^\mathbf{u}|\mathbf{0}\rangle$ is:
\begin{align}
    p &= \sum_{\substack{0 \le i < r - 1\\s \subset[i+1\dots r-1]}}2^{|s|}(-1)^{|s|+1}\Big(\mathbf{a}\mathbf{x}_i \cdot (\prod_{j\in s}\mathbf{x}_j)\Big)u_i\prod_{j \in s}u_j\\
    &= \sum_{s \subset[0\dots r-1]}2^{|s|-1}(-1)^{|s|}\mathbf{a} \cdot (\prod_{j \in s}\mathbf{x}_j)\prod_{j \in s}u_j\,.
\end{align}
Applying  Proposition~\ref{prop:dot_prod_alt}, $\mathbf{a} \cdot (\prod_{j \in s}\mathbf{x}_j) = 0$ if $|s| < r$ and $(-1)^r$ otherwise. Hence, $p = 2^{r-1}\prod_{0 \le j < r}u_j$ as required.
\end{proof}

We now show how to optimise the embedding operator to reduce the number of qubits required to represent the action of generalised controlled phase operators $CP(p/q,\mathbf{v})$ where $p/q = 1/2$:
\begin{proposition}[Weighted Hypergraph States - Optimised Version]\label{prop:whg2xp_optimised}

Let the state $|\psi\rangle$, the precision $N$, the variables $u_i$ and the operators $A_i$ be as defined in Proposition $\ref{prop:whg2xp}$. Let $C := XP_N(0|\mathbf{0}|\mathbf{a})$ where $\mathbf{a}$ is the alternating vector as defined in Eq.~\eqref{eq:alternating_vector}. Let $B_i = A_i C^{-1}$  and $\mathbf{S}_X = \{B_i : 0 \le i < r\}$. Then:
\begin{enumerate}[label=(\alph*)]
    \item The group commutator of the operators $B_i, B_j$ fixes elements of the Z-support of $|\psi\rangle$; and
    \item The phase component of $\mathbf{S}_X^\mathbf{u}|\mathbf{0}\rangle$ is $p = 2^{r-1}\prod_{0 \le i < r}u_i$.
\end{enumerate}

\end{proposition}
\begin{proof}
Calculating the group commutator between two of the elements of $\mathbf{S}_X$:
\begin{align}
 B_i B_j B_i^{-1} B_j^{-1} &= D_N(2\mathbf{x}_i\mathbf{z}_j - 2\mathbf{x}_j\mathbf{z}_i + 4\mathbf{x}_i\mathbf{x}_j\mathbf{z}_i - 4\mathbf{x}_i\mathbf{x}_j\mathbf{z}_j)\\
&= D_N(2\mathbf{a}(\mathbf{x}_i(\mathbf{x}_j-1) - \mathbf{x}_j(\mathbf{x}_i-1) + 2\mathbf{x}_i\mathbf{x}_j(\mathbf{x}_i-1) - 2\mathbf{x}_i\mathbf{x}_j(\mathbf{x}_j-1))\\
&= D_N(2\mathbf{a}(\mathbf{x}_i\mathbf{x}_j - \mathbf{x}_i-\mathbf{x}_i\mathbf{x}_j + \mathbf{x}_j+ 2\mathbf{x}_i\mathbf{x}_j - 2\mathbf{x}_i\mathbf{x}_j- 2\mathbf{x}_i\mathbf{x}_j+ 2\mathbf{x}_i\mathbf{x}_j ))\\
&=D_N(2\mathbf{a}(\mathbf{x}_j-\mathbf{x}_i)) =XP_N(2\mathbf{a}\cdot(\mathbf{x}_j-\mathbf{x}_i)|\mathbf{0}|2\mathbf{a}(\mathbf{x}_i-\mathbf{x}_j))
\end{align}
We now show that $D(2\mathbf{a}(\mathbf{x}_j-\mathbf{x}_i))$ fixes all elements of the Z-support of $|\psi\rangle$. Let $\mathbf{e}_u := \mathbf{u}S_X + \mathbf{m}  \in \ZSupp(|\psi\rangle)$. Then using the notation of Proposition~\ref{prop:xp_phase_function} and Eq.~\eqref{eq:venn_diagram}, the phase applied by $D_N(2\mathbf{a}(\mathbf{x}_j-\mathbf{x}_i))$ to $|\mathbf{e}_u\rangle$ can be written
\begin{align}
q &= 2\mathbf{a}\cdot(\mathbf{x}_j-\mathbf{x}_i) + 4\mathbf{a}(\mathbf{x}_i-\mathbf{x}_j)\cdot \Big[\bigoplus_{0 \le k < r}u_k\mathbf{x}_k\Big]\\
&= 2\mathbf{a}\cdot(\mathbf{x}_j-\mathbf{x}_i) + 4\mathbf{a}(\mathbf{x}_i-\mathbf{x}_j)\cdot \Big[\sum_{s \subset [0\dots r-1]}(-2)^{|s|-1} \prod_{k\in s}u_k\mathbf{x}_k\Big]\\
&= 2\mathbf{a}\cdot(\mathbf{x}_j-\mathbf{x}_i) + \mathbf{a}\cdot \Big[\sum_{s \subset [0\dots r-1]}(-2)^{|s|+1} (\mathbf{x}_i-\mathbf{x}_j)\prod_{k\in s}u_k\mathbf{x}_k\Big]
\end{align}
By Proposition~\ref{prop:dot_prod_alt}, if $r > 1$ then $\mathbf{a}\cdot \mathbf{x}_j = \mathbf{a}\cdot \mathbf{x}_j = 0$. Similarly, if $|s \cup \{i\}| < r$ then $\mathbf{a} \cdot \mathbf{x}_i\prod_{k\in s}\mathbf{x}_k =0$. If $|s \cup \{i\}| = r$, then a factor of $2^r = 0 \mod 2N$ always occurs - hence the operator applies a trivial phase and part (a) follows.

Due to Proposition~\ref{prop:dot_prod_alt}, we can multiply the generators $A_i$ by powers of $C$ without changing the phase function of the state $|\psi\rangle$, so part (b) follows.
\end{proof}

We now show how to apply Proposition~\ref{prop:whg2xp_optimised} to reduce the number of qubits required to represent  $|\phi\rangle = CP(1/2,\mathbf{1})|+\rangle^{\otimes r}$ as an XP stabiliser state. We set $N = 2^{r-1}$ and let $\mathbf{x}_i$ be the $i$th row of $M^r_r$. We define the non-diagonal stabiliser generators $\mathbf{S}_X = \{B_i\}$ where $B_i = XP_N(0|\mathbf{x}_i|\mathbf{a}(\mathbf{x}_i - 1))$ - these are the same operators as in Proposition~\ref{prop:whg2xp_optimised}. The phase function of $|\psi\rangle = O_{\mathbf{S}_X}|\mathbf{0}\rangle$ is $f(u_0,\dots,u_{r-1}) = 2^{r-1}\prod_{0 \le i < r}u_i$. This imparts a phase of $-1 = \omega^N$ when $\mathbf{u} = \mathbf{1}$. Define the $\mathbf{S}_Z$ as in Eq.~\eqref{eq:whg_sz}. By part (a) of Proposition~\ref{prop:whg2xp_optimised}, the elements of $\mathbf{S}_X$ commute up to a diagonal operator which fixes elements of the Z-support of $|\psi\rangle$ hence $\mathbf{S}_X, \mathbf{S}_Z$ stabilises $|\psi\rangle$. 

To extend this method to generalised controlled Z operators of form $CP(1/2,\mathbf{v})$, we need to ensure that $\mathbf{u}\mathbf{v} \ne \mathbf{v}$ for any other operator $CP(p/q,\mathbf{u})$ involved in the weighted hypergraph state. Otherwise, we cannot guarantee that the commutators of the operators in $\mathbf{S}_X$ act trivially on the Z-support of the embedded state.

\section{Logical Operators - Proof of Results}\label{app:LO}
In this appendix, we provide proofs for the main results in Chapter~\ref{chap:LO}. We first prove two results about the properties of logical XP operators. We next show that the algorithms of Sections~\ref{sec:LI} and~\ref{sec:LO} produce valid logical operators, and that they produce all possible operators of XP form. We then show how to find valid phase and Z components for logical X operators. We then prove the observations of Example \ref{eg:reedmuller} regarding Reed-Muller codes. Finally, we prove that the more efficient algorithms of Section~\ref{sec:LO_Modified} work correctly.

\subsection{Properties of Logical XP Operators}
In this section, we prove two results on the properties of logical XP operators. Given the codewords $|\kappa_i\rangle$ of Section~\ref{sec:code_word_algorithm_summary}, the first result states that an XP operator is a logical operator if and only if its action on the codewords can be described in terms of:
\begin{itemize}
    \item A permutation of the codewords; and
    \item A phase applied to each codeword
\end{itemize}
This result is used to prove that the logical operator and logical identity algorithms work correctly.

The second result states that an XP operator is a logical operator if and only if its commutators with logical identities are logical identities. It is an efficient way in practice to verify if an XP operator is a logical operator on the codespace.

\begin{proposition}[Action of Logical Operators]\label{prop:LXP_action}
An XP operator $A$ is a logical operator if and only if its action can be described by a permutation $\pi$ of the codewords such that $\pi^2 = 1$ and a vector $\mathbf{f}  \in \mathbb{Z}_{2N}^{\dim(\mathcal{C})}$ specifying the phase applied to each codeword. 
\end{proposition}
\begin{proof}
Let $A = XP_N(p|\mathbf{x}|\mathbf{z})$ be an XP operator. Applying the rule in Eq.~\eqref{eq:action_on_basis_elts}, $A$ acts on computational basis elements as follows:
\begin{align}
    A|\mathbf{e}\rangle = \omega^{p + 2\mathbf{e}\cdot\mathbf{z}}|\mathbf{e}\oplus\mathbf{x}\rangle
\end{align}
Hence, the image of a computational basis element under an XP operator cannot be a superposition of computational \textbf{basis elements}. 

Now consider how $A$ acts on the codewords $\{ |\kappa_i\rangle \}$ of Section~\ref{sec:code_word_algorithm_summary}. Let $E_i$ be the Z-support of the codeword $|\kappa_i\rangle$ and let $r$ be the number of non-diagonal canonical generators (i.e. $r= |\mathbf{S}_X|$). Then $|E_i| = 2^r$ and $E_i \cap E_j = \emptyset$ for $i \ne j$. Hence, the image of a codeword under an XP operator cannot be a superposition of \textbf{codewords}. 

Now assume $A$ is a logical operator and so preserves the codespace $\mathcal{C}$. As $A|\kappa_i\rangle \in \mathcal{C}$ for each codeword and $A$ cannot create superpositions of codewords, then $A|\kappa_i\rangle = \omega^{q_i}|\kappa_j\rangle$ for some codeword $|\kappa_j\rangle$  and $q_i \in \mathbb{Z}_{2N}$. Because the image of the codewords under $A$ must span the codespace, $A$ must permute the codewords. The square of any XP operator is diagonal (see Section~\ref{sec:algebraic_identities}) so the square of the permutation must be $1$ and the action of $A$ is as claimed.
\end{proof}

In the PSF, any logical operator $L$ on the codespace must commute with the stabiliser generators. In the XPF, we instead work with the generators $\mathbf{M}$ of the logical identity group. The group commutator of a logical operator $L$ with each element of $\mathbf{M}$ must be in the diagonal subgroup $\langle \mathbf{M}_Z\rangle$. In practice, this gives an efficient test for determining if a given operator is a logical operator on the codespace.

\begin{proposition}[Commutators Logical Operators and Logical Identities]\label{prop:LOtest}
Let $\mathbf{M}$ be the logical identity generators as in Section~\ref{sec:LI}. $L$ is a logical XP operator if and only if $A^{-1}L^{-1}AL \in \langle \mathbf{M}_Z \rangle, \forall A \in \mathbf{M}$.
\end{proposition}

\begin{proof}

Assume $A \in \mathbf{M}$, $L$ is a logical XP operator and $\{|\kappa_i\rangle\}$ are the codewords of Section~\ref{sec:code_word_algorithm_summary}. By Proposition~\ref{prop:LXP_action}, there exists a phase vector $\mathbf{f}$ and a permutation $\pi$ such that for all $i$:
$$
A L |\kappa_i\rangle = A \omega^{\mathbf{f}[i]} |\kappa_{
\pi(i)}\rangle = \omega^{\mathbf{f}[i]} |\kappa_{
\pi(i)}\rangle =  \omega^{\mathbf{f}[i]} A |\kappa_{
\pi(i)}\rangle = L A |\kappa_i\rangle
$$
Hence $A^{-1}L^{-1}AL|\kappa_i\rangle = |\kappa_i\rangle$ so $A^{-1}L^{-1}AL \in  \mathcal{I}_\text{XP}$. From Section~\ref{sec:algebraic_identities}, we know that group commutators are always diagonal operators, hence $A^{-1}L^{-1}AL \in \langle \mathbf{M}_Z \rangle$.

Conversely, assume that $A^{-1}L^{-1}AL \in \langle \mathbf{M}_Z \rangle, \forall A \in \mathbf{M}$. This is also true for any $A \in \langle \mathbf{M}\rangle$. Then for each codeword $|\kappa_i\rangle$:
\begin{align}
A^{-1}L^{-1}AL|\kappa_i\rangle &= |\kappa_i\rangle, \forall A \in \langle \mathbf{M}\rangle\\
A(L|\kappa_i\rangle) &= LA |\kappa_i\rangle = L |\kappa_i\rangle
\end{align}
Hence $L|\kappa_i\rangle$ is in the codespace for all $|\kappa_i\rangle$. Because $L$ is an XP operator $L|\kappa_i\rangle$ cannot be a superposition of codewords and so $L|\kappa_i\rangle = \omega^{p_i}|\kappa_j\rangle$ for some $p_i \in \mathbb{Z}_{2N}$ and $|\kappa_j\rangle$. Because $L^2$ is diagonal the map $\pi: i \mapsto j$ squares to $1$ and hence is a permutation. Therefore we can describe the action of $L$ as $L|\kappa_i\rangle = \omega^{p_i}|\kappa_{\pi(i)}\rangle$ and so $L$ is a logical XP operator by Proposition~\ref{prop:LXP_action}.
\end{proof}

\subsection{Logical Identity and Logical Operator Algorithms}
In this section, we prove that the logical identity algorithm of Section~\ref{sec:LI} and the logical operator algorithm of Section~\ref{sec:LO} yield generating sets of operators. We first consider the algorithms for diagonal operators, then those for non-diagonal operators. We then show how to find logical X operators by using the intersection of affine spans algorithm of Section~\ref{sec:affinespan}. 

\subsubsection{Diagonal Operator Algorithms}\label{sec:E_M+E_L}
In this section, we show that the algorithms of Sections~\ref{sec:L_Z} and~\ref{sec:M_Z} produce sets of diagonal logical generators and diagonal logical identity generators respectively.

Assume we have the codewords in the orbit form of Eq.~\eqref{eq:orbitform} - i.e. $|\kappa_i \rangle =  \sum_{0 \le j < 2^r} \omega^{p_{ij}} |\mathbf{e}_{ij}\rangle$.  Let $E$ be the Z-support of the codewords (see Section~\ref{sec:code_words_notation}). We first look at diagonal logical identity operators and prove the following proposition:

\begin{proposition}
[Diagonal Logical Identity Group]\label{prop:M_Z}
The algorithm in Section~\ref{sec:M_Z} produces a list of diagonal XP operators $\mathbf{M}_Z$ which generate the diagonal logical identity XP operators for the codespace.
\end{proposition} 
\begin{proof}
The binary matrix $E_M$ is the matrix formed by taking  $(\mathbf{e}|1)$ as rows, where $\mathbf{e} \in E$. If $(\mathbf{z}_k|p_k)$ is a row of the Howell basis $K_M$ of $\Ker(E_M)$, then $(\mathbf{z}_k|p_k) \cdot (\mathbf{e}_{ij}|1) \mod N = 0$. Let $A_k := XP_N(2 p_k|\mathbf{0}|\mathbf{z}_k)$ then $A_k|\mathbf{e}\rangle = \omega^{2p_k + 2 \mathbf{e} \cdot \mathbf{z}_k}|\mathbf{e}\rangle =\omega^{2(\mathbf{z}_k|p_k) \cdot (\mathbf{e}|1)}|\mathbf{e}\rangle = |\mathbf{e}\rangle$. Hence, $A_k$ applies a trivial phase on each element of $E$ and hence on each codeword. Because $\Ker_{\mathbb{Z}_N}(E_M) = \Span_{\mathbb{Z}_N}(K_M)$, the $A_k$ generate all diagonal logical identity operators and so $\mathbf{M}_Z := \{A_k\}$.
\end{proof}
We next look at the diagonal logical operators and show the following:

\begin{proposition}[Diagonal Logical Operators]\label{prop:L_Z}
The algorithm in Section~\ref{sec:L_Z} produces a list of diagonal XP operators $\mathbf{L}_Z$ which together with $\omega I$ and $\mathbf{M}_Z$ generate the diagonal logical XP operators for the codespace.
\end{proposition}
\begin{proof}
We define the binary matrix $E_L$ used in the logical operator algorithm as follows. For each $\mathbf{e}_{ij}$ in Eq.~\eqref{eq:orbitform}, let $\mathbf{i}$ be a binary vector of length $\dim(\mathcal{C})$ which is all zeros apart from the $i$th component which is $1$. The vector $\mathbf{i}$ is a ``codeword index'' that tells us which codeword the row belongs to. Let $E_L$ be the matrix formed by taking  $(\mathbf{e}_{ij}|\mathbf{i})$ as rows. Let  $K_L$ be the Howell basis of $\Ker_{\mathbb{Z}_N}(E_L)$ and let $(\mathbf{z}_k|\mathbf{p}_k)$ be a row of $K_L$. Then for all values of $i$ and $j$, $(\mathbf{z}_k|\mathbf{p}_k) \cdot (\mathbf{e}_{ij}|\mathbf{i}) \mod N = 0$ and so $\mathbf{z}_k \cdot \mathbf{e}_{ij} \mod N = - \mathbf{p}_k[i] \mod N$.

Now consider how the diagonal XP operator $B_k:= XP_N(0|\mathbf{0}|\mathbf{z}_k)$ acts on the Z-support of the codewords:
\begin{align}
    B_k|\mathbf{e}_{ij}\rangle &= \omega^{2\mathbf{e}_{ij}\cdot \mathbf{z}_k}|\mathbf{e}_{ij}\rangle = \omega^{-2\mathbf{p}_k[i]}|\mathbf{e}_{ij}\rangle
\end{align}
In other words, the action of $B_k$ is constant on the Z-support of each codeword and so is a logical operator with phase vector $\mathbf{f}_k = -2\mathbf{p}_k$. The $B_k$, together with the operator $\omega I$, generate all possible diagonal logical operators because $\Ker_{\mathbb{Z}_N}(E_L) = \Span_{\mathbb{Z}_N}(K_L)$.

We now show how to find non-trivial diagonal operators $\mathbf{L}_Z$ which together with $\omega I$ and $\mathbf{M}_Z$ generate the same group as the $\mathbf{B}_Z := \{B_k\}$ and $\omega I$, reflecting steps 4 and 5 of the algorithm in Section~\ref{sec:L_Z}. Because we include $\omega I$ as a generator, we can just consider the space spanned by the Z components of the operators over $\mathbb{Z}_N$. Let $M_Z$ and $B_Z$ be the matrices whose rows are the Z-components of the $\mathbf{M}_Z$ and $\mathbf{B}_Z$ respectively. For each row $\mathbf{z}_k$ of $B_Z$ we let $\mathbf{r}_k = \res_{\mathbf{Z}_N}(M_Z)$ so that $\mathbf{z}_k = (\mathbf{r}_k + \mathbf{u}M_Z) \mod N$ for some vector $\mathbf{u} \in \mathbf{Z}_N^{|M_Z|}$. Let $L_Z$ be the Howell basis of the matrix with the $\mathbf{r}_k$ as rows. Then clearly $\Span_{\mathbb{Z}_N}(B_Z) = \Span_{\mathbb{Z}_N}(M_Z) + \Span_{\mathbb{Z}_N}(L_Z)$. Letting $\mathbf{L}_Z = \{XP_N(0|\mathbf{0}|\mathbf{z}) : \mathbf{z} \in L_Z\}$ then we have that  $\langle \omega I, \mathbf{M}_Z,\mathbf{L}_Z\rangle = \langle \omega I, \mathbf{B}_Z\rangle$.
\end{proof}

\subsubsection{Non-diagonal Operator Algorithms}
In this section, we show that the algorithms for non-diagonal operators yield valid logical operators and logical identities respectively. We then show that the operators produced are generating sets for the respective groups. Assume we have the coset decomposition of the Z-support of the codewords as in Eq.~\eqref{eq:Emcoset} i.e. $E = E_q + \langle S_X\rangle + \langle L_X \rangle$. Recall that the rows of $S_X$ are the X-components of the non-diagonal canonical generators and that the binary vector $\mathbf{x} \in \langle L_X\rangle$ if and only if $E_m \oplus \mathbf{x} = E_m$ (see Section~\ref{sec:L_X}).

\begin{proposition}[Non-diagonal Logical Operators]\label{prop:L_X}
Given a binary vector $\mathbf{x} \in L_X$, the algorithm in Section~\ref{sec:L_X} yields a non-diagonal logical operator with X-component equal to $\mathbf{x}$ or FALSE if no such operator exists.
\end{proposition}
\begin{proof}
Given $\mathbf{x} \in L_X$ and applying Proposition~\ref{prop:LXP_action}, we wish to find  $A = XP_N(0|\mathbf{x}|\mathbf{z})$ for which $A|\kappa_i\rangle = \omega^{\mathbf{f}[i]}|\kappa_{\pi(i)}\rangle$ for some phase vector $\mathbf{f}$ and some permutation $\pi$. 

Because $\mathbf{x} \in L_X$, we know that $\mathbf{m}_i \oplus \mathbf{x} =\mathbf{m}_j$ for orbit representatives $\mathbf{m}_i, \mathbf{m}_j$. Define $\pi: i \mapsto j$ if $\mathbf{m}_i \oplus \mathbf{x} =\mathbf{m}_j$. Then $\pi$ is a map which squares to $1$ and so is a permutation.

Now assume we have the codewords in the orbit format of Eq.~\eqref{eq:orbitform}. Let $\mathbf{e}_{ij}' = \mathbf{e}_{ij} \oplus \mathbf{x}$ and let $p_{ij}'$ be the phase of $\mathbf{e}_{ij}'$ in the codewords. We require that $A\omega^{p_{ij}}|\mathbf{e}_{ij}\rangle = \omega^{\mathbf{f}[i] + p_{ij}'}|\mathbf{e}_{ij}'\rangle, \forall i, j$. Calculating the action of $A$ on $\omega^{p_{ij}}|\mathbf{e}_{ij}\rangle$:
\begin{align}
    A\omega^{p_{ij}}|\mathbf{e}_{ij}\rangle &= \omega^{p_{ij}}\omega^{ 2\mathbf{e}_{ij}\cdot \mathbf{z}}|\mathbf{e}_{ij}'\rangle
    \end{align}
The phase of  $|\mathbf{e}_{ij}'\rangle$ is $\omega^{\mathbf{f}[i] +p_{ij}'}$ when:
\begin{align}
    (p_{ij} -p_{ij}' - \mathbf{f}[i] + 2\mathbf{e}_{ij} \cdot \mathbf{z}) \mod 2N &= 0 
    \end{align}
As the phase $\mathbf{f}[i]$ is fixed for each codeword $|\kappa_i\rangle$, for there to be a valid solution $p_{ij} - p_{ij}'$ to be either even or odd for all $j$. Let $a_i = (p_{ij} - p_{ij}') \mod 2$ and $p_{ij}'' = (p_{ij} - p_{ij}' -a_i)/2$. Define the vector $\mathbf{q} := (a_i-\mathbf{f})/2$. Hence:
\begin{align}
    (p_{ij}'' + \mathbf{q}[i] + \mathbf{e}_{ij}\cdot \mathbf{z}) \mod N &= 0
    \end{align}
Letting $\mathbf{p}''$ be the vector corresponding to the  $p_{ij}''$, we can write this in matrix form:
\begin{align}
    (\mathbf{p}'' + (\mathbf{z}|\mathbf{q})E_L^T) \mod N &= 0
    \end{align}
Solutions to this equation are members of the affine span (see Eq.~\eqref{eq:affine_span}):
\begin{align}
    (\mathbf{z}|\mathbf{q}) &\in \mathbf{b} + \Span_{\mathbb{Z}_N}( K_L)
    \end{align}
where $K_L$ is the Howell basis of $\Ker_{\mathbb{Z}_N}(E_L)$ and $\mathbf{b} \in \mathbb{Z}_N^n \times \mathbb{Z}_N^{\dim(\mathcal{C})}$ is a constant. We can either find $\mathbf{b}$ or show that no solution exists by using linear algebra modulo $N$  (see Section~\ref{sec:linalgmodN}). If no such solution exists, return FALSE. Otherwise, the operator $A = XP_N(0|\mathbf{x}|\mathbf{z})$ is the required non-diagonal logical operator.
\end{proof}

The algorithm for non-diagonal logical identity operators can be considered a special case of the algorithm for non-diagonal logical operators, with the matrix $E_M$ substituted for $E_L$ and the proof is omitted.
\begin{corollary}[Non-diagonal Logical Identity Group]\label{prop:M_X}
Given codewords $|\kappa_i\rangle$ in orbit format and $\mathbf{x} \in S_X$, the algorithm in Section~\ref{sec:M_X} yields a non-diagonal operator $A$ with X-component $\mathbf{x}$ such that $A|\kappa_i\rangle = |\kappa_i\rangle, \forall i$, or returns FALSE if no such operator exists.
\end{corollary}

We next show that the operators produced by the algorithms give us a generating set. We first show that the logical identity group is generated by the operators yielded by the logical identity algorithm.

\begin{proposition}[Logical Identity Group is Generated by $\mathbf{M}$]\label{prop:LIeqM}
The logical identity group for an XP code is the group generated by $\mathbf{M}$ as calculated in Section~\ref{sec:LI}:
\begin{align}\mathcal{I}_\text{XP} = \langle\mathbf{M}\rangle\end{align}
\end{proposition}
\begin{proof}
If $A\in \langle\mathbf{M}\rangle$  then $A| \kappa_i\rangle = | \kappa_i\rangle$ for all codewords $| \kappa_i\rangle$ by construction so $A \in \mathcal{I}_\text{XP}$, and so $\mathcal{I}_\text{XP} \supset \langle\mathbf{M}\rangle$.

We now prove the converse. In Proposition~\ref{prop:M_Z}, we showed that $\mathbf{M}_Z$ generates the diagonal subgroup of $\mathcal{I}_\text{XP}$

Let $B = XP_N(p|\mathbf{x}|\mathbf{z})$ be a non-diagonal operator in $ \mathcal{I}_\text{XP}$. Then $\mathbf{x} \in \Span_{\mathbb{Z}_2}(S_X)$ as defined above. Hence, we can find a binary vector $\mathbf{v}$ such that $\mathbf{x} = \mathbf{v}S_X \mod 2$. Let $B' = \mathbf{M}_X^{\mathbf{v}} \in \langle \mathbf{M} \rangle$. Then $B$ and $B'$ have the same $X$ component so $B'= XP_N(p'|\mathbf{x}|\mathbf{z}')$. Then $A' = B B' = XP_N(p''|\mathbf{0}|\mathbf{z}')$ is a diagonal logical identity operator because both $B$ and $B'$ are logical identities. Hence $A' \in \langle\mathbf{M}\rangle$. Therefore $B = A' B'^{-1} \in \langle\mathbf{M}\rangle$.
\end{proof}

Finally, we show that the logical operators with X-components in $L_X$, along with the stabiliser generators and the diagonal logical operators generate the entire logical operator group.

\begin{proposition}\label{prop:L_X_X-components}
Let $E$ be the Z-support of the codewords of an XP code and let $E = E_q+ \langle S_X \rangle + \langle L_X \rangle$ be the coset decomposition of $E$ as in Eq.~\eqref{eq:Emcoset}. Let $\mathbf{M}$ be the logical identity group as in Section~\ref{sec:LI}, let $\mathbf{L}_Z$ be the set of diagonal operators as in Section~\ref{sec:L_Z} and let $\mathbf{L}_X$ be a set of logical operators with X-components drawn from the rows of $L_X$. Then the logical operator group is given by: 
\begin{align}
    \mathcal{L}_\text{XP} &= \langle \omega I, \mathbf{M}, \mathbf{L}_Z, \mathbf{L}_X \rangle
\end{align}
\end{proposition}
\begin{proof}
Say we have a logical operator with a non-trivial X-component of the form $A=XP_N(p|\mathbf{x}|\mathbf{z})$. The action of $A$ on $|\mathbf{e}\rangle$ for $\mathbf{e} \in E$ is $ A|\mathbf{e}\rangle = \omega^{p + 2\mathbf{e}\cdot\mathbf{z}}|\mathbf{e}\oplus\mathbf{x}\rangle$.
We have shown that the codespace is spanned by the $|\kappa_i\rangle$ which have Z-support exactly equal to $E$. Because logical operators preserve the codespace, X-components of logical operators must satisfy the constraint: 
\begin{align}
    \mathbf{e}\oplus\mathbf{x} \in E, \forall \mathbf{e}\in E\label{eq:xcompconstr}
\end{align}

All possible X-components for logical operators are given by $\langle S_X\rangle + \langle L_X\rangle$. To see this, assume there exists some $\mathbf{x} \notin \langle S_X\rangle + \langle L_X\rangle$ meets the constraint in Eq.~\eqref{eq:xcompconstr}. Then $\text{Res}_{\mathbb{Z}_2}(S_X,\mathbf{x}) \notin \langle L_X\rangle$ meets the constraint in Eq.~\eqref{eq:xcompconstr}, which is a contradiction.

Next, we show that any logical operator can be written as $A = XP_N(p|\mathbf{x}|\mathbf{z}) = L S$ where $L$ has X-component in $\langle L_X\rangle$ and $S \in \langle\mathbf{M}_X\rangle$. We know that $\mathbf{x} \in \langle S_X\rangle + \langle L_X\rangle$ so $\mathbf{x} = (\mathbf{u}S_X +\mathbf{v}L_X) \mod 2$ for some binary vectors $\mathbf{u},\mathbf{v}$. Let $S = \mathbf{M}_X^\mathbf{u}$ then $L = AS^{-1}$ is a logical operator with X component in $\langle L_X\rangle$. 

Now assume that there are two different logical operators with X-component $\mathbf{x} \in \langle L_X\rangle$ - say $A_1$ and $A_2$. The product $B = A_1 A_2^{-1}$ is a diagonal logical operator so  $A_1$ is the same as $A_2$ up to a product of a diagonal logical operator. Applying Proposition~\ref{prop:L_Z}, $B \in \langle \omega I, \mathbf{M}_Z, \mathbf{L}_Z\rangle$

Accordingly, we can generate all possible logical operators by finding the non-diagonal logical operators with X-components in $L_X$. Together with the logical identity generators $\mathbf{M}$, $\omega I$ and $\mathbf{L}_Z$, these generate the full set of logical operators.
\end{proof}

\subsubsection{Diagonal Component of Logical X Operators}\label{sec:diagonal_component_L_X}
We have demonstrated how to find a non-diagonal logical operator with X-component in $L_X$. This will not necessarily act as a logical X operator. In particular, a logical X should square to a logical identity and in this section, we show how to ensure that this is the case. 

Assume we have used the algorithm in Section~\ref{sec:L_X} to find a logical operator $B = XP_N(0|\mathbf{x}|\mathbf{b})$. Our aim is to calculate a logical X operator $A = XP_N(p|\mathbf{x}|\mathbf{z})$ such that $A^2$ is a logical identity. As $A^2$ is diagonal, we require $A^2 \in \langle M_Z\rangle$. Our strategy is to first find the Z component of $A$, then adjust the phase component. 

We first look at which Z components are possible for logical operators with the same X component as $B$. We can multiply $B$ by any diagonal logical operator in $\langle  \mathbf{M}_Z, \mathbf{L}_Z\rangle$ and get a logical operator, and all possible Z components of such operators arise in this way. Let $M_{Z},L_{Z}$ be the matrices formed from the Z components of $\mathbf{M}_Z, \mathbf{L}_Z$ respectively. Writing $\langle C \rangle = \Span_{\mathbb{Z}_{N}}(C)$, the possible Z components are given by the affine span:
\begin{align}
    \mathbf{z} \in \text{Sp}_b := \mathbf{b} +  \langle M_{Z}\rangle + \langle L_{Z}\rangle
\end{align}
We next address the condition that the Z component of $A^2$, which we denote $\mathbf{z}_2$ here, has to be in $\langle M_{Z}\rangle$. By applying the square rule of Section~\ref{sec:algebraic_identities}:
\begin{align}
    \mathbf{z}_2[i] = \begin{cases}2\mathbf{z}[i] &: \mathbf{x}[i] = 0\\
    0 &: \mathbf{x}[i] = 1\end{cases}\label{eq:LX_square}
\end{align}
When $\mathbf{x}[i]=1,\mathbf{z}_2[i] =0$. Hence, $\mathbf{z}_2$ has to be an element of $\langle M_{Z}\rangle$ which is zero when $\mathbf{x}[i]=1$.  Because $\mathbf{z}_2$ is the Z component of a square, we also know that it is a multiple of $2$. Hence $\mathbf{z}_2 \in \text{Sp}_2 := \langle M_{Z} \rangle \cap \langle 2\text{diag}(1-\mathbf{x}) \rangle$. All generators of $\text{Sp}_2$ are divisible by 2, so let $\text{Sp}_M$ be the span with generators from $\text{Sp}_2$ divided by $2$. Any $\mathbf{z} \in \text{Sp}_M$ will have $\mathbf{z}_2 \in \langle M_{Z}\rangle$.

Where the precision $N$  is a multiple of 2 and $\mathbf{x}[i] = 0$, adding $\frac{N}{2}$ to $\mathbf{z}[i]$ does not change whether $\mathbf{z}_2 \in \langle M_{Z}\rangle$ because $\mathbf{z}_2[i] = 2\mathbf{z}[i]$. Hence, we can add any element of $\langle \frac{N}{2} \text{diag}(1-\mathbf{x})\rangle$ to $\mathbf{z}$.

We can add any element of $\langle \text{diag}(\mathbf{x})\rangle$ to $\mathbf{z}$ and still ensure $\mathbf{z}_2 \in \langle M_{Z}\rangle$ because where $\mathbf{x}[i] = 1$, $\mathbf{z}_2[i]$ is guaranteed to be zero. Hence $\mathbf{z}_2 \in \langle M_{Z}\rangle$ when $\mathbf{z} \in \text{Sp}_a$ defined as:
\begin{align}\text{Sp}_a := \text{Sp}_M + \langle \frac{N}{2} \text{diag}(1-\mathbf{x})\rangle + \langle  \text{diag}( \mathbf{x})\rangle
\end{align}
Hence, $\mathbf{z} \in \text{Sp}_a \cap \text{Sp}_b$. If this span is not empty then we can find $\mathbf{z}$ such that $\mathbf{z}_2 = \mathbf{u}M_Z \mod N \in \langle M_{Z}\rangle$ for some vector $\mathbf{u}$. Let $p_z,p_u$ be the phase components of $XP_N(0|\mathbf{x}|\mathbf{z})^2$ and $\mathbf{M}_Z^\mathbf{u}$ respectively. We set the phase component of $A$ to be $p = (p_u-p_z)/2$.

\subsection{Reed-Muller Codes}\label{sec:reed-muller}
Reed-Muller codes are well-known Pauli Stabiliser codes. We can also look at them as XP codes and better understand their logical operator structure. In the Proposition below, we show that they are self-dual codes and have transversal logical operators at precision $N> 2$. To demonstrate the results, we apply the techniques of Appendix \ref{app:hypergraph}.
\begin{proposition}[Reed-Muller XP Codes]\label{prop:reed-muller}
The Reed-Muller code on $2^r-1$ qubits can be written as the codespace of a precision $N = 2^{r-2}$ code whose stabiliser generators are symmetric in X and P. The code has a transversal logical $\text{diag}(1,\exp(-2\pi i/2^{r-1}))$ operator.
\end{proposition}

\begin{proof}
Let $N = 2^{r-2}$ and the number of qubits $n = 2^r-1$. Let $\mathbf{S}_X = \{XP_N(0|\mathbf{x}_i|\mathbf{0})\}$ and $\mathbf{S}_Z = \{XP_N(0|\mathbf{0}|\mathbf{x}_i)\}$ where $\mathbf{x}_i$ is the $i$th row of $M^r_r$ as in Section~\ref{sec:whg2xp}.  Let $|0\rangle_L = O_{\mathbf{S}_X}|\mathbf{0}\rangle$ and let $|1\rangle_L = XP_N(0|\mathbf{1}|\mathbf{0})|0\rangle_L$ where $\mathbf{1}$ is the vector of length $2^r-1$ with all entries $1$.

We claim that the codespace of $\mathbf{S}_X, \mathbf{S}_Z$ is spanned by $\{|0\rangle_L,|1\rangle_L\}$. For this, we need to show that the simultaneous +1 eigenspace of the $\mathbf{S}_Z$ is given by $\{|\mathbf{e}\rangle\}$ where $\mathbf{e} \in \Span_{\mathbb{Z}_2}\{\mathbf{1}, \mathbf{x}_i : 0 \le i < r \}$. Let $u \subset [0 \dots r-1]$ and let $\mathbf{s}_u:= \bigoplus_{j \in u} \mathbf{x}_j$. Note that $XP_N(0|\mathbf{0}|\mathbf{x}_k)|\mathbf{s}_u\rangle = \omega^{2 \mathbf{x}_k \cdot \mathbf{s}_u} |\mathbf{s}_u\rangle$. Hence we need to show that $\mathbf{x}_k \cdot \mathbf{s}_u \mod 2^{r-2} = \mathbf{x}_k \cdot (\mathbf{1} \oplus \mathbf{s}_u) \mod 2^{r-2} = 0$.

Applying Proposition \ref{prop:sums_products_mod_2} the weight of the product of $t$ distinct $\mathbf{x}_i$ is $2^{r-t}$ because:
\begin{align}
\mathbf{1}\cdot \prod_{0 \le i < t}\mathbf{x}_i &= \sum_{k \le t \le r}\binom{r-t}{k-t} = \sum_{0 \le j \le r-t}\binom{r-t}{j} = 2^{r-t}\label{eq:wt_prod}
\end{align}
Now consider the dot product of $\mathbf{x}_k$ with $\mathbf{s}_u$ and apply Eq.~\eqref{eq:venn_diagram}:
\begin{align}
\mathbf{x}_k \cdot \mathbf{s}_u &= \sum_{s \subset u}(-2)^{|s|-1}\mathbf{x}_k \cdot \prod_{i \in s}\mathbf{x}_i = \sum_{s \subset u}(-2)^{|s|-1}2^{r-|s|-1} = \sum_{s \subset u}(-1)^{|s|-1}2^{r-2} =0 \mod 2^{r-2}
\end{align}
By setting $t=1$ in Eq.~\eqref{eq:wt_prod} we see that $\mathbf{x}_k \cdot \mathbf{1} = 2^{r-1}$. Hence:
\begin{align}
\mathbf{x}_k \cdot (\mathbf{1} \oplus \mathbf{s}_u) &= \mathbf{x}_k \cdot (\mathbf{1} - \mathbf{s}_u) = \mathbf{x}_k \cdot \mathbf{1} -  \mathbf{x}_k \cdot \mathbf{s}_u = 0 \mod 2^{r-2}
\end{align}
Now rescale the code to precision $N = 2^{r-1}$ consider the action of $B = XP_N(0|\mathbf{0}|\mathbf{1})$ on $|0\rangle_L$ and $|1\rangle_L$. We can write $|0\rangle_L = \sum_{u \subset [0\dots r-1]}\mathbf{s}_u$. Now $B|\mathbf{s}_u\rangle = \omega^{2\mathbf{1}\cdot\mathbf{s}_u}|\mathbf{s}_u\rangle$. Calculating $\mathbf{1}\cdot\mathbf{s}_u$:
\begin{align}
    \mathbf{1}\cdot\mathbf{s}_u = \sum_{s \subset u}(-2)^{|s|-1}\mathbf{1} \cdot \prod_{i \in s}\mathbf{x}_i = \sum_{s \subset u}(-2)^{|s|-1}2^{r-|s|} = \sum_{s \subset u}(-1)^{|s|-1}2^{r-1} = 0 \mod 2^{r-1}
\end{align}
Hence, the phase applied by $B$ to each basis element making up $|0\rangle_L$ is trivial and so $B|0\rangle_L = |0\rangle_L$.

The action of $B$ on $|1\rangle_L$ is $B |1\rangle_L =  XP_N(0|\mathbf{0}|\mathbf{1}) XP_N(0|\mathbf{1}|\mathbf{0})|0\rangle_L=XP_N(0|\mathbf{1}|\mathbf{1})D(\mathbf{2})|0\rangle_L$. The phase component of $D(\mathbf{2})$ is $2(2^{r} -1) = -2 \mod N$. Hence $B|1\rangle_L = \omega^{-2}|1\rangle_L$ and $B$ acts as a logical $\text{diag}(1,\exp(-2\pi i/2^{r-1}))$ operator.
\end{proof}

\subsection{Modified Algorithms for Logical Identity Group and Logical Operators}\label{app:modified_LI+LO}

Our objective here is to show that where $N = 2^t$, we can calculate the logical identity group and logical operators without first calculating the codewords in full. Assume that the codewords in orbit form are as in Eq.~\eqref{eq:orbitform} and let $E$ be the Z-support of the codewords.

\begin{itemize}
\item To calculate the logical identity group generators, we just need to consider the elements of $E$ at most orbit distance $t$ from the core elements $E_q =\{\mathbf{q}_l\}$.
\item For the logical operators, we need to consider the elements of $E$ at most orbit distance $t$ from the orbit representatives $E_m = \{\mathbf{m}_i\}$.
\end{itemize}

We need a number of preliminary results to prove the main propositions. In these results, we work with a binary matrix $L$ in RREF. Let the rows of $L$ be the binary vectors $\mathbf{x}_i, 0 \le i < r$. Sums and products of vectors are component wise in $\mathbb{Z}$ (i.e. $(\mathbf{x} + \mathbf{y})[i] = \mathbf{x}[i]+\mathbf{y}[i], (\mathbf{x}\mathbf{y})[i] = \mathbf{x}[i]\mathbf{y}[i]$). 

We first write an expression for the sum modulo 2 of a subset of the $\mathbf{x}_i$:

\begin{lemma}
Let $L$ be a binary matrix in RREF, with non-zero rows $\mathbf{x}_i, 1 \le i < r$ and let $c \subset [0\dots r-1]$. Define:
\begin{align}
\mathbf{s}_c &:= \Big(\sum_{i \in c}\mathbf{x}_i\Big)\mod 2
\end{align}
and for $d \subset c$ define :
\begin{align}
\mathbf{p}_d &:= (-2)^{|d|-1}\prod_{i \in d}\mathbf{x}_i
\end{align}
then:
\begin{align}
\mathbf{s}_c &= \sum_{d \subset c, |d| \ge 1}\mathbf{p}_d \label{eq:binsum}
\end{align}
where we work in $\mathbb{Z}$ on the RHS.
\end{lemma}
\begin{proof}
Induction on $|c|$.

\paragraph{Base Step:}
Let $|c| = 1$ so that $c = \{j\}$. Then the LHS of Eq.~\eqref{eq:binsum} is:
\begin{align}
    \Big(\sum_{i \in c}x_i\Big)\mod 2 &= \mathbf{x}_j
    \end{align}
While the RHS is:
\begin{align}
    \sum_{d \subset c, |d| \ge 1}(-2)^{|d|-1}\prod_{i \in d}\mathbf{x}_i &= \mathbf{x}_j
\end{align}
So the base step holds.

\paragraph{Induction Step:} Assume $|c| > 1$ and the lemma is true for any sets $d: |d| < |c|$. Let $j \in c$ and $c' = c \setminus \{j\}$. For binary vectors $\mathbf{x}, \mathbf{y}$: 
\begin{align}
(\mathbf{x} + \mathbf{y})\mod 2 &= \mathbf{x} + \mathbf{y} -2 \mathbf{x} \mathbf{y}\,. 
\end{align}
Hence, the LHS of  Eq.~\eqref{eq:binsum} is:
\begin{align}
    \mathbf{s}_c &= (\mathbf{x}_j + \mathbf{s}_{c'}) \mod 2= \mathbf{x}_j + \mathbf{s}_{c'} -2\mathbf{x}_j \mathbf{s}_{c'}
\end{align}
Because $|c'| = |c|-1 < |c|$, we can apply the induction hypothesis:
\begin{align}
    \mathbf{s}_c &= \mathbf{x}_j + \sum_{d \subset c', |d| \ge 1}\mathbf{p}_d - 2\mathbf{x}_j \Big(\sum_{d \subset c', |d| \ge 1}\mathbf{p}_d\Big)\\
    &= \sum_{d \subset c, |d| \ge 1 , j \notin d }\mathbf{p}_d+ \sum_{d \subset c,  |d| \ge 1, j \in d}\mathbf{p}_d\\
    &= \sum_{d \subset c, |d| \ge 1}\mathbf{p}_d\,.
\end{align}
The induction step holds and the result follows.
\end{proof}

Next, we let $N=2^t$ and consider spans over $\mathbf{Z}_N$ of the binary vectors $E = E_q + \langle L\rangle$. We first show that certain products are in the span:

\begin{lemma}[Products of $\mathbf{q}, \mathbf{x}_i$ of Degree $\le t$]\label{lemma:polynomials}
Let $E$ be a set of binary vectors with coset decomposition $E = E_q + \langle L \rangle$ as in Eq.~\eqref{eq:Emcoset}. Let the nonzero rows of $L$ be $\mathbf{x}_i, 1 \le i < r$. Let $N = 2^t$ where $t \in [1 \dots r]$. Let $E_t$ be the elements of $E$ at most orbit distance $t$ from $E_q$ so that $E_{t} := \{(\mathbf{q} + \mathbf{v}L) \mod 2: \mathbf{q} \in E_q,  \mathbf{v} \in \mathbb{Z}_2^r, \wt(\mathbf{v})  \le t\}$. Let $c \subset [0 \dots r-1], 1\le |c| \le t$ and let $\mathbf{p}_c := 2^{|c|-1}\prod_{i \in c}\mathbf{x}_i$. Then:
\begin{align}
    \mathbf{p}_c - 2\mathbf{q}\mathbf{p}_c \in \Span_{\mathbb{Z}_N}(E_t)\label{eq:binpoly}
\end{align}
\end{lemma}
\begin{proof}
Induction on $|c|$. 

\paragraph{Base Case:} let $|c| = 1$ so that $c = \{j\}$. Then:
\begin{align}
    \mathbf{p}_c - 2\mathbf{q}\mathbf{p}_c &= \mathbf{x}_j - 2\mathbf{q} \mathbf{x}_j\
    = (\mathbf{q}+\mathbf{x}_j)\mod 2 - \mathbf{q}.
\end{align}
As $t \ge 1$, $(\mathbf{q}+\mathbf{x}_j)\mod 2 \in E_t$ and so is $\mathbf{q}$ (as it is of orbit distance 0). Hence the linear combination is in $\Span_{\mathbb{Z}_N}(E_t)$ and the base case holds.

\paragraph{Induction Step:} Assume $|c| > 1$ and that the lemma is true for any sets $d: |d| < |c|$. As $|c| \le t$, the binary vector $\mathbf{e} = (\mathbf{q} + \sum_{i \in c}x_i)\mod 2 = (\mathbf{q} + \mathbf{s}_c) \mod 2 \in E_t$. Expanding the expression for $\mathbf{e}$ using Eq.~\eqref{eq:binsum}:
\begin{align}
    \mathbf{e} &= \mathbf{q} + \mathbf{s}_c - 2\mathbf{q} \mathbf{s}_c\\
    &= \mathbf{q} +\sum_{d \subset c, |d| \ge 1}\mathbf{p}_d -  2\mathbf{q}\sum_{d \subset c, |d| \ge 1}\mathbf{p}_d\\
    &= \mathbf{q} +\sum_{d \subsetneq c, 1 \le |d| < |c|}(\mathbf{p}_d -  2\mathbf{q}\mathbf{p}_d) + (\mathbf{p}_c -  2\mathbf{q}\mathbf{p}_c)
\end{align}
By the induction hypothesis $\mathbf{p}_d -  2\mathbf{q}\mathbf{p}_d \in \Span_{\mathbb{Z}_N}(E_t)$ for any $|d| < |c|$. As $\mathbf{e}, \mathbf{q}$ are also in the span, the induction step holds.
\end{proof}

We next show that the product of up to $t$ of the binary vectors $\mathbf{x}_i$ is in $\Span_{\mathbb{Z}_N}(E_t)$:
\begin{corollary}[Products of $\mathbf{x}_i$ of Degree $t$]\label{cor:prodx}
Let $|c| = t$ and $\mathbf{p}_c, N, E_t$ be defined as in Lemma~\ref{lemma:polynomials}. Then:
\begin{align}
    \mathbf{p}_c \in \Span_{\mathbb{Z}_N}(E_t)\label{eq:p_c}
\end{align}
\end{corollary}
\begin{proof}
For $|c| \le t$, we have shown in Lemma~\ref{lemma:polynomials} that $\mathbf{p}_c - 2\mathbf{q}\mathbf{p}_c \in \Span_{\mathbb{Z}_N}(E_t)$. Expanding the second $\mathbf{p}_c$ we have:
\begin{align}
\mathbf{p}_c - 2\mathbf{q}\mathbf{p}_c &= \mathbf{p}_c -2(-2)^{|c|-1}\prod_{i \in c}\mathbf{x}_i
= \mathbf{p}_c +(-2)^{t}\prod_{i \in c}\mathbf{x}_i
\end{align}
As the span is over $\mathbf{Z}_N$, any multiples of $N = 2^t$ can be disregarded so $\mathbf{p}_c \in \Span_{\mathbb{Z}_N}(E_t)$.
\end{proof}
The results for the modified algorithms follow directly from the following proposition:
\begin{proposition}\label{prop:SpanEteqSpanE}
Let $E, E_t, N$ be as in Lemma~\ref{lemma:polynomials}. Then:
\begin{align}
    \Span_{\mathbf{Z}_N}(E) = \Span_{\mathbf{Z}_N}(E_t)
\end{align}
\end{proposition}
\begin{proof}
Let $\mathbf{e} = (\mathbf{q} + \mathbf{u}L)\mod 2 \in E$. Let $c = \text{supp}(\mathbf{u}):= \{0 \le i < r: \mathbf{u}[i] \ne 0\}$. Using Eq.~\eqref{eq:binsum} we can write:
\begin{align}
    \mathbf{e} &= (\mathbf{q} + \mathbf{s}_c)\mod 2\\
    &= \mathbf{q} + \mathbf{s}_c -2\mathbf{q} \mathbf{s}_c\\
    &= \mathbf{q} + \sum_{d \subset c,|d| \ge 1}(\mathbf{p}_d - 2\mathbf{q} \mathbf{p}_d)\\
    &= \mathbf{q} + \sum_{d \subset c,|d| \ge 1}\Big((-2)^{|d|-1}\prod_{i \in d}\mathbf{x}_i + (-2)^{|d|}\mathbf{q} \prod_{i \in d}\mathbf{x}_i\Big)\,.
    \end{align}
We can disregard any multiples of $N = 2^t$ so we can write:
\begin{align}
    \mathbf{e} \mod N &= \Big(\mathbf{q} + \sum_{d \subset c, 1 \le |d| \le t-1}(\mathbf{p}_d - 2\mathbf{q}\mathbf{p}_d) + \sum_{d \subset c, |d| = t}\mathbf{p}_d\Big) \mod N
\end{align}
By Lemma~\ref{lemma:polynomials} and Corollary~\ref{cor:prodx}, we know that all elements on the RHS are in $\Span_{\mathbf{Z}_N}(E_t)$ hence $\mathbf{e} \in \Span_{\mathbf{Z}_N}(E_t), \forall \mathbf{e} \in E$. Hence, all linear combinations of elements of $E$ are also in the span and the result follows.
\end{proof}
The main results of the section now follow by application of the previous proposition:

\begin{proposition}[Modified Logical Identity Algorithm]\label{prop:modified_LI}
Let $E_M$ be as defined in the algorithm of Section~\ref{sec:M_Z} i.e. the rows of $E_M$ are the binary vectors $\mathbf{e} \in E$, plus a final column of all $1$'s. Let $N = 2^t$ and  $E_t$ be the elements of $E$ of orbit distance at most $t$ from the core $E_q$. Let $E_M'$ be the matrix formed from the rows of $E_t$ plus a row of $1$'s. Then:
\begin{align}
\Span_{\mathbb{Z}_N}(E_M) = \Span_{\mathbb{Z}_N}(E_M')
\end{align}
\end{proposition}

\begin{proof}
Let $E = E_q + \langle S_X \rangle  + \langle L_X \rangle$ be the coset decomposition of $E$. Let $E_q'$ be the vectors in $E_q$ with an appended 1 (i.e. $E_q' = \{(\mathbf{q} |1) : \mathbf{q} \in E_q$. Let $G_X = \RREF(L_X \cup S_X)$ and let $L$ be $G_X$ with an appended 0 so $L = \{(\mathbf{x}|0) : \mathbf{x} \in G_X\}$. Then $E_M = E_q' + \langle L \rangle$. The result follows by applying Proposition~\ref{prop:SpanEteqSpanE}.
\end{proof}

\begin{proposition}[Modified Logical Operator Algorithm]\label{prop:modifiedLO}
Let $E_L$ be as defined as in the algorithm of Section~\ref{sec:L_Z} i.e. the rows of $E_L$ are the binary vectors $\mathbf{e} \in E$, plus a codeword index $\mathbf{i}$ which identifies which codeword the row belongs to. Let $N = 2^t$ and  $E_t$ be the elements of $E$ of orbit distance at most $t$ from the \textbf{orbit representatives} $E_m$. Let $E_L'$ be the matrix formed from the rows of $E_t$ plus the codeword index. Then:
\begin{align}
    \Span_{\mathbb{Z}_N}(E_L) = \Span_{\mathbb{Z}_N}(E_L')
\end{align}
\end{proposition}

\begin{proof}
Let $E = E_m + \langle S_X \rangle$ as in Section~\ref{sec:code_words_notation}. Let $E_m'$ be the vectors in $E_m$ with an appended $\mathbf{i}$ as in the Logical Operator Algorithm  - i.e. $E_m' = \{(\mathbf{m} |\mathbf{i}) : \mathbf{m} \in E_m\}$. Form $L$ from the rows of $S_X$ with an appended $\mathbf{0}$ of length $\dim(\mathcal{C})$ so $L = \{(\mathbf{x}|\mathbf{0}) : \mathbf{x} \in S_X\}$. Then $E_L = E_m' + \langle L \rangle$. The result follows by applying Proposition~\ref{prop:SpanEteqSpanE}.
\end{proof}

\section{Measurements in the XP Formalism - Proof of Results}\label{app:measurement}

In this appendix, we prove the results discussed in Chapter~\ref{chap:Measurements}. We first prove the results for determining the  outcome probabilities for measurement of XP operators. We then explain the stabiliser algorithm for measurement of diagonal Pauli operators of Section~\ref{sec:meas_diag_paulis}.

\subsection{Measurement Outcome Probabilities - Proof of Results}
In this section, we show how to determine the  outcome probabilities for measurement of XP operators assuming we are given the codewords in the orbit form of Eq.~\eqref{eq:orbitform}. We first show how to determine probabilities for measurement of diagonal operators, then look at non-diagonal operators. 

\begin{proposition}[Outcome Probabilities - Diagonal XP Operators]\label{prop:prob_diag}
Let $A$ be a diagonal XP operator which has $+1$ as an eigenvalue. Let $E$ be the Z-support of the codewords $|\kappa_i\rangle$. Let $E^+$ be the set of binary vectors $E^+ = \{\mathbf{e}\in E: A|\mathbf{e}\rangle =  |\mathbf{e}\rangle\}$. Let $\rho$ be an even superposition of the codewords as in Eq.~\eqref{eq:rho} so that $\rho = \frac{1}{|E|}\sum_i |\kappa_i\rangle\langle \kappa_i|$.  The probability of obtaining outcome $+1$ when measuring $A$ on $\rho$ is:
\begin{align}
\Pr(+1) = \frac{|E^+|}{|E|}
\end{align}
\end{proposition}
\begin{proof}
Let the projector onto the $+1$ eigenspace of $A$ be $A^+$ and let $|\kappa_i^+\rangle := A^+|\kappa_i\rangle$. The probability of obtaining outcome $+1$ is $\text{Tr}(A^+\rho A^+) = \frac{1}{|E|}\sum_i \langle\kappa_i^+|\kappa_i^+\rangle$.

Applying Eq.~\eqref{eq:proj_diag} with $\lambda = +1$, the action of the projector $A^+$ on a computational basis vectors $|\mathbf{e}\rangle$ is given by:
\begin{align}
A^+ |\mathbf{e}\rangle &= \begin{cases} |\mathbf{e}\rangle : \text{ if } A |\mathbf{e}\rangle = |\mathbf{e}\rangle\\ 0: \text{ if }A |\mathbf{e}\rangle \neq |\mathbf{e}\rangle\end{cases}
\end{align}
Let $E_i^+ := \{\mathbf{e} \in \ZSupp(|\kappa_i\rangle): A|\mathbf{e}\rangle = |\mathbf{e}\rangle\}$. The action of $A^+$ on the codeword $|\kappa_i\rangle$ is:
\begin{align}
|\kappa_i^+\rangle &:= A^+|\kappa_i\rangle = A^+\Big(\sum_{0 \le j < 2^r}\omega^{p_{ij}}|\mathbf{e}_{ij}\rangle\Big) = \sum_{\mathbf{e}_{ij} \in E_i^+}\omega^{p_{ij}}|\mathbf{e}_{ij}\rangle\
\\\langle\kappa_i^+|\kappa_i^+\rangle & = \sum_{\mathbf{e}_{ij} \in E_i^+}|\omega^{p_{ij}}|^2=|E_i^+|
\end{align}

The result follows by noting that $\Pr(+1) =  \frac{1}{|E|}\sum_i \langle\kappa_i^+|\kappa_i^+\rangle$ and $E^+ = \bigcup_i E_i^+$. 
\end{proof}
\begin{proposition}[Outcome Probabilities - Non-diagonal XP Operators]\label{prop:prob_nondiag}
Let $\rho$ be an even superposition of the codewords as in Eq. 
\eqref{eq:rho}. Let $A=XP_N(p|\mathbf{x}|\mathbf{z})$ be a non-diagonal XP operator with eigenvalue $+1$. Let $E$ be the Z-support of the codewords $|\kappa_i\rangle$ as calculated in Chapter~\ref{chap:codewords} and let $\mathbf{S}_X$ be the canonical non-diagonal generators as in Section~\ref{sec:canonical_generators}. Let $E^\pm$ be the set of binary vectors $E^\pm := \{\mathbf{e}\in E: A^2|\mathbf{e}\rangle =  |\mathbf{e}\rangle\}$.  

If there exists $B =XP_N(p_M|\mathbf{x}|\mathbf{z}_M) \in \mathbf{S}_X$ with the same X-component as $A$  then let $C := XP_N(p_M-p|\mathbf{0}|\mathbf{z}_M-\mathbf{z})$ and let $C|\mathbf{e}\rangle = \omega^{q_e}|\mathbf{e}\rangle$. The probability of obtaining outcome $+1$ when measuring $A$ on $\rho$ is:
\begin{align}
    \Pr(+1) = \begin{cases}
    \frac{1}{2|E|}\Big(\sum_{\mathbf{e} \in E^\pm}(1 + \cos\frac{q_e\pi}{N})\Big): \text{ if } B \text{ exists}\\
    \frac{1}{2|E|}|E^\pm|: \text{ otherwise}.
    \end{cases}
\end{align}

\end{proposition}
\begin{proof}
Let $A=XP_N(p|\mathbf{x}|\mathbf{z})$ be the operator we wish to measure, and assume we obtain outcome $+1$. Applying Eq.~\eqref{eq:proj_nondiag} with $\lambda = +1$, the action of $A^+$ on a basis vector $|\mathbf{e}\rangle$ is:
\begin{align}
A^+ |\mathbf{e}\rangle &= \begin{cases} \frac{1}{2}(|\mathbf{e}\rangle + A|\mathbf{e}\rangle): \text{ if } A^2 |\mathbf{e}\rangle = |\mathbf{e}\rangle\\ 0: \text{ if } A^2 |\mathbf{e}\rangle \neq |\mathbf{e}\rangle\end{cases}\label{eq:A+_nondiag}
\end{align}
Any elements $\mathbf{e} \in \ZSupp(|\kappa_i\rangle)$ where $A^2|\mathbf{e}\rangle \ne |\mathbf{e}\rangle$ will be eliminated, so let:
\begin{align}
E_i^\pm &= \{\mathbf{e} \in \ZSupp(|\kappa_i\rangle) : A^2|\mathbf{e}\rangle = |\mathbf{e}\rangle\}
\end{align}
By calculating $\text{Res}_{\mathbb{Z}_2}(S_X,\mathbf{x})$, we can determine if there exists some $B  \in \langle \mathbf{S}_X\rangle$ with the same X-component as $A$. There are two possible cases:

\paragraph{Case 1: $B$ does not exist:}
In this case, whenever $\mathbf{e} \in \ZSupp(|\kappa_i\rangle)$,  $\mathbf{e} \oplus \mathbf{x} \notin \ZSupp(|\kappa_i\rangle)$. Applying Eq.~\eqref{eq:A+_nondiag}: 
\begin{align}
    |\kappa_i^+\rangle &:= A^+|\kappa_i\rangle = \sum_{\mathbf{e}_{ij} \in E_i^\pm} \frac{\omega^{p_{ij}}}{2}(|\mathbf{e}_{ij}\rangle + A|\mathbf{e}_{ij}\rangle) = \frac{1}{2} \sum_{\mathbf{e}_{ij} \in E_i^\pm} (\omega^{p_{ij}}|\mathbf{e}_{ij}\rangle + \omega^{q_{ij}}|\mathbf{e}_{ij}\oplus \mathbf{x}\rangle)
    \end{align}
for some $q_{ij} \in \mathbb{Z}_{2N}$. Hence:
\begin{align}
    \langle\kappa_i^+|\kappa_i^+\rangle &= \frac{1}{4}\sum_{\mathbf{e}_{ij} \in E_i^\pm}(1+1) = \frac{|E_i^\pm|}{2}
\end{align}

\paragraph{Case 2:  $B$ exists:}
Let $B$  be the element of  $\langle \mathbf{S}_X\rangle$ with the same X-component as $A$ and let $\mathbf{e} \in E_i^\pm$ so that $A^2|\mathbf{e}\rangle = |\mathbf{e}\rangle$. Let $B|\mathbf{e}\rangle = \omega^b|\mathbf{e}\oplus \mathbf{x}\rangle$ and $A|\mathbf{e}\rangle = \omega^a|\mathbf{e}\oplus \mathbf{x}\rangle$ for some phases $a, b \in \mathbb{Z}_{2N}$. Then $\mathbf{e}, \mathbf{e}\oplus\mathbf{x} \in \ZSupp(|\kappa_i\rangle)$ with relative phase $\omega^b$. Consider the effect of the projector $A^+$ on $|\mathbf{e}\rangle + \omega^b|\mathbf{e}\oplus \mathbf{x} \rangle$:
\begin{align}
A^+(|\mathbf{e}\rangle + \omega^b|\mathbf{e}\oplus \mathbf{x} \rangle) &= A^+(|\mathbf{e}\rangle + \omega^{b-a}A|\mathbf{e} \rangle)\\&=\frac{1}{2}\Big(|\mathbf{e}\rangle + A|\mathbf{e}\rangle + \omega^{b-a}A|\mathbf{e} \rangle+\frac{b}{a}A^2|\mathbf{e} \rangle\Big)\\&=\frac{1}{2}\Big(|\mathbf{e}\rangle + A|\mathbf{e}\rangle + \omega^{b-a}(A|\mathbf{e} \rangle+|\mathbf{e} \rangle)\Big)\\&=\frac{1 + \omega^{b-a}}{2} (|\mathbf{e}\rangle + A|\mathbf{e}\rangle )
\end{align}
Now we determine $\big|\frac{1 + \omega^{b-a}}{2}\big|$:
\begin{align}
\big|\frac{1 + \omega^{b-a}}{2}\big| &= \frac{1}{4}(1 + \omega^{b-a})(1 + \omega^{a-b})\\
&= \frac{1}{4}(2 + 2\cos((b-a)\pi/N))\\
&= \frac{1}{2}(1+\cos((b-a)\pi/N))
\end{align}
Let $B= XP_N(p_B|\mathbf{x}|\mathbf{z}_B)$, let $C := XP_N(p_M - p|\mathbf{0}|\mathbf{z}_M-\mathbf{z})$ and let $q_e:= b-a$. Then $C|\mathbf{e}\rangle = \omega^{q_e}|\mathbf{e}\rangle$.
Hence:
\begin{align}
\langle\kappa_i^+|\kappa_i^+\rangle &= \sum_{\mathbf{e} \in E_i^\pm} |\langle \mathbf{e} | \kappa_i^+\rangle|^2 = \frac{1}{2}\sum_{\mathbf{e} \in E_i^\pm}(1 + \cos\frac{q_e\pi}{N})
\end{align}
The result follows by noting that $\Pr(+1)  = \frac{1}{|E|}\sum_i \langle\kappa_i^+|\kappa_i^+\rangle$ and $E^\pm = \bigcup_i E_i^\pm$. 
\end{proof}

\subsection{Analysis of Algorithm for Measuring Diagonal Paulis}
In this section, we explain why the algorithm of Section~\ref{sec:meas_diag_paulis} for measurement of diagonal Pauli operators works. We also consider implications for the complexity of measuring higher precision diagonal operators. 

Firstly, we show that simulating the measurement of a diagonal Pauli operator reduces to determining the change in the Z-support of the codewords $E$. Let $A = XP_2(0|\mathbf{0}|\mathbf{z})$  and let $|\kappa_i\rangle = \sum_j\omega^{p_{ij}} |\mathbf{e}_{ij}\rangle$ be the codewords of the code in orbit format of Eq.~\eqref{eq:orbitform} with Z-support $E = \{\mathbf{e}_{ij}\}$. Let $A^+$ be  the projector  onto the $+1$ eigenspace of $A$.  The action of $A^+$  on $|\kappa_i\rangle$ is given by Eq.~\eqref{eq:proj_diag}:
\begin{align}
A^+|\kappa_i\rangle = \sum_{j: A|\mathbf{e}_{ij}\rangle = |\mathbf{e}_{ij}\rangle} \omega^{p_{ij}} |\mathbf{e}_{ij}\rangle \,.
\end{align}
Hence, the Z-support of the codewords after measurement outcome $+1$ is $E^+ = \{\mathbf{e} \in E: A|\mathbf{e}\rangle = |\mathbf{e}\rangle\}$. As $A$ is of precision 2, the action of $A$ on $|\mathbf{e}\rangle$ is $A|\mathbf{e}\rangle = i^{2 \mathbf{e} \cdot \mathbf{z}}|\mathbf{e}\rangle= (-1)^{ \mathbf{e} \cdot \mathbf{z}}|\mathbf{e}\rangle$. Hence  $E^+ = \{\mathbf{e} \in E : \text{Par}_\mathbf{z}(\mathbf{e}) = 0\}$.  From Proposition~\ref{prop:prob_diag}, the probability of obtaining outcome $+1$ is $|E^+|/|E|$ and the phase of any $|\mathbf{e}_{ij}\rangle$ remaining in the codewords is unchanged. Similarly, the Z-support of the codewords after obtaining outcome $-1$ is $E^- = E \setminus E^+$ with probability $|E^-|/|E|$. Therefore, we can simulate measurement by determining $E^+$ and $E^-$ which in turn requires us to determine $\text{Par}_\mathbf{z}(\mathbf{e})$ for all $\mathbf{e} \in E$. 

Next, we show how the algorithm in Section~\ref{sec:meas_diag_paulis} correctly determines $E^+, E^-$. To calculate $E^+$, we need to determine $\text{Par}_\mathbf{z}(\mathbf{e})$ for all $\mathbf{e} \in E$.  For any binary vectors $\mathbf{x}, \mathbf{y}$, taking parity with respect to $\mathbf{z}$ commutes with taking the sum modulo 2. That is:
\begin{align}
\text{Par}_\mathbf{z}(\mathbf{x} \oplus \mathbf{y}) &= (\mathbf{x} + \mathbf{y} - 2\mathbf{x}\mathbf{y})\cdot \mathbf{z} \mod 2\nonumber\\
&= (\mathbf{x}\cdot \mathbf{z} + \mathbf{y}\cdot \mathbf{z} - 2\mathbf{x}\mathbf{y}\cdot \mathbf{z}) \mod 2\nonumber\\
&= \text{Par}_\mathbf{z}(\mathbf{x}) \oplus \text{Par}_\mathbf{z}(\mathbf{y})\label{eq:(x+y).z}
\end{align}
Looking at Step 1 of the algorithm, assume there exists operators $B, C \in \mathbf{S}_X \cup \mathbf{L}_X$ with X-components $\mathbf{x}, \mathbf{y}$ such that $\text{Par}_\mathbf{z}(\mathbf{x})  =\text{Par}_\mathbf{z}(\mathbf{y})= 1$. Then the X-component of $BC$ is $\mathbf{x} \oplus \mathbf{y}$ and from Eq.~\eqref{eq:(x+y).z}, we know that $\text{Par}_\mathbf{z}(\mathbf{x} \oplus \mathbf{y})= 1\oplus 1 = 0$. Hence, updating $\mathbf{S}_X \cup \mathbf{L}_X$ in Step 1 ensures that $\text{Par}_\mathbf{z}(\mathbf{x}_i) = 0$, for all X-components $\mathbf{x}_i$ of the updated set $\mathbf{S}_X \cup \mathbf{L}_X$. We remove $B$ from $\mathbf{S}_X \cup \mathbf{L}_X$ so updating $E_q$ by adding vectors of form $\mathbf{q}_l \oplus \mathbf{x}, \mathbf{q}_l \in E_q$ to $E_q$ ensures that  $E = E_q + \langle S_X\rangle + \langle L_X\rangle$ is unchanged ($S_X, L_X$ are the matrices formed from the X-components of $\mathbf{S}_X, \mathbf{L}_X$ respectively).

Now looking at Step 2 of the algorithm,  we can express the binary vectors $\mathbf{e} \in E$ as linear combinations modulo $2$ i.e. $\mathbf{e} = (\mathbf{q}_l + \mathbf{u}S_X + \mathbf{v}L_X) \mod 2$ (see Section~\ref{sec:quantum_numbers}).  After Step 1, we have ensured that $\text{Par}_\mathbf{z}(\mathbf{x}_i) = 0$, for all $\mathbf{x}_i \in S_X \cup L_X$. Linear combinations of the $\mathbf{x}_i$ also have parity $0$ and so the parity of $\mathbf{e}$ is the same as the parity of $\mathbf{q}_l$ because: 
\begin{align}
\text{Par}_\mathbf{z}(\mathbf{e}) &= \text{Par}_\mathbf{z}((\mathbf{q}_l + \mathbf{u}S_X + \mathbf{v}L_X)\mod 2)\nonumber\\
&= \text{Par}_\mathbf{z}(\mathbf{q}_l) \oplus \text{Par}_\mathbf{z}(\mathbf{u}S_X) \oplus \text{Par}_\mathbf{z}(\mathbf{v}L_X)\nonumber\\
&=\text{Par}_\mathbf{z}(\mathbf{q}_l)\label{eq:parity_commutes} \;.
\end{align}
So $\mathbf{e}$ is in $E^+$ if only if $\text{Par}_\mathbf{z}(\mathbf{q}_l)=0$ which is equivalent to $\mathbf{q}_l \in E_q^+$. Let $r = |S_X|, k = |L_X|$ then $E = |E_q|2^{r+k}, E^+ = |E_q^+|2^{r+k}$. Hence, the probability of obtaining outcome $+1$ is $\Pr(+1) = |E^+|/|E| = |E_q^+|/|E_q|$. By a similar argument, the result for the outcome $-1$ is also correct.

The underlying reason the algorithm works is that the parity function of Eq.~\eqref{eq:parity_function} commutes with addition of vectors modulo $2$. Hence, the parity of the binary vector $\mathbf{e}_{l,\mathbf{u},\mathbf{v}} \in E$ labelled in accordance with the quantum numbers of~\ref{sec:quantum_numbers} is the same as the parity of $\mathbf{q}_l \in E_q$ (see Eq.~\eqref{eq:parity_commutes}). This  relationship breaks down for higher precision XP operators (e.g. $N=4$).

\medskip


\end{document}